\documentclass[11pt,letterpaper,onecolumn]{article}

\usepackage[T1]{fontenc}
\usepackage{lmodern}
\usepackage[letterpaper,margin=1in]{geometry}
\usepackage{microtype}
\usepackage{mathtools}
\usepackage{amssymb}
\usepackage{amsthm}
\usepackage{graphicx}
\usepackage[font=normalsize,labelfont=bf,skip=6pt]{caption}
\usepackage{needspace}
\usepackage{colortbl}
\usepackage{tikz}
\usetikzlibrary{arrows.meta,positioning}
\usepackage{aliascnt}
\usepackage[hidelinks]{hyperref}
\usepackage[capitalise,nameinlink,noabbrev]{cleveref}
\usepackage{authblk}

\allowdisplaybreaks[2]

\theoremstyle{plain}
\newtheorem{theorem}{Theorem}[section]

\newaliascnt{lemma}{theorem}
\newtheorem{lemma}[lemma]{Lemma}
\aliascntresetthe{lemma}

\newaliascnt{proposition}{theorem}

\aliascntresetthe{proposition}

\newaliascnt{corollary}{theorem}
\newtheorem{corollary}[corollary]{Corollary}
\aliascntresetthe{corollary}

\newaliascnt{claim}{theorem}

\aliascntresetthe{claim}

\theoremstyle{definition}
\newaliascnt{definition}{theorem}
\newtheorem{definition}[definition]{Definition}
\aliascntresetthe{definition}

\newaliascnt{procedure}{theorem}
\newtheorem{procedure}[procedure]{Decision Procedure}
\aliascntresetthe{procedure}

\theoremstyle{remark}
\newaliascnt{remark}{theorem}
\newtheorem{remark}[remark]{Remark}
\aliascntresetthe{remark}

\crefname{theorem}{Theorem}{Theorems}
\Crefname{theorem}{Theorem}{Theorems}
\crefname{lemma}{Lemma}{Lemmas}
\Crefname{lemma}{Lemma}{Lemmas}
\crefname{proposition}{Proposition}{Propositions}
\Crefname{proposition}{Proposition}{Propositions}
\crefname{corollary}{Corollary}{Corollaries}
\Crefname{corollary}{Corollary}{Corollaries}
\crefname{claim}{Claim}{Claims}
\Crefname{claim}{Claim}{Claims}
\crefname{definition}{Definition}{Definitions}
\Crefname{definition}{Definition}{Definitions}
\crefname{procedure}{Procedure}{Procedures}
\Crefname{procedure}{Procedure}{Procedures}
\crefname{remark}{Remark}{Remarks}
\Crefname{remark}{Remark}{Remarks}

\newcommand{\AlgNums}{\mathbb C}
\newcommand{\cF}{\mathcal F}

\newcommand{\CC}{\mathsf{CC}}
\newcommand{\BO}{\mathsf{BO}}
\newcommand{\TP}{\mathsf{TP}}
\newcommand{\Mal}{\mathsf{Mal}}
\newcommand{\CSP}[1]{\mathsf{\#CSP}(#1)}
\newcommand{\CSPdeg}[2]{\mathsf{\#CSP}^{#1}(#2)}
\newcommand{\WFdeg}[1]{W_{\cF}^{#1}}
\newcommand{\SuppRel}[1]{\Phi_{#1}}
\newcommand{\RowEq}[1]{\Omega_{#1}}
\newcommand{\Pow}[2]{#1^{\langle #2\rangle}}
\newcommand{\zvec}[1]{\mathbf z_{#1}}
\newcommand{\InstMonoid}[1]{\mathfrak M_{#1}}

\newcommand{\emptytuple}{\epsilon}
\newcommand{\YES}{\textnormal{\textsc{yes}}}
\newcommand{\NO}{\textnormal{\textsc{no}}}
\DeclareMathOperator{\Pure}{Pure}
\DeclareMathOperator{\arity}{arity}
\DeclareMathOperator{\supp}{supp}
\DeclareMathOperator{\typeop}{type}
\DeclareMathOperator{\Row}{Row}
\DeclareMathOperator{\pr}{pr}

\DeclareMathOperator{\lcm}{lcm}
\DeclarePairedDelimiter{\abs}{\lvert}{\rvert}

\definecolor{PinkFoldMark}{HTML}{B74B7D}
\newcommand{\pinkfoldmark}{%
  \raisebox{-0.15em}{\resizebox{!}{0.9em}{%
    \begin{tikzpicture}[x=1pt,y=-1pt]
      \path[use as bounding box] (0,0) rectangle (24,24);
      \draw[draw=PinkFoldMark,line width=1.7pt,line join=miter,line cap=butt]
        (12,2)--(20.660,7)--(20.660,17)--(12,22)--(3.340,17)--(3.340,7)--cycle
        (12,12)--(20.660,7)
        (12,12)--(3.340,7)
        (12,12)--(12,22);
    \end{tikzpicture}%
  }}%
}

\begin{document}

\title{From Block Orthogonality to 
Decidability \\[0.15em] in Complex-Weighted Counting CSP}

\author[1]{Chenghua Liu}
\author[2]{Boning Meng}

\affil[1]{Institute of Software, Chinese Academy of Sciences, Beijing, China}
\affil[2]{University of Regensburg, Regensburg, Germany}
\affil[ ]{\texttt{liuch.russell@gmail.com}, \texttt{mengboning2013@gmail.com}}
\date{}

\hypersetup{pageanchor=false}
\begin{titlepage}
\maketitle
\thispagestyle{empty}

\begin{abstract}
\normalsize
In a landmark JACM paper recognized with the 2021 G{\"o}del Prize, Cai and
Chen established a complexity dichotomy for counting CSPs over arbitrary
finite domains with algebraic complex weights.  Their tractability criterion
consists of three conditions on the infinite family of functions obtained
from CSP instances by summing out some variables: the Block Orthogonality
and Type Partition conditions, and a Mal'tsev polymorphism condition on support and
row-equivalence relations.  They left open whether this criterion is
decidable from the finite set of constraint functions.  We resolve this
question by giving a total exact algorithm for every $q$-element domain
($q\ge 1$), with algebraic weights encoded exactly.  We also prove that
the Block Orthogonality condition implies the other two conditions, reducing the
tractability criterion to a single structural property.  We obtain similar
decidability and structural results for Lin's tractability criterion for
$\#\mathrm{CSP}^{d}$, where each variable occurs a multiple of $d$ times.

Our method first shows that every failure of Block Orthogonality can be
witnessed by a generated function of arity at most $q^2+1$.  At each such
arity, we reduce checking infinitely many generated functions to finitely
many exact algebraic tests, using tensor products and a pinned extension of
Young's all-instance isomorphism theorem.  For the structural implication, we use powers
of generated functions to control complex cancellation and construct
functions whose supports encode row equivalence.  These constructions
yield the Type Partition and Mal'tsev polymorphism conditions.
\end{abstract}
\end{titlepage}

\clearpage
\pagenumbering{arabic}
\hypersetup{pageanchor=true}
% BEGIN SHORT PAPER
\section{Introduction}
\label{s:sec:introduction}

Counting constraint satisfaction problems provide a common language for
partition functions, weighted graph homomorphisms, and many other counting
problems. Cai and Chen established a complete complexity dichotomy for counting CSPs
over arbitrary finite domains with algebraic complex weights: for every
fixed finite set $\cF$ of constraint functions, $\CSP{\cF}$ is either
polynomial-time computable or $\#\mathrm P$-hard~\cite{CaiChen2017}.
Such a classification rules out intermediate complexity throughout this
class. Ladner's theorem shows that intermediate complexity can occur
between $\mathrm P$ and $\mathrm{NP}$-completeness if
$\mathrm P\ne\mathrm{NP}$~\cite{Ladner1975}.

The dichotomy leaves a separate algorithmic question: given $\cF$, can
one decide whether its tractability criterion holds? Cai and Chen's
criterion concerns an infinite family $W_{\cF}$ of functions obtained
from instances by retaining some variables and summing over the others.
Both the arity of a generated function and the size of an instance
representing it are unbounded. The resulting family can contain
infinitely many distinct functions even at a fixed arity.
The criterion imposes the Block
Orthogonality, Type Partition, and Mal'tsev polymorphism conditions on
this entire family and its derived relations. Cai and Chen explicitly left decidability
of their criterion open~\cite[Section~9]{CaiChen2017}.

We resolve this question and show that the three conditions reduce to
one. Fix a domain $D=\{0,\ldots,q-1\}$, where $q\ge1$, and write
$\BO(\cF)$, $\TP(\cF)$, and $\Mal(\cF)$ for the three conditions,
defined in \cref{s:sec:preliminaries}. Their conjunction is denoted by
$\CC(\cF)$. The input consists of the arities and algebraic complex
values of finitely many functions on $D$.

\begin{theorem}[Uniform decidability]
\label{s:thm:main}
There is a deterministic algorithm that, given $q\ge1$ and a finite set
$\cF$ of algebraic complex-valued functions on $D$, halts and decides
whether $\CC(\cF)$ holds. Moreover,
\[
  \CC(\cF)\quad\Longleftrightarrow\quad\BO(\cF).
\]
\end{theorem}

Thus one can decide from the finite input which side of Cai and Chen's
dichotomy applies.

Our result extends to Lin's dichotomy for $\CSPdeg{d}{\cF}$, where every
variable occurs a multiple of $d$ times, counting all positions in all
constraints~\cite[Theorem~3]{Lin2021}. Let $\WFdeg{d}$ be the family
generated by instances satisfying this restriction. Write
$\BO^d(\cF)$, $\TP^d(\cF)$, and $\Mal^d(\cF)$ for the corresponding
conditions on $\WFdeg{d}$, and let $\CC^d(\cF)$ be their conjunction.
The extension requires an identity test for the restricted instance class
and constructions that preserve the degree restriction.

\begin{corollary}[Uniform decidability for $\#\mathrm{CSP}^{d}$]
\label{s:cor:degree}
For every $d\ge1$, there is a deterministic algorithm that, given
$q\ge1$ and a finite set $\cF$ of algebraic complex-valued functions
on $D$, halts and decides whether $\CC^d(\cF)$ holds. Moreover,
\[
  \CC^d(\cF)\quad\Longleftrightarrow\quad\BO^d(\cF).
\]
\end{corollary}

\paragraph{From infinitely many functions to a finite test.}
An algorithm accepting $\cF$ must certify Block Orthogonality throughout
the generated family, despite unbounded arity and instance size.
A failure is witnessed by two rows.
Identifying variable positions with the same pair of values in their
row indices preserves both rows and leaves at most $q^2$ row variables.
Purification invariance then bounds the witness arity by $q^2+1$.
This bounds the number of entries in each function, but leaves
infinitely many possible representing instances at each fixed arity.

At each remaining arity, Block Orthogonality combines support and
magnitude restrictions with a choice between dependence and blockwise
orthogonality for each row pair. All generated values lie in the number
field determined by the input. Its finite group of roots of unity lets
us express the purified condition by finite systems of polynomial
equations in the original values. Combining these alternatives gives
one finite set of identities. The task becomes
deciding whether each identity holds on every generated function
(\cref{s:sec:algebraic}).

To handle this universal quantifier, we use how instances produce
function values. Gluing two instances at their retained variables
multiplies the corresponding values. Consequently, each monomial $M$
defines a multiplicative function $\chi_M$ on instances, with its input
tuples fixed by $M$. These functions are multiplicative characters.
Distinct characters are linearly independent, so an identity holds on
all instances exactly when its coefficients sum to zero within each
group of monomials defining the same character.

This leaves one question: whether two characters agree on all instances.
Tensor products express them as pinned partition functions on finite
domains. A pinned extension of Young's isomorphism theorem makes their
agreement equivalent to the existence of a domain bijection satisfying
explicit function-value and pin conditions. Enumerating these finitely
many bijections therefore determines the character groups. Together
with the coefficient sums, this gives a finite test for the original
universal identities (\cref{s:sec:identities}).

\paragraph{From row conditions to a common operation.}
The Mal'tsev polymorphism condition requires one ternary operation
$\varphi$ satisfying $\varphi(a,b,b)=\varphi(b,b,a)=a$ and preserving
both the supports and the row-equivalence relations of all generated
functions. Block Orthogonality constrains rows without specifying such
an operation. We connect the two by making the operation's values the
variables of a constraint system.

For finitely many generated supports, the constraints require that
applying the candidate operation to three tuples in any chosen support
produce another tuple in it. Products of representing functions realize
the solution relation as a generated support. We also need projections
of this relation, but summing over omitted variables may cancel the
contributions from satisfying extensions. Taking a suitable common
positive power before summing prevents this cancellation, simultaneously
for every retained assignment, by a consequence of the
Skolem--Mahler--Lech theorem.

These projections are therefore generated supports, so Block
Orthogonality forces their overlapping nonempty row supports to be
equal. Starting with the third projection operation, which preserves every
relation and already satisfies one Mal'tsev identity, we use this
property to impose the other identity one input triple at a time.
The previously prescribed values remain extendible to a solution of
all preservation constraints. This gives a common operation for each
finite collection of supports. Since there are finitely many candidate
operations, one preserves all generated supports: otherwise finitely
many supports would exclude every candidate.

The same operation must also preserve row equivalence, which relates
indices of nonzero linearly dependent rows. We construct a generated
function on pairs of row indices whose support is exactly this relation.
Suitable positive powers and summation use the phase restrictions from
Block Orthogonality to give zero on unrelated pairs, while a power-sum
argument ensures nonzero values on every related pair. The operation
already constructed for all generated supports thus preserves row
equivalence too. Its preservation of row equivalence also gives Type
Partition (\cref{s:sec:structural}).

\paragraph{Relation to earlier effective classifications.}
Effective dichotomies were known for unweighted counting CSP
\cite{Bulatov2013,DyerRicherby2013}, for nonnegative weighted counting
CSP~\cite{CaiChenLu2016}, and for complex-weighted counting CSP on the
Boolean domain~\cite{CaiLuXia2014}. For graph homomorphisms with algebraic
complex weights, Cai, Chen, and Lu gave a polynomial-time algorithm for
deciding the tractable case~\cite[Theorem~1.2]{CaiChenLu2013}.
Our result makes the tractability criteria of Cai and
Chen~\cite{CaiChen2017} and Lin~\cite{Lin2021} decidable over arbitrary
finite domains with algebraic complex weights, including Lin's degree
restriction. Our decision procedure uses finite exhaustive
searches. The claimed algorithmic guarantee is termination on every
exactly represented input.

\paragraph{Organization.}
After the definitions in \cref{s:sec:preliminaries},
\cref{s:sec:algebraic,s:sec:identities} give the two reductions and
assemble the decision procedure. \Cref{s:sec:structural} proves the
structural implication, and \cref{s:sec:degree} treats the degree
restriction. The full paper follows this exposition. References with
prefix F point to its detailed statements and proofs.

\section{Problems and tractability conditions}
\label{s:sec:preliminaries}

Write $[n]=\{1,\ldots,n\}$ and $[0]=\varnothing$. As in the full
paper, $\mathbb C$ denotes the algebraic complex numbers, with input values
represented exactly. Arithmetic, conjugation, and equality are effective
\cite[Section~4.2]{Cohen1993}. Each $f\in\cF$ has an arity $r_f\ge1$.
An instance $I$ has a variable set $V(I)=\{x_1,\ldots,x_n\}$, with
$n=n(I)\ge1$, and a finite multiset $\mathcal C_I$ of constraints
$(f;i_1,\ldots,i_{r_f})$.
Indices may repeat within a constraint. Its instance function and
partition value are
\[
 F_I(a_1,\ldots,a_n)
 =\prod_{(f;i_1,\ldots,i_{r_f})\in\mathcal C_I}
     f(a_{i_1},\ldots,a_{i_{r_f}}),
 \qquad Z(I)=\sum_{\mathbf a\in D^n}F_I(\mathbf a).
\]
The problem $\CSP{\cF}$ asks for $Z(I)$ on input $I$.
For $1\le t\le n$, retaining the first $t$ variables gives
\[
 F_I^{[t]}(\mathbf x)=\sum_{\mathbf y\in D^{n-t}}F_I(\mathbf x,\mathbf y),
 \qquad \mathbf x\in D^t.
\]
Define
\[
 W_{\cF}
 =\{F_I^{[t]}:I\text{ is a finite }\CSP{\cF}\text{ instance},\quad
                  1\le t\le n(I)\}.
\]
Here $I$ ranges over all variable orderings, so any ordered subset of
variables may be retained. For $t=n(I)$, the sum has one term.
There is no bound on $n$ or on the number of constraints. In particular,
fixing $t$ fixes the number of function values but leaves infinitely
many possible representing instances. Products with fresh internal
variables, identification of retained variables, and summation over
some variables are again represented by instances. These constructions
allow us to apply the global conditions to auxiliary functions.

For $F:D^r\to\mathbb C$, write $\arity(F)=r$. When $r\ge2$, regard $F$ as a matrix with
rows $F(\mathbf x,*)=(F(\mathbf x,z))_{z\in D}$ indexed by
$\mathbf x\in D^{r-1}$. The function is \emph{block-rank one} if
any two nonzero rows have disjoint supports or proportional magnitude
vectors. For two nonzero rows $\mathbf u,\mathbf v$ with proportional
magnitudes, partition their common support into the level sets $B$ of
$|u_z|$. They are \emph{block-orthogonal} if
$\sum_{z\in B}u_z\overline{v_z}=0$ on every such set.
The function $F$ is block-orthogonal if it is block-rank one and every
pair of nonzero rows with proportional magnitudes is either linearly
dependent or block-orthogonal.

The global condition uses Cai and Chen's
\emph{purification}~\cite[Lemma~7]{CaiChen2017}. This transformation
replaces multiplicative generators of the nonzero values by distinct
positive primes, retains their root-of-unity factors, and leaves zeros
fixed. After scaling, the values are nonnegative integers times roots
of unity. It puts the multiplicative structure into a form on which
the row conditions above are imposed. We use this theory through its
consequences for those conditions. The precise construction is given
in the full paper at \cref{def:generating-set} and the subsequent
definition of joint purification.

\begin{definition}
\label{s:def:bo}
The \emph{Block Orthogonality condition} $\BO(\cF)$ holds if, for every
nonempty finite tuple of distinct functions in $W_{\cF}$, each component
of arity at least two is block-orthogonal after joint purification.
\end{definition}

Thus $\BO(\cF)$ requires the row conditions throughout the generated
family, including functions represented by arbitrarily large instances.

For a function $F:D^r\to\mathbb C$, its support is
$\SuppRel F=\{\mathbf a\in D^r:F(\mathbf a)\ne0\}$.
For $r\ge2$, partition its nonzero row indices into classes
$S_1,\ldots,S_s$ of linearly dependent rows. Its row-equivalence relation
is $\RowEq F=\bigcup_{j\in[s]}S_j\times S_j$.
For a prefix $\boldsymbol\alpha\in D^\ell$, where $0\le\ell\le r-1$,
define
\[
 \typeop_F(\boldsymbol\alpha)
 =\{j\in[s]:\text{some row index in }S_j
                  \text{ extends }\boldsymbol\alpha\}.
\]
For the zero function, there are no row classes, and these relations
and types are empty.

An operation $\varphi:D^3\to D$ preserves a relation if applying
$\varphi$ to the corresponding entries of any three tuples in the
relation produces another tuple in it. Such an operation is a
\emph{polymorphism} of the relation. It is a \emph{Mal'tsev operation}
if $\varphi(a,b,b)=\varphi(b,b,a)=a$ for all $a,b\in D$.
Define the relation family
\[
 \Lambda_{\cF}
 =\{\SuppRel F:F\in W_{\cF}\}
  \cup\{\RowEq F:F\in W_{\cF},\ \arity(F)\ge2\}.
\]

\begin{definition}
\label{s:def:tp-mal}
The \emph{Type Partition condition} $\TP(\cF)$ holds if, for every
$F\in W_{\cF}$ of arity at least two, any two types of prefixes of the
same length are equal or disjoint. The \emph{Mal'tsev polymorphism
condition} $\Mal(\cF)$ holds if a single Mal'tsev operation preserves
every relation in $\Lambda_{\cF}$.
\end{definition}

For supports, applying the operation to three tuples where $F$ is
nonzero must give another such tuple. For row equivalence, applying it
to the indices of three pairs of nonzero dependent rows must again
give indices of such a pair. One operation must satisfy both requirements
for all $F\in W_{\cF}$. We establish them by different constructions
in \cref{s:sec:structural}.

Cai and Chen's dichotomy makes $\CSP{\cF}$ polynomial-time computable when
$\CC(\cF)=\BO(\cF)\wedge\TP(\cF)\wedge\Mal(\cF)$ holds and
$\#\mathrm P$-hard otherwise~\cite[Theorem~1 and
Sections~3.1--3.2]{CaiChen2017}. Lin's conditions use the same definitions
with $W_{\cF}$ replaced by $\WFdeg{d}$~\cite{Lin2021}.

\section{From Block Orthogonality to polynomial identities}
\label{s:sec:algebraic}

We first eliminate the need to purify arbitrary tuples jointly, then
bound the arity of a counterexample. To obtain a finite test at each
remaining arity, we describe the purified condition by equations in the
original function values.

\subsection{Bounding the arity}

We first reduce joint purification to a property of individual functions.
The Block Orthogonality row conditions are invariant under the legal
choice of purification, including purification as part of a larger
tuple. Consequently, $\BO(\cF)$ holds exactly when each generated
function of arity at least two has a block-orthogonal singleton
purification (\cref{lem:purification-invariance,cor:singleton-purification}).
The same invariance lets us transfer a failure between a function and
the lower-arity function constructed next.

\begin{lemma}
\label{s:lem:arity}
If $\BO(\cF)$ fails, some $F\in W_{\cF}$ of arity between two and
$q^2+1$ has a singleton purification that is not block-orthogonal.
\end{lemma}

\begin{proof}
Choose two rows witnessing failure in a purified $r$-ary function,
with row indices $\mathbf x,\mathbf y\in D^{r-1}$.
Identify positions $i,j$ whenever $(x_i,y_i)=(x_j,y_j)$, keeping the
column variable separate. There are at most $q^2$ resulting row
variables. Identifying the retained variables in an instance representing
$F$ gives a generated function $H$ of arity at most $q^2+1$.
Assigning the first or second entry of each ordered pair to the new
variables recovers exactly the two original rows.

To transfer the failure through purification, purify $F$ and $H$
jointly before normalization. The same map is applied to their shared
values, so the two row equalities persist. Purification invariance
transfers the failure from the original singleton purification to this
joint purification and then to a singleton purification of $H$.\qedhere
\end{proof}

\subsection{Finite algebraic conditions}
\label{s:subsec:finite-algebraic-conditions}

All generated values lie in the computable number field
$K=\mathbb Q(f(\mathbf a),\overline{f(\mathbf a)}:
f\in\cF,\ \mathbf a\in D^{r_f})$, generated over $\mathbb Q$ by
the input values and their conjugates.
Its finite group of roots of unity $\mu(K)$ is effectively computable
\cite[Algorithm~4.9.10]{Cohen1993}.
Purification relates magnitude comparisons to equations with factors in
$\mu(K)$, preserves row dependence, and preserves ratios that are roots
of unity (\cref{lem:intrinsic-purification-tests}). Thus the purified
row conditions can be expressed using the original function values and
finitely many root-of-unity choices. We now give that construction.

Fix $k\ge2$, with row indices $\mathbf x\in D^{k-1}$ and columns
$c\in D$, and fix a total order on the row indices.
Introduce one variable $X_{\mathbf x,c}$ for each function
value. A certificate specifies how the alternatives in Block
Orthogonality are to hold. Each choice contributes polynomial equations.

\begin{definition}[Block Orthogonality certificate]
\label{s:def:certificate}
A certificate consists of the following finite choices.
\begin{enumerate}
\item Choose a possibly empty family of rectangles $R_i\times C_i$,
where the nonempty row sets $R_i\subseteq D^{k-1}$ are pairwise
disjoint, as are the nonempty column sets $C_i\subseteq D$.
Impose $X_{\mathbf x,c}=0$ outside their union. These rectangles
describe the possible groups of rows sharing a support.
\item For distinct rows $\mathbf x,\mathbf y\in R_i$ and distinct
columns $c,e\in C_i$, choose $\tau_{\mathbf x,\mathbf y;c,e}\in\mu(K)$
and impose
\[
 X_{\mathbf x,c}X_{\mathbf y,e}
 -\tau_{\mathbf x,\mathbf y;c,e}X_{\mathbf x,e}X_{\mathbf y,c}=0.
\]
These equations enforce proportional magnitudes after purification.
\item For each pair of rows $\mathbf x<\mathbf y$ in $R_i$, choose
dependence or orthogonality. For dependence, impose all ordinary
$2\times2$ minors on $C_i$, that is, the preceding equations with
$\tau_{\mathbf x,\mathbf y;c,e}=1$. For orthogonality, choose a partition
$\mathcal T_{\mathbf x,\mathbf y}$ of $C_i$. For each group
$T\in\mathcal T_{\mathbf x,\mathbf y}$, choose an anchor $c_T\in T$ and
roots $\rho_{T,c},\sigma_{T,c}\in\mu(K)$ for $c\in T$, with
$\rho_{T,c_T}=\sigma_{T,c_T}=1$, such that
$\sum_{c\in T}\rho_{T,c}\overline{\sigma_{T,c}}=0$. Impose
\[
 X_{\mathbf x,c}=\rho_{T,c}X_{\mathbf x,c_T},\qquad
 X_{\mathbf y,c}=\sigma_{T,c}X_{\mathbf y,c_T}
 \qquad(c\in T\setminus\{c_T\}).
\]
Within each group these equations fix the root-of-unity ratios, so
its purified inner product is a common factor times the chosen zero
root sum.
\end{enumerate}
\end{definition}

We evaluate these equations on $K$-valued functions.
The row and column sets, their partitions, and $\mu(K)$ are finite.
We enumerate all choices and retain those whose root sums are zero.
The resulting equations have degree one or two.
In particular, the construction uses only equations: it does not
require entries inside the chosen rectangles to be nonzero, or distinct
groups to have different magnitudes. This is essential for the
polynomial identity test that follows.

Every block-orthogonal singleton purification supplies a certificate
by its actual support rectangles, magnitude ratios, and magnitude
blocks. Conversely, the equations with root-of-unity factors force any
overlapping nonzero rows to have the same support and proportional
purified magnitudes. In orthogonal mode, a group whose anchor is zero
disappears from that row. Every surviving group lies within one
purified magnitude block and has zero inner product there. Several
groups may merge into one magnitude block, and their zero inner
products add to zero. Hence allowing zeros and merged groups introduces
no false solutions (\cref{thm:finite-bo-certificates}).

Different functions may satisfy different certificates. We replace this
disjunction by a common system of equations for the identity test.
Let $E_1,\ldots,E_s$ be the retained certificates' sets of polynomial
left-hand sides, with syntactically zero polynomials removed. If some $E_j$ is empty, put
$Q_k=\varnothing$. Otherwise put
\[
 Q_k=\left\{\prod_{j\in[s]}e_j:e_j\in E_j\text{ for all }j\in[s]\right\}.
\]
Write $P(F)$ for evaluation at $X_{\mathbf x,c}=F(\mathbf x,c)$.
A function satisfies at least one certificate exactly when every
$P\in Q_k$ satisfies $P(F)=0$: a satisfied certificate supplies a zero
factor to every product, whereas if none is satisfied, choosing a
nonzero equation value from each gives a nonzero product.
Writing $W_{\cF}^{(k)}$ for the $k$-ary members of $W_{\cF}$, we obtain
\begin{equation}
 \BO(\cF)\ \Longleftrightarrow\
 \forall\,2\le k\le q^2+1\ \forall P\in Q_k\ \forall F\in W_{\cF}^{(k)},
 \quad P(F)=0.
 \label{s:eq:bo-reduction}
\end{equation}
The remaining universal quantifier ranges over functions represented by
instances of arbitrary size. We now give a finite test for it.

\section{Testing identities on all instances}
\label{s:sec:identities}

The remaining task in \eqref{s:eq:bo-reduction} is to decide one
polynomial identity over all generated functions of a fixed arity.
We establish a general test, also allowing conjugates of function values.

\begin{theorem}[Universal identity test, \cref{thm:universal-identity}]
\label{s:thm:identity}
Given a finite set $\cF$ of algebraic complex-valued functions on $D$,
an integer $k\ge1$, and a polynomial $P$ with algebraic coefficients
in variables $X_{\mathbf a},Y_{\mathbf a}$ for $\mathbf a\in D^k$,
one can decide whether
\[
 P\bigl((F(\mathbf a))_{\mathbf a\in D^k},
        (\overline{F(\mathbf a)})_{\mathbf a\in D^k}\bigr)=0
 \quad\text{for every }F\in W_{\cF}^{(k)}.
\]
\end{theorem}

The test uses the behavior of monomials under gluing. It compares the
finitely many monomials of $P$ as functions on instances and then tests
sums of their coefficients. No Block Orthogonality assumption is needed.

\subsection{Gluing and multiplicative characters}

Fix $k\ge1$ and label the $k$ retained variables of an instance by
$1,\ldots,k$. For a labelled instance $I$, let $\mathbf z_I(\mathbf a)$
be its partition value with these variables fixed to
$\mathbf a\in D^k$. Thus $\mathbf z_I$ is exactly its generated
$k$-ary function. To multiply two such functions, take disjoint copies
of their representing instances and identify equally labelled variables.
All unlabelled variables remain separate. If $IJ$ is the glued instance,
independence of the two hidden assignments gives
\[
 \mathbf z_{IJ}(\mathbf a)=\mathbf z_I(\mathbf a)\mathbf z_J(\mathbf a).
\]
Up to label-preserving isomorphism, this operation forms a commutative
monoid $\mathfrak M_k$. Its identity has just the $k$ labelled variables
and no constraints.

Introduce formal variables $X_{\mathbf a},Y_{\mathbf a}$, evaluated
as $\mathbf z_I(\mathbf a),\overline{\mathbf z_I(\mathbf a)}$.
For a monomial, let $p_+,p_-\ge0$ count its $X$- and $Y$-factors,
respectively, with multiplicity. Write it as
\[
 M=\prod_{i=1}^{p_+}X_{\mathbf a^{(i)}}
   \prod_{j=1}^{p_-}Y_{\mathbf b^{(j)}},
 \qquad \mathbf a^{(i)},\mathbf b^{(j)}\in D^k.
\]
Evaluating $M$ at the generated values gives
\[
 \chi_M(I)=\prod_{i=1}^{p_+}\mathbf z_I(\mathbf a^{(i)})
           \prod_{j=1}^{p_-}\overline{\mathbf z_I(\mathbf b^{(j)})}.
\]
The input tuples and their multiplicities are fixed by $M$.
As $I$ varies, $\chi_M$ is a function on instances satisfying
$\chi_M(IJ)=\chi_M(I)\chi_M(J)$ and $\chi_M(1)=1$.
Such a function is a \emph{multiplicative character}.

A polynomial identity is therefore a linear relation among finitely
many characters. We first decide which monomials define the same
character. We then show that distinct characters cannot cancel in
a nontrivial linear combination.

\subsection{Tensor realization and finite comparison}

The product defining $\chi_M(I)$ involves $p_++p_-$ independent
assignments to the same instance. Collect their values at each variable
into one element of $D_M=D^{p_++p_-}$. Define the indexed family
$\mathcal H_M=(h_f:f\in\cF)$ on $D_M$, where $h_f=T_{p_+,p_-}(f)$
is given by
\[
 h_f(\mathbf u_1,\ldots,\mathbf u_{r_f})
 =\prod_{\ell=1}^{p_+}f(u_{1,\ell},\ldots,u_{r_f,\ell})
  \prod_{\ell=p_++1}^{p_++p_-}
                 \overline{f(u_{1,\ell},\ldots,u_{r_f,\ell})}.
\]
Pin label $t$ to
$\beta_M(t)=(a_t^{(1)},\ldots,a_t^{(p_+)},
b_t^{(1)},\ldots,b_t^{(p_-)})$.
Replacing every occurrence of $f$ by $h_f$ gives an instance $I^M$.
Its partition value with pins $\beta_M$, denoted by
$Z_{\beta_M}^{\mathcal H_M}(I^M)$, equals $\chi_M(I)$ by expanding
the sum over tuple-valued assignments.
For a constant monomial, $D_M=D^0$ is a singleton and the transformed
functions are constant one. If a second monomial $N$ has $r_+$
unconjugated factors and $r_-$ conjugated factors, the same construction
gives $g_f=T_{r_+,r_-}(f)$ on $D_N=D^{r_++r_-}$, the indexed family
$\mathcal H_N=(g_f:f\in\cF)$, and pins $\beta_N$
(\cref{def:tensor-realization,lem:tensor-character}).

We compare these pinned partition functions using the following
extension of Young's all-instance isomorphism
theorem~\cite[Theorem~8]{Young2025}.\footnote{The
\href{https://arxiv.org/abs/2608.14845v1}{first arXiv version of this paper}
used Young's simple-instance strengthening~\cite[Corollary~35]{Young2025}.
During the
\href{https://github.com/liuchliuch/complex-csp-decidability-lean}{Lean formalization},
we found that its premise does not imply the asserted isomorphism.
The all-instance pinned theorem below supplies the conclusion required
by our identity test. Appendix~\ref{app:all-instance-isomorphism} states
Young's results, gives an explicit counterexample to the simple-instance
corollary, and identifies the gap in its proof.}
Two domain elements are \emph{twins} if replacing one by the other in
any position of any constraint function leaves its value unchanged.
Two finite sets of functions are \emph{similar} if an arity-preserving
bijection between them is specified. The same instance is evaluated
in the two families by replacing each constraint function with its
corresponding function.

\begin{theorem}[Pinned isomorphism, \cref{thm:young-pinned}]
\label{s:thm:pinned}
Let $\mathcal H,\mathcal G$ be similar finite sets of constraint
functions on finite domains $D_{\mathcal H}$ and $D_{\mathcal G}$,
respectively, with values in a common field of characteristic zero.
Fix $k\ge1$ and pins $\beta:[k]\to D_{\mathcal H}$,
$\gamma:[k]\to D_{\mathcal G}$.
Their pinned partition values agree on every finite $k$-labelled
instance if and only if there is a bijection
$\sigma:D_{\mathcal H}\to D_{\mathcal G}$ preserving all
corresponding function values such that $\sigma(\beta(t))$ and
$\gamma(t)$ are twins in $\mathcal G$ for every $t\in[k]$.
\end{theorem}

The proof is given in \cref{thm:young-pinned}, following Young's
interpolation approach. We use the theorem to compare the tensor
realizations. Its pin condition reflects what instances can distinguish:
replacing a pinned domain element by a twin leaves every constraint
value unchanged.

Tensorization may identify distinct functions, whereas
\cref{s:thm:pinned} compares sets. At positive tensor degree, a transformed
function is zero exactly when its original is. Degree zero gives constant-one
functions. Thus each family has at most one zero function per arity.
Common positive rational scalars $\lambda_f$ can make $\lambda_fh_f$ and
$\lambda_fg_f$ distinct within each arity in their respective families:
each collision excludes at most one choice for the next scalar.
Both partition values are multiplied by the same nonzero factor
$\prod_f\lambda_f^{N_f(I)}$, where $N_f(I)$ counts occurrences of $f$.
Their equality is therefore preserved (\cref{lem:scalar-tags}).

The tagged functions correspond bijectively to $\cF$, so every labelled
instance over either tagged family comes from an original instance.
We may therefore apply \cref{s:thm:pinned}: enumerate the bijections
$D_M\to D_N$ and test corresponding function values and pin-twin
conditions exactly. These finite tests decide equality on all instances.
If the domain sizes differ, no bijection exists
(\cref{lem:character-comparison}).

\subsection{The identity test and the decision procedure}

For any polynomial $P$ in generated values and their conjugates, expand
it into monomials, compare their characters, and group equal ones.
Distinct multiplicative characters on a monoid are linearly independent
over a field. To prove this, suppose a nontrivial relation
$\sum_{i=1}^s c_i\chi_i=0$ has the fewest possible nonzero terms.
Evaluation at the identity excludes $s=1$. Choose $J$ with
$\chi_1(J)\ne\chi_s(J)$. Subtracting $\chi_s(J)$ times the relation
at $I$ from the relation at $IJ$ gives
\[
 \sum_{i=1}^{s-1}c_i\bigl(\chi_i(J)-\chi_s(J)\bigr)\chi_i(I)=0
 \qquad\text{for every }I.
\]
This is a nontrivial relation with fewer terms, a contradiction
(\cref{lem:dedekind-independence}).
Consequently, writing $P=\sum_M c_MM$, the finite tests
\begin{equation}
 \sum_{M\in E}c_M=0
 \quad\text{for every class }E\text{ of monomials defining the same character}
 \label{s:eq:coefficient-test}
\end{equation}
decide
whether $P(\mathbf z_I,\overline{\mathbf z_I})=0$ for all
$I\in\mathfrak M_k$. All comparisons and coefficient sums take place
in effectively constructed number fields. This proves
\cref{s:thm:identity}.

Equation~\eqref{s:eq:coefficient-test} involves only the finitely many coefficients
of $P$. The proof of linear independence supplies the link to all
instances: if one class sum is nonzero, the resulting linear combination
is a nonzero function on $\mathfrak M_k$, so some labelled instance
violates the identity.

Combining this test with \eqref{s:eq:bo-reduction} gives the algorithm
for $\BO(\cF)$:
\begin{enumerate}
 \item Construct $K$ and $\mu(K)$ from the input functions.
 \item For each $2\le k\le q^2+1$, construct the finite set $Q_k$.
 \item Test every $P\in Q_k$ on all $k$-labelled instances by character
       comparison and coefficient sums. Accept exactly when every test succeeds.
\end{enumerate}
Each step is finite. A rejected identity has a generated counterexample
at a tested arity. Conversely, any failure of $\BO(\cF)$ compresses
to a tested arity and violates one of its identities. This proves exact
decidability of $\BO(\cF)$ (\cref{thm:global-bo-decidable}).
When $\cF=\varnothing$, all generated functions are positive constants,
and the algorithm accepts immediately.

\section{From Block Orthogonality to the other conditions}
\label{s:sec:structural}

Assume $\BO(\cF)$. The Mal'tsev polymorphism condition requires an
operation $\varphi:D^3\to D$ that simultaneously preserves all generated
supports and all row-equivalence relations. Block Orthogonality
constrains rows, but does not specify the values of such an operation.
Even though there are finitely many ternary operations on $D$, each
candidate faces infinitely many preservation requirements.

We first construct a common operation for supports. We express
preservation as constraints on its values and use Block Orthogonality
to impose the Mal'tsev identities while keeping those constraints
satisfiable. To make Block Orthogonality applicable, we realize the
solution relation and its projections as generated supports. We then
realize row equivalence as another generated support, which makes it
preserved by the same operation. Type Partition follows from this
second preservation property.

\subsection{Realizing relations as generated supports}

We realize conjunctions and projections of generated supports as generated
supports, making Block Orthogonality applicable.
Fix $h\ge1$, functions $F_1,\ldots,F_h\in W_{\cF}$, and free and
existential variable tuples $\mathbf x,\mathbf y$ of lengths
$\ell\ge1,s\ge0$, respectively. For a finite set $\mathcal E$ of atom
indices, choose $i(e)\in[h]$ and an ordered scope $\mathbf z_e$ of length
$\arity(F_{i(e)})$ from $(\mathbf x,\mathbf y)$ for each $e\in\mathcal E$.
Define $R\subseteq D^\ell$ by
$R(\mathbf x)\Longleftrightarrow
\exists\mathbf y\,\bigwedge_{e\in\mathcal E}
\SuppRel{F_{i(e)}}(\mathbf z_e)$. Variables may repeat within and across
the scopes $\mathbf z_e$. Such a definition is an
\emph{equality-free primitive-positive definition}: it uses conjunction
and existential quantification without adding equality atoms.
Multiplication realizes conjunction. Summation realizes projection only
if the nonzero contributions do not cancel. We address this by taking
a common positive power before summing:
\begin{equation}
 H_m(\mathbf x)=\sum_{\mathbf y\in D^s}
                   \prod_{e\in\mathcal E}F_{i(e)}(\mathbf z_e)^m,
 \qquad m\ge1.
 \label{s:eq:support-realization}
\end{equation}
Taking $m$ copies of each representing instance, sharing the specified
variables and keeping its internal variables fresh, realizes $H_m$
inside $W_{\cF}$.

For each fixed assignment $\mathbf a\in R$, let $\mathcal A_{\mathbf a}$
be the nonempty multiset of nonzero products contributed by its
satisfying extensions $\mathbf b\in D^s$. Then
$H_m(\mathbf a)=\sum_{w\in\mathcal A_{\mathbf a}}w^m$.
For any finite collection of nonempty multisets of nonzero algebraic
numbers, there is a common positive $m$ making all these power sums
nonzero (\cref{lem:simultaneous-power-sums}).
Each power sum satisfies a linear recurrence with constant coefficients.
By the Skolem--Mahler--Lech theorem~\cite{Lech1953}, its zeros form finitely
many arithmetic progressions and a finite exceptional set. A progression
of period $p$ cannot have residue zero: along sufficiently large multiples
of $p$, grouping equal $p$th powers would give a Vandermonde system
forcing positive integer multiplicities to be zero. A sufficiently large
common multiple of all periods therefore avoids both the progressions
and the finite exceptional sets. Thus enumerating $m$ and testing the
finitely many sums terminates.

With this choice, \eqref{s:eq:support-realization} is nonzero exactly
on $R$, since every term is zero outside $R$. Thus these relations,
including projections retaining at least one free variable, are
generated supports (\cref{lem:equality-free-support-realization}).
Purification preserves support, so $\BO(\cF)$ implies the following
rectangularity property: when one free variable is put last, any two
nonempty row supports of such a relation are equal or disjoint
(\cref{cor:generated-singleton-rectangularity}). This is the property
used to impose the Mal'tsev identities.

\subsection{Constructing a common Mal'tsev operation}

Fix a finite family $\Delta$ of generated supports. Discard empty
relations. If none remains, any Mal'tsev operation suffices.
We treat the values of a candidate operation as an assignment to a
constraint system. For every input triple $\tau=(a,b,c)$ occurring in
the preservation tests below, introduce a variable $x_\tau$ whose
value is the operation's output at $\tau$. Let $T\subseteq D^3$ be
the set of these triples. The constraints are
\begin{equation}
 \bigwedge_{R\in\Delta}\ \bigwedge_{\mathbf u,\mathbf v,\mathbf w\in R}
 R\bigl((x_{(u_i,v_i,w_i)})_{i\in[r_R]}\bigr),
 \qquad R\subseteq D^{r_R}.
 \label{s:eq:preservation}
\end{equation}
They require one common operation to preserve all relations in $\Delta$.
Write $U_\Delta\subseteq D^T$ for their solution relation. Replacing
each relation atom in \eqref{s:eq:preservation} by its representing
function and multiplying gives a function with support $U_\Delta$.
The construction in \eqref{s:eq:support-realization} also realizes
every projection retaining at least one variable as a generated support.
Hence Block Orthogonality supplies rectangularity for precisely the
relations in which we will construct the operation's values.

Every projection operation $\pi_i(a_1,a_2,a_3)=a_i$ satisfies
\eqref{s:eq:preservation}. These supply initial solutions and values
with which to modify them. The third projection satisfies
$\pi_3(a,a,b)=b$, whereas the first satisfies $\pi_1(a,b,b)=a$.
We combine these two properties without losing extendibility to a
solution of the full constraint system.

The required values are $b$ on
$A=\{(a,a,b)\in T\}$ and $a$ on
$C=\{(a,b,b)\in T:a\ne b\}$.
The set $A$ is nonempty: repeating a tuple of a nonempty relation
three times supplies a constant triple in $T$.
Enumerate $C=\{c_1,\ldots,c_s\}$ and start with $p_0=\pi_3|_{A\cup C}$.
At step $j$, the assignment $p_j$ on $A\cup C$ extends to a
solution of \eqref{s:eq:preservation}, keeps all prescribed values on
$A$, and has the required values at $c_1,\ldots,c_j$.
The following step imposes the next value while maintaining this
invariant.

Suppose $0\le j<s$ and write $c_{j+1}=(a,b,b)$ for the next triple.
Let $K_j=A\cup\{c_i:i\in[j]\}$. Lift $p_j$ to a solution
$\widehat p_j\in U_\Delta$, then extend it arbitrarily outside $T$
to a common polymorphism $t_j:D^3\to D$.
The operation $\widehat t_j(x,y,z)=t_j(x,y,y)$ also preserves every
relation in $\Delta$. It agrees with $\pi_1$ on $K_j$ and with $p_j$
at $c_{j+1}$. For $(x,x,z)\in A$, repeating a relation tuple containing
$x$ gives $(x,x,x)\in A$, where the prescribed value forces
$t_j(x,x,x)=x$.
On the previously treated triples of $C$, it follows from the
values already imposed. Write $X=\pi_1|_{K_j}$, $Y=p_j|_{K_j}$,
and $v=p_j(c_{j+1})$.
In the projection of $U_\Delta$ to $K_j\cup\{c_{j+1}\}$, the restrictions
of $\widehat t_j$, $t_j$, and $\pi_1$ give respectively
\[
 (X,v),\qquad(Y,v),\qquad(X,a).
\]
These are three corners of a rectangle: rows $X,Y$ share column $v$,
and row $X$ contains $a$, so rectangularity gives $(Y,a)$. Lift this
tuple to $U_\Delta$ and restrict to $A\cup C$ to obtain $p_{j+1}$,
retaining the values on $K_j$ and setting $p_{j+1}(c_{j+1})=a$.
Each lift exists by membership in the projection and is found by finite
search. After all triples are treated, lift to $U_\Delta$ and extend
outside $T$ with the prescribed Mal'tsev values, obtaining a common
Mal'tsev polymorphism of $\Delta$ (\cref{lem:finite-support-maltsev}).

There are only finitely many Mal'tsev operations on $D$. If none
preserved every generated support, one could choose a support excluding
each candidate. That finite family would have no common Mal'tsev
operation, contradicting the construction above. Hence one operation
$\varphi$ preserves $\SuppRel F$ for every $F\in W_{\cF}$
(\cref{thm:common-support-maltsev}).

\subsection{Realizing row equivalence}

For $F\in W_{\cF}$ of arity $r\ge2$, we realize $\RowEq F$ as a
generated support, which the same operation $\varphi$ will preserve.
The function must be zero on unrelated row pairs and simultaneously
nonzero on all pairs of nonzero linearly dependent rows.

Although Block Orthogonality involves conjugates, we use positive powers,
products, and summation, realized by instances over the original $\cF$.
Let $e_K$ be the exponent of $\mu(K)$ and
let $L_q=\lcm(1,\ldots,q)$. Thus $e_K\ge2$ and
$\xi^{e_K}=1$ for every $\xi\in\mu(K)$.
Consider
\begin{equation}
 C_n(\mathbf x,\mathbf y)
 =\sum_{z\in D}F(\mathbf x,z)^{(e_K-1)n}F(\mathbf y,z)^n,
 \qquad n=1+tL_q,\quad t\ge0.
 \label{s:eq:detector}
\end{equation}
Copies of an instance representing $F$ realize both positive powers,
sharing the displayed variables and keeping all internal variables
fresh. Summing over $z$ gives $C_n\in W_{\cF}$.

The exponent $e_K$ removes root-of-unity ratios within purified
magnitude blocks. The progression $n=1+tL_q$ preserves the relative
phases between rows, so that their required zero sums remain zero.
Finally, choosing $t$ avoids cancellation on dependent pairs.
We verify these properties in turn.

The relative phases of overlapping rows have bounded order. For such
rows $\mathbf u,\mathbf v$, choose $z_0$ in their common support and write
$v_z=\lambda\zeta_zu_z$, where $\lambda=v_{z_0}/u_{z_0}$ and
$\zeta_{z_0}=1$. We show $\zeta_z\in\mu(K)$ and $\zeta_z^{L_q}=1$.
Every positive entrywise power of $F$ is generated. Its singleton
purification is block-orthogonal by $\BO(\cF)$, and the Purification
Lemma~\cite[Lemma~7 and Corollary~1]{CaiChen2017} gives Block Orthogonality
of the original power. Fix a nonempty magnitude block $B$ of the
original row $\mathbf u$. If the powered rows were independent for every
exponent $1,\ldots,q$, the first $|B|$ power sums of $u_z\overline{v_z}$
on $B$ would be zero, since positive powering preserves these blocks.
Newton's identities would force the product of these nonzero numbers
to be zero. Thus the rows become dependent at some exponent
$1\le h\le q$, giving $\zeta_z^h=1$. As $\zeta_z\in K$ and
$h\mid L_q$, the phase bound follows (\cref{lem:finite-row-phases}).

We now apply this phase bound to \eqref{s:eq:detector}.
Zero rows or disjoint supports contribute zero term by term.
For overlapping independent rows, now let $B$ be a magnitude block of
the row corresponding to $\mathbf u$ in a singleton purification of $F$.
Relative to an anchor $z_B\in B$, the
original ratios $u_z/u_{z_B}$ and $v_z/v_{z_B}$ belong to $\mu(K)$.
Write these ratios as $\eta_z=u_z/u_{z_B}$ and
$\theta_z=v_z/v_{z_B}$. With the relative phases anchored at $z_B$,
we have $v_z=\lambda\zeta_zu_z$ with $\lambda=v_{z_B}/u_{z_B}$,
and hence $\theta_z=\zeta_z\eta_z$.
Purified block orthogonality and preservation of root-of-unity ratios
give $\sum_{z\in B}\eta_z\overline{\theta_z}=0$.
Conjugating yields $\sum_{z\in B}\zeta_z=0$.
Since $n\equiv1\pmod{L_q}$,
the contribution of this block is
\[
 \sum_{z\in B}u_z^{(e_K-1)n}v_z^n
 =\lambda^n u_{z_B}^{e_Kn}\sum_{z\in B}\zeta_z=0.
\]
Thus every independent pair gives zero for every allowed $n$.

For a nonzero dependent pair $\mathbf v=\lambda\mathbf u$, the
value is $\lambda^n\sum_z u_z^{e_Kn}$, with the sum over the common
support. Along $n=1+tL_q$, group equal bases in
$\sum_z u_z^{e_K}(u_z^{e_KL_q})^t$.
Equal bases at $z,w$ imply $u_z/u_w\in\mu(K)$, hence equal coefficients
$u_z^{e_K}=u_w^{e_K}$. Grouped coefficients are thus positive integer
multiples of nonzero values. A root-of-unity quotient between distinct
bases would likewise force $u_z/u_w\in\mu(K)$ and equality of those
bases, a contradiction. The
Skolem--Mahler--Lech theorem and a Vandermonde argument therefore
give only finitely many zeros. Choose $t$ beyond the exceptional
values for all the finitely many dependent row pairs. Then
\[
 \SuppRel{C_n}=\RowEq F.
\]
Existence of such a common exponent ensures that enumeration of $t$
and exact tests on the finitely many row pairs terminate
(\cref{thm:omega-support-realization}). Since $C_n$ is generated,
the operation $\varphi$ already constructed for all supports also
preserves $\RowEq F$. This establishes both preservation requirements
of $\Mal(\cF)$ with the same operation.

\subsection{Type Partition and the main theorem}

For equal-length tuples $\mathbf u,\mathbf v,\mathbf w$, write
$\varphi(\mathbf u,\mathbf v,\mathbf w)$ for componentwise application
of $\varphi$.
Fix $F$ and two equal-length prefixes $\boldsymbol\alpha,
\boldsymbol\beta$ whose types intersect. Choose row indices
$\mathbf a,\mathbf b$ in one common class extending these prefixes.
For any row index $\mathbf c$ extending $\boldsymbol\alpha$,
the pairs $(\mathbf c,\mathbf c)$, $(\mathbf a,\mathbf a)$, and
$(\mathbf b,\mathbf a)$ lie in $\RowEq F$ whenever $\mathbf c$
indexes a nonzero row. Preservation by $\varphi$ gives
\[
 \bigl(\varphi(\mathbf c,\mathbf a,\mathbf b),
       \varphi(\mathbf c,\mathbf a,\mathbf a)\bigr)
 =\bigl(\varphi(\mathbf c,\mathbf a,\mathbf b),\mathbf c\bigr)
 \in\RowEq F.
\]
The first row index extends $\boldsymbol\beta$, because
$\varphi(\alpha_i,\alpha_i,\beta_i)=\beta_i$ at each prefix
position, and it belongs to the same class as $\mathbf c$.
Every class reached from $\boldsymbol\alpha$ is therefore reached
from $\boldsymbol\beta$. Interchanging the prefixes proves that
their types are equal (\cref{lem:omega-implies-type}).

We have proved $\BO(\cF)\Rightarrow\TP(\cF)\wedge\Mal(\cF)$.
The converse implication $\CC(\cF)\Rightarrow\BO(\cF)$ holds by
definition. For the empty input family, generated functions are positive
constants, so the same conclusion holds directly.
The algorithm of \cref{s:sec:identities} consequently decides
$\CC(\cF)$, proving \cref{s:thm:main}.

\section{The degree restriction}
\label{s:sec:degree}

For \cref{s:cor:degree}, restrict to instances whose variable degrees are
divisible by $d$. Gluing adds degrees of identified variables, so these
form a submonoid $\mathfrak M_{k,d}\subseteq\mathfrak M_k$.
We reduce character equality here to finite comparisons on $\mathfrak M_k$.

Let $C_d=\mathbb Z/d\mathbb Z$ and choose a primitive $d$th root
of unity $\omega_d$, taking $\omega_1=1$. The identity
\[
 \sum_{s\in C_d}\omega_d^{sm}
 =\begin{cases}d,&d\mid m,\\0,&d\nmid m\end{cases}
\]
tests one variable's degree. Augment the tensor functions for $M$ to
the domain $D_M\times C_d$ by defining, for $a_i\in D_M$ and $s_i\in C_d$,
\[
 h_f^{(d)}((a_1,s_1),\ldots,(a_{r_f},s_{r_f}))
 =h_f(a_1,\ldots,a_{r_f})\omega_d^{s_1+\cdots+s_{r_f}}.
\]
Each variable occurrence contributes its auxiliary root once.
For $\mathbf s\in C_d^k$, pin label $t$ to $(\beta_M(t),s_t)$.
The resulting partition function $\rho^M_{\mathbf s}$ on original
instances is a multiplicative character. It sums auxiliary values at
hidden variables, and summing over $\mathbf s$ covers the labels, giving
\[
 \chi_M(I)\prod_{v\in V(I)}
       \left(\sum_{s\in C_d}\omega_d^{s\deg_I(v)}\right),
\]
where $\deg_I(v)$ counts all occurrences of $v$, including repeated
positions within a constraint. Consequently,
\begin{equation}
 R_M(I):=\sum_{\mathbf s\in C_d^k}\rho^M_{\mathbf s}(I)
 =\begin{cases}
   d^{|V(I)|}\chi_M(I),&I\in\mathfrak M_{k,d},\\
   0,&I\notin\mathfrak M_{k,d}.
  \end{cases}
 \label{s:eq:degree-filter}
\end{equation}
Construct $R_N$ similarly. The common nonzero factor in
\eqref{s:eq:degree-filter} gives
$\chi_M=\chi_N$ on $\mathfrak M_{k,d}$ exactly when $R_M=R_N$
on $\mathfrak M_k$. The added factors are nonzero, so each
augmented family still has at most one zero function per arity.
\Cref{lem:character-comparison} uses common scalar tags to compare the
$2d^k$ terms. Character independence makes $R_M=R_N$ equivalent to
equal multiplicities in every character class.
Grouping monomials by this restricted equality and testing their
coefficient sums, using character independence on $\mathfrak M_{k,d}$,
gives the identity test on
$\WFdeg{d}$ (\cref{lem:degree-filtered-character-comparison,thm:degree-multiple-universal-identity}).

The purification-invariance, singleton, and certificate arguments apply
unchanged to $\WFdeg{d}$. All instance constructions preserve degree
divisibility: identifying retained variables adds degrees, fresh internal
copies retain theirs, and summation preserves occurrences. Thus the same
arity bound, polynomials $Q_k$, and restricted identity test decide
$\BO^d(\cF)$. The support and detector constructions
\eqref{s:eq:support-realization} and \eqref{s:eq:detector} stay in
$\WFdeg{d}$, so the structural proof gives
$\BO^d(\cF)\Rightarrow\TP^d(\cF)\wedge\Mal^d(\cF)$.
This proves \cref{s:cor:degree}, with the full argument in
\cref{subsec:degree-tractability}.

\section*{AI methodology statement}
\addcontentsline{toc}{section}{AI methodology statement}

We used OpenAI's GPT-5.6 Sol to assist with mathematical exploration and
proof development throughout this paper. Its most substantial role was
in developing the approach presented in
Sections~\ref{s:subsec:finite-algebraic-conditions} and~\ref{s:sec:identities},
with full details given in
Sections~\ref{sec:bo-certificates}--\ref{sec:universal-identities}.
This approach expresses Block Orthogonality through
polynomial identities and decides their validity over the infinite
generated family using multiplicative characters and finite isomorphism
tests.

We used OpenAI's GPT-6 Astra to improve the writing and presentation and
to assist with the \href{https://github.com/liuchliuch/complex-csp-decidability-lean}{\pinkfoldmark\nobreak\hspace{0.2em}Lean formalization}. During formalization, we identified
a gap in an early version of the paper: the argument relied on a
simple-instance corollary of Young's isomorphism theorem whose conclusion
does not follow from its stated hypotheses. We repaired the argument by
proving the all-instance pinned version needed here and completed the
Lean formalization of the proof.
Appendix~\ref{app:all-instance-isomorphism} explains the gap and gives an
explicit counterexample to Young's corollary. The authors take full
responsibility for the correctness and originality of the paper.

\section*{Acknowledgments}
\addcontentsline{toc}{section}{Acknowledgments}

We thank Jin-Yi Cai and Xi Chen for their helpful comments and suggestions
on improving the writing and presentation of this paper.

\label{s:end:main}

% END SHORT PAPER

\clearpage
\phantomsection
\addcontentsline{toc}{section}{References}
% BEGIN EMBEDDED REFERENCES

% END EMBEDDED REFERENCES

\clearpage
\setcounter{section}{0}
\setcounter{equation}{0}
\setcounter{figure}{0}
\setcounter{table}{0}
\setcounter{footnote}{0}
\setcounter{page}{1}
\renewcommand{\thepage}{F-\arabic{page}}
\renewcommand{\thesection}{F.\arabic{section}}
\renewcommand{\theHsection}{full.\arabic{section}}
\renewcommand{\theequation}{F.\arabic{equation}}
\renewcommand{\theHequation}{full.\arabic{equation}}
\renewcommand{\thefigure}{F.\arabic{figure}}
\renewcommand{\theHfigure}{full.\arabic{figure}}
\renewcommand{\thetable}{F.\arabic{table}}
\renewcommand{\theHtable}{full.\arabic{table}}
\section*{Full paper}
% BEGIN FULL PAPER
\section{Introduction}
\label{sec:introduction}

Counting constraint satisfaction problems provide a common language for
partition functions, weighted graph homomorphisms, and many other counting
problems.  In a landmark JACM paper recognized with the 2021 G{\"o}del Prize,
Cai and Chen established a complete complexity dichotomy for counting CSPs
over arbitrary finite domains with algebraic complex weights: for every
fixed finite set $\cF$ of constraint functions, $\CSP{\cF}$ is either solvable
in polynomial time or $\#\mathrm P$-hard~\cite{CaiChen2017}.  This rules out
intermediate complexity throughout this broad class of counting problems.
Such a dichotomy is far from automatic: Ladner's theorem~\cite{Ladner1975}
shows that, if $\mathrm P\ne\mathrm{NP}$, there are problems in $\mathrm{NP}$
that are neither polynomial-time solvable nor $\mathrm{NP}$-complete.

Cai and Chen's tractability criterion consists of three conditions on an
infinite family of generated functions.  They explicitly left open whether
one can decide, from the finite set of constraint functions alone, whether
all three conditions hold~\cite[Section~9]{CaiChen2017}.  We resolve this
question by giving a total exact decision procedure, making their
classification effective.

To state our result, write $[n]=\{1,\ldots,n\}$ for $n\ge0$, with
$[0]=\varnothing$.  Fix the domain $D=\{0,\ldots,q-1\}$, where $q\ge1$,
and let $\cF$ be a finite set of algebraic complex-valued constraint functions
on $D$.  For a finite $\CSP{\cF}$ instance $I$, choose an ordering of its
$n=\abs{V(I)}\ge1$ variables.  Let $F_I:D^n\to\mathbb C$ be the function
obtained by multiplying its constraint functions applied to their specified
variables. A variable may occur more than once within a constraint.
For $1\le t\le n$, define
\[
  F_I^{[t]}(\mathbf x)
  =\sum_{\mathbf y\in D^{\,n-t}}F_I(\mathbf x,\mathbf y),
  \qquad \mathbf x\in D^t.
\]
Here $\mathbf x$ assigns values to the first $t$ variables, and the sum ranges
over all assignments to the remaining $n-t$ variables.  In particular,
$F_I^{[n]}=F_I$.  Following Cai and Chen~\cite[Section~3]{CaiChen2017}, define
\begin{equation}
  W_{\cF}
  =\left\{F_I^{[t]}:
    \begin{array}{l}
      I\text{ is a finite }\CSP{\cF}\text{ instance}\\
      \text{with an ordering of its variables},\quad
      1\le t\le \abs{V(I)}
    \end{array}
    \right\}.
  \label{eq:intro-generated-family}
\end{equation}
Both the size of the representing instances and the arity of the resulting
functions are unbounded.

Cai and Chen's tractability criterion imposes the Block Orthogonality,
Type Partition, and Mal'tsev polymorphism conditions on the entire family
$W_{\cF}$ and the relations derived from it.  We denote them by
$\BO(\cF)$, $\TP(\cF)$, and $\Mal(\cF)$, respectively. Their precise
definitions are given in \cref{def:global-bo,def:type-maltsev}.  Their conjunction
\[
  \CC(\cF)=\BO(\cF)\wedge\TP(\cF)\wedge\Mal(\cF)
\]
is the tractability criterion in their dichotomy: when $\CC(\cF)$ holds,
$\CSP{\cF}$ is polynomial-time computable. Otherwise, it is
$\#\mathrm P$-hard~\cite[Theorem~1 and Sections~3.1--3.2]{CaiChen2017}.
In addition to deciding this criterion, we prove that the Block Orthogonality
condition alone implies the other two conditions.

\begin{theorem}[Uniform decidability]
\label{thm:main}
There is a deterministic algorithm that, given $q\ge1$ and a finite set
$\cF$ of algebraic complex-valued functions on $D$, halts and decides
whether $\CC(\cF)$ holds.  Moreover,
\begin{equation}
  \CC(\cF)\quad\Longleftrightarrow\quad\BO(\cF).
  \label{eq:main-equivalence}
\end{equation}
\end{theorem}

For $d\ge1$, let $\CSPdeg{d}{\cF}$ denote the restriction of $\CSP{\cF}$
to instances in which the total number of occurrences of each variable in
the ordered tuples of variables in the constraints is divisible by $d$.
Define $\WFdeg{d}$ by
\eqref{eq:intro-generated-family}, restricting $I$ to $\CSPdeg{d}{\cF}$
instances.  Lin proved a dichotomy theorem for this class, with a tractability
criterion given by the Block Orthogonality, Type Partition, and Mal'tsev
polymorphism conditions on $\WFdeg{d}$ and its derived
relations~\cite[Theorem~3]{Lin2021}.  Denote these conditions by
$\BO^d(\cF)$, $\TP^d(\cF)$, and $\Mal^d(\cF)$, respectively, and write
\[
  \CC^d(\cF)
  =\BO^d(\cF)\wedge\TP^d(\cF)\wedge\Mal^d(\cF).
\]

We obtain the following result by developing an identity test for
$\WFdeg{d}$ based on a root-of-unity filter
(\cref{subsec:degree-multiple-identities}) and proving that the support and
row-equivalence constructions preserve the degree restriction
(\cref{subsec:degree-tractability}).

\begin{corollary}[Uniform decidability for $\#\mathrm{CSP}^{d}$]
\label{cor:degree-multiple-main}
For every $d\ge1$, there is a deterministic algorithm that, given $q\ge1$
and a finite set $\cF$ of algebraic complex-valued functions on $D$,
halts and decides whether $\CC^d(\cF)$ holds.  Moreover,
\[
  \CC^d(\cF)\quad\Longleftrightarrow\quad\BO^d(\cF).
\]
\end{corollary}

\subsection{Proof overview}
\label{subsec:proof-sketch}

The proof has two parts: deciding $\BO(\cF)$ and showing that it implies
$\TP(\cF)$ and $\Mal(\cF)$.  The first part must handle generated functions
of unbounded arity, represented by instances of unbounded size.  The second
must overcome cancellation when summing over variables, so that we can
construct functions whose supports express the required relations.

\paragraph{From arbitrary arity to finitely many polynomial conditions.}
We first show that the Block Orthogonality condition $\BO(\cF)$ can be
checked by purifying each generated function separately.  For a $k$-ary
function $F$, write $F(\mathbf x,z)$, where $\mathbf x\in D^{k-1}$ and
$z\in D$.  We view $F$ as a matrix whose rows are indexed by $\mathbf x$
and whose columns are indexed by $z$.  A failure of Block Orthogonality is
witnessed by two rows, indexed by tuples $\mathbf x$ and $\mathbf y$.
Grouping their positions according to $(x_i,y_i)\in D^2$ and identifying
the corresponding variables leaves at most $q^2$ row variables.  Keeping
the column variable separate preserves both rows, so every failure has a
witness of arity at most $q^2+1$ (\cref{sec:purification-compression}).

At each fixed arity, all generated function values lie in the number field
determined by the input functions.  Using the finitely many roots of unity
in this field, we express Block Orthogonality after purification as finitely
many alternatives, each given by polynomial equations in the original
function values.  Products of these equations then give a finite set of
polynomial identities.  Block Orthogonality holds at this arity exactly
when every one of these identities holds on every generated function of
this arity (\cref{sec:bo-certificates}).

\paragraph{Testing identities through multiplicative characters.}
For a fixed arity, label the variables that remain unsummed.  The values of
a generated function are then the partition values of its representing
instance with these labelled variables fixed.  Gluing two instances at
equally labelled variables multiplies the corresponding values.

Each monomial $M$ in the polynomials constructed above is a product of
powers of function values and their complex conjugates, with the input
tuples and exponents fixed by $M$.  Evaluating this expression for each
labelled instance $I$ defines a function $\chi_M$ on labelled instances.
The gluing property gives
\[
  \chi_M(IJ)=\chi_M(I)\chi_M(J),
\]
where $IJ$ denotes the glued instance.  These functions are called
multiplicative characters.  Distinct characters are linearly independent,
so a polynomial identity holds on every generated function exactly when
its coefficients sum to zero within each group of monomials defining the
same character.

It remains to decide which characters are equal.  Tensor products represent
the monomials as pinned partition functions over finite domains.  After
distinguishing the resulting functions by scalar factors, we apply a pinned
extension of Young's all-instance isomorphism
theorem~\cite[Theorem~8]{Young2025}, proved here as
\cref{thm:young-pinned}.\footnote{The
\href{https://arxiv.org/abs/2608.14845v1}{first arXiv version of this paper}
relied on Young's simple-instance strengthening~\cite[Corollary~35]{Young2025}.
During the \href{https://github.com/liuchliuch/complex-csp-decidability-lean}{Lean formalization}, we found that its simple-instance premise
does not imply the asserted isomorphism.  \Cref{thm:young-pinned} establishes
the all-instance pinned statement needed here: our character comparisons
require equality on all finite labelled instances.
Appendix~\ref{app:all-instance-isomorphism} gives an explicit counterexample
to the simple-instance statement and identifies the gap in its proof.}
The theorem replaces agreement on infinitely many labelled instances by
the existence of a single domain bijection preserving the corresponding
constraint functions and satisfying the pin conditions.  Since the tensor
constructions have finite domains and finitely many constraint functions,
we can enumerate all candidate bijections and check each one using finitely
many exact comparisons.  We obtain an exact identity test and, together
with the arity bound, a terminating decision procedure for $\BO(\cF)$
(\cref{sec:universal-identities,sec:bo-algorithm}).

\paragraph{From Block Orthogonality to the other tractability conditions.}
The Mal'tsev polymorphism condition requires a single ternary operation
$\varphi:D^3\to D$ satisfying
\[
  \varphi(a,b,b)=\varphi(b,b,a)=a
  \qquad\text{for all }a,b\in D,
\]
and preserving every support and row-equivalence relation of the generated
functions.  The support of a generated function is defined as the set of
input tuples at which it takes a nonzero value.  An operation preserves a
relation if applying it to the corresponding entries of any three tuples
in the relation produces another tuple in that relation.

Assume $\BO(\cF)$.  We first construct a common Mal'tsev operation for the
generated supports.  For a finite collection of these supports, we introduce
variables $x_{a,b,c}$ representing the outputs $\varphi(a,b,c)$ of a candidate
ternary operation.  For each relation $R\subseteq D^r$ in the collection
and every three tuples $\mathbf u,\mathbf v,\mathbf w\in R$, we impose
the constraint
\[
  \bigl(x_{u_i,v_i,w_i}\bigr)_{i\in[r]}\in R.
\]

These constraints require one common operation to preserve all the chosen
relations.  To realize the system, choose a generated function $F_R$ with
support $R$ and replace each constraint by an instance representing $F_R$,
sharing the specified variables and keeping internal variables separate.
The resulting function is a product of the corresponding $F_R$-values,
so it is nonzero exactly when all constraints are satisfied.  We also need
projections of this solution relation.  Summing over the omitted variables
may cancel contributions from satisfying extensions, so we first use
repeated copies of the instances to raise each contribution to a common
positive power.  A consequence of the Skolem--Mahler--Lech
theorem~\cite{Lech1953} guarantees an exponent for which the sum is nonzero
whenever a satisfying extension exists, simultaneously for all retained
assignments.  Thus both the solution relation and its required projections
are generated supports.

Since purification preserves support, Block Orthogonality forces any two
overlapping nonempty rows of these projected relations to have the same
support.  We use this property to construct the desired operation, starting
with the third projection $(a,b,c)\mapsto c$, which preserves every relation
and already satisfies $\varphi(b,b,a)=a$.  The row-support property allows
us to impose the remaining requirements $\varphi(a,b,b)=a$ one input triple
at a time, retaining the previously prescribed outputs and their
extendibility to an operation preserving all the chosen relations.  This
gives a common Mal'tsev operation for each finite collection of generated
supports.  Since there are only finitely many Mal'tsev operations on $D$,
this yields one preserving every generated support.  Otherwise, finitely
many supports would rule out all candidates (\cref{sec:generated-supports}).

To complete the Mal'tsev polymorphism condition, we must show that the same
operation also preserves every row-equivalence relation, which relates
indices of nonzero, linearly dependent rows.  We do this by realizing each
such relation as a generated support.  For a generated function $F$, we
multiply suitable positive powers of $F(\mathbf x,z)$ and $F(\mathbf y,z)$,
then sum over the shared variable $z$.  Copies of an instance representing
$F$, with separate internal variables, realize this construction within
$W_{\cF}$.  The phase restrictions imposed by Block Orthogonality allow us
to choose exponents for which the resulting function is zero outside the
row-equivalence relation.  A power-sum argument then selects a choice
making it nonzero on every pair of nonzero, linearly dependent rows.  Its
support is therefore exactly the required relation, which is preserved by
the common Mal'tsev operation from \cref{sec:generated-supports}
(\cref{sec:omega-realization}).

Preservation of row equivalence also yields the Type Partition condition.
For each prefix of a fixed length, its type records the row-equivalence
classes reached by completing that prefix.  The Mal'tsev identities and
preservation of row equivalence show that, whenever two types overlap,
every class reached from one prefix can also be reached from the other.
Thus the types are either equal or disjoint, establishing the Type Partition
condition and completing the implication from the Block Orthogonality
condition to the other two tractability conditions (\cref{sec:type-collapse}).

\paragraph{The degree restriction.}
For $\#\mathrm{CSP}^{d}$, we must compare characters on instances in which
every variable's occurrence degree is divisible by $d$.  We reduce this
restricted comparison to the unrestricted identity test using a root-of-unity
filter.  Assign an auxiliary $d$th root of unity to each variable and include
it once for each occurrence of that variable.  Summing over these auxiliary
choices gives zero whenever any variable violates the degree condition,
and a nonzero factor otherwise.  Applying the filter to two characters
therefore makes both zero outside the permitted class and multiplies them
by the same nonzero factor within it.  The filtered expressions are finite
sums of characters on unrestricted instances, so their equality can be
decided by the character comparison and linear-independence arguments above
(\cref{subsec:degree-multiple-identities}).

We must also ensure that the arity compression and the constructions used
to establish the Type Partition and Mal'tsev polymorphism conditions
respect the degree restriction.  Copying instances with separate internal
variables preserves their occurrence degrees, identifying retained variables
adds degrees divisible by $d$, and summing out a variable does not change
its occurrence degree.  Thus the compressed functions and the functions
realizing supports and row equivalence remain in $\WFdeg{d}$.  Together
with the filtered identity test, this gives uniform decidability of
$\BO^d(\cF)$ and shows that it implies the other two tractability conditions.
Hence Lin's criterion $\CC^d(\cF)$ is uniformly decidable, as stated in
\cref{cor:degree-multiple-main} (\cref{subsec:degree-tractability}).

\subsection{Related work}

\paragraph{Counting graph homomorphisms ($\#\mathrm{GH}$).}
For a symmetric $q\times q$ matrix $A$ indexed by $D$, on an
input graph $G=(V,E)$, the problem $\#\mathrm{GH}(A)$ evaluates
\[
  Z_A(G)=\sum_{\sigma:V(G)\to D}
    \prod_{\{u,v\}\in E(G)}A_{\sigma(u),\sigma(v)}.
\]
The problem $\#\mathrm{GH}(A)$ is the special case of $\#\mathrm{CSP}$ with
one symmetric binary constraint and is therefore covered by the
complex-weighted $\#\mathrm{CSP}$ dichotomy theorem of Cai and
Chen~\cite{CaiChen2017}.  Dyer and Greenhill classified $\{0,1\}$-valued
matrices, Bulatov and Grohe
extended the dichotomy to nonnegative matrices, Goldberg, Grohe, Jerrum, and
Thurley handled arbitrary real matrices, and Cai, Chen, and Lu completed the
classification for algebraic complex matrices
\cite{DyerGreenhill2000,DyerGreenhill2004,BulatovGrohe2005,
GoldbergGroheJerrumThurley2010,CaiChenLu2013}.  For algebraic-complex
matrices, the dichotomy theorem of Cai, Chen, and Lu also provides a
polynomial-time decision algorithm for its tractable
case~\cite[Theorem~1.2]{CaiChenLu2013}.

Planar input restrictions admit further tractable cases through
holographic algorithms and matchgates.  Relevant classifications include
real symmetric matrices on three elements~\cite{CaiMaran2023} and
nonnegative full-rank matrices on four elements~\cite{CaiMaran2024}.
Cai, Maran, and Young connect more general planar classifications to
vertex separation and quantum automorphism groups~\cite{CaiMaranYoung2026}.

\paragraph{Counting CSP.}
For general constraint languages, Dyer, Goldberg, and Jerrum classified
nonnegative weighted Boolean $\#\mathrm{CSP}$
\cite{DyerGoldbergJerrum2009}. Bulatov proved the finite-domain unweighted
dichotomy, and Dyer and Richerby subsequently gave an elementary effective
criterion
\cite{Bulatov2013,DyerRicherby2013}.  For algebraic weights, Cai, Chen, and
Lu obtained an effective dichotomy in the nonnegative setting. Cai, Lu, and
Xia classified complex-weighted Boolean languages, and Cai and Chen
established the dichotomy for every finite set of functions over arbitrary
finite domains with
algebraic complex weights
\cite{CaiChenLu2016,CaiLuXia2014,CaiChen2017}.  Lin proved a dichotomy
theorem for the degree-divisibility restriction $\#\mathrm{CSP}^{d}$
whose tractable side is characterized by the corresponding three infinitary
conditions on $\WFdeg{d}$~\cite[Theorem~3]{Lin2021}.  We establish
uniform decidability of the tractability criteria in both dichotomies.

\paragraph{Holant.}
Valiant's holographic algorithms introduced basis transformations for
tensor-network partition functions~\cite{Valiant2008}, and Cai, Lu, and Xia
formulated Holant as a general counting framework containing $\#\mathrm{CSP}$
and graph homomorphisms~\cite{CaiLuXia2009}.  On the Boolean domain, a
sequence of dichotomies progressed from symmetric and auxiliary-signature
settings to arbitrary nonnegative signatures, complex-valued
$\operatorname{Holant}^{c}$, and finally
arbitrary real-valued signatures
\cite{CaiGuoWilliams2016,LinWang2018,
Backens2021,Backens2025,ShaoCai2020}.  For complex-valued Boolean Holant problems, recent work establishes a
dichotomy in the presence of a nontrivial odd-arity
signature~\cite{MengWangXiaZheng2025,GuanShaoShi2026}.

Higher-domain results include single-signature and edge-coloring
classifications~\cite{CaiLuXia2013,CaiGuoWilliams2014,LiuFanCai2025},
and a dichotomy with all unary signatures available for arbitrary sets of
symmetric real-valued signatures on three elements~\cite{CaiIhm2025}.

\paragraph{Decision CSP.}
For the decision problem, Schaefer established the Boolean-domain dichotomy
\cite{Schaefer1978}.  Feder and Vardi placed finite-domain CSP in a broad
logical framework and formulated the dichotomy conjecture that every fixed
finite constraint language yields either a polynomial-time problem or an
$\mathsf{NP}$-complete one~\cite{FederVardi1998}.  The universal-algebraic
approach relates CSP complexity to polymorphisms
\cite{BulatovJeavonsKrokhin2005}. Mal'tsev polymorphisms also underlie
effective counting-CSP algorithms~\cite{DyerRicherby2013}.  Bulatov and Zhuk
independently proved the full finite-domain dichotomy in
2017~\cite{Bulatov2017,Zhuk2017}. See also Zhuk's
full journal account~\cite{Zhuk2020}.

\subsection{Organization}

\Cref{sec:preliminaries} defines the problems and the three tractability
conditions.  \Cref{sec:purification-compression} reduces the Block
Orthogonality condition to singleton purification at bounded arity and
establishes the purification tests used in
\cref{sec:bo-certificates} to construct polynomial conditions.
\Cref{sec:universal-identities} gives a universal identity test, and
\cref{sec:bo-algorithm} combines these results into the Block Orthogonality
decision procedure.  \Crefrange{sec:generated-supports}{sec:type-collapse}
construct a common Mal'tsev polymorphism, establish the Type Partition
condition, and complete the proof of \cref{thm:main}.
\Cref{sec:degree-extension} extends the argument to $\#\mathrm{CSP}^{d}$.
\Cref{sec:conclusion} summarizes the results.
\Cref{app:tractability-examples} illustrates the tractability conditions,
variable identification, multiplicative characters, and tensor realization.
\Cref{app:all-instance-isomorphism} explains the failure of Young's
simple-instance strengthening.

\section{Preliminaries and the Three Tractability Conditions}
\label{sec:preliminaries}

This section formalizes the three tractability conditions in the
complex-weighted $\#\mathrm{CSP}$ dichotomy theorem of Cai and
Chen~\cite[Sections~2.8--2.9 and~3.1]{CaiChen2017}. Examples are collected in
\cref{app:tractability-examples}.

\subsection{Algebraic and group notation}

For an integer $n\ge0$, write $[n]=\{1,\ldots,n\}$, with
$[0]=\varnothing$.  We use $[n]$ to index variables and finite lists.
The CSP domain is $D=\{0,\ldots,q-1\}$, where $q\ge1$, and
$D^0=\{\emptytuple\}$, with $\emptytuple$ the empty tuple.
Any nonempty finite domain can be put in this form by a bijection.
Transporting constraint functions along the bijection preserves partition
values and all three tractability conditions.
The integer $d\ge1$ denotes Lin's occurrence-divisibility parameter.
Throughout, $\mathbb C$ denotes the field of algebraic complex numbers.
Input values have standard exact representations specifying the chosen
complex root.  Arithmetic, conjugation, and equality tests are
effective~\cite[Section~4.2]{Cohen1993}.

The following group notation describes the exponent maps in
\cref{lem:purification-invariance} and the root-of-unity tests in
\cref{lem:intrinsic-purification-tests}.  If $L$ is a field,
then $L^\times=L\setminus\{0\}$ is its multiplicative group.  For a possibly
infinite subset $S\subseteq L^\times$, the subgroup generated by $S$ is
\[
  \langle S\rangle
  =\biggl\{\prod_{a\in T}a^{n_a}:
      T\subseteq S\text{ is finite},\
      n_a\in\mathbb Z\text{ for every }a\in T\biggr\}.
\]
We write $H\le\Gamma$ when $H$ is a subgroup of the abelian group
$\Gamma$.  A homomorphism from a multiplicative group to an additive
group satisfies $e(xy)=e(x)+e(y)$ and has kernel
$\ker e=\{x:e(x)=0\}$.  For $\mathbf v\in\mathbb Z^m$ and $N\ge1$,
$N\mathbf v=0$ implies $\mathbf v=0$.

Let $\mu_\infty\le\AlgNums^\times$ be the group of all roots of unity,
namely the elements $\xi$ satisfying $\xi^N=1$ for some $N\ge1$.
For a number field $K$, the group $\mu(K)=K^\times\cap\mu_\infty$ is
finite and cyclic.  For $\Gamma\le\AlgNums^\times$, a homomorphism
$\kappa:\Gamma\to\mathbb C^\times$ fixes every root of unity in $\Gamma$ if
\begin{equation}
  \kappa(\xi)=\xi
  \qquad\text{for every }\xi\in\Gamma\cap\mu_\infty.
  \label{eq:fixes-torsion-definition}
\end{equation}
Every homomorphism $e:\Gamma\to\mathbb Z^m$ maps each root of unity in
$\Gamma$ to $0$.
The \emph{exponent} of a finite group is the least positive integer $e$
such that $x^e=1$ for every group element $x$, equivalently the least
common multiple of the element orders.  We write $\lcm$ for least
common multiple.

If a homomorphism $e:\Gamma\to\mathbb Z^m$ has kernel
$\Gamma\cap\mu_\infty$, it induces the injective homomorphism
$\Gamma/(\Gamma\cap\mu_\infty)\hookrightarrow\mathbb Z^m$ given by
$x(\Gamma\cap\mu_\infty)\mapsto e(x)$.

\subsection{Constraint functions and instances}
\label{subsec:constraint-instances}

Let $\cF$ be a finite set of functions $f:D^{r_f}\to\AlgNums$ with
$r_f\ge1$, each given together with its arity and exact values.  Such a set
is also called a \emph{weighted constraint language}~\cite{CaiChen2017}.
If the input is a list, we identify functions with the same arity and the
same values. This leaves $W_{\cF}$ unchanged.  An instance $I$ of
$\CSP{\cF}$ consists of a variable set
$V(I)=\{x_1,\ldots,x_n\}$, where $n(I)=\abs{V(I)}=n\ge1$, and a finite
multiset $\mathcal C_I$ of constraints $(f;i_1,\ldots,i_{r_f})$, where
$f\in\cF$ and $i_1,\ldots,i_{r_f}\in[n]$.  A scope is ordered and its
indices need not be distinct.  The instance function $F_I:D^n\to\AlgNums$
is
\begin{equation}
  F_I(a_1,\ldots,a_n)
  =\prod_{(f;i_1,\ldots,i_{r_f})\in\mathcal C_I}
      f(a_{i_1},\ldots,a_{i_{r_f}}),
  \label{eq:instance-function}
\end{equation}
where multiplicities are retained and the empty product equals $1$.
For an assignment $\sigma:V(I)\to D$, we also write
$F_I(\sigma)=F_I(\sigma(x_1),\ldots,\sigma(x_n))$.
For $1\le t\le n$, define
\begin{equation}
  F_I^{[t]}(a_1,\ldots,a_t)
  =\sum_{(a_{t+1},\ldots,a_n)\in D^{n-t}}
      F_I(a_1,\ldots,a_n).
  \label{eq:partial-marginal}
\end{equation}
Following Cai and Chen~\cite{CaiChen2017}, define the set of functions
\begin{equation*}
  W_{\cF}
  =\left\{F_I^{[t]}:
    \begin{array}{l}
      I\text{ is a finite }\CSP{\cF}\text{ instance},\\
      1\le t\le n(I)
    \end{array}
    \right\}.
\end{equation*}
The variable ordering in an instance is arbitrary, so
\eqref{eq:partial-marginal} retains any prescribed ordered subset after
relabeling.  Repeated scope indices permit diagonal minors and arbitrary
identifications.  Instances may contain arbitrarily many variables,
constraints, hidden variables, and retained variables.

For $d\ge1$, the \emph{occurrence degree} of a variable is the number
of positions in all constraint scopes occupied by that variable, counting
constraint multiplicity and repeated positions.  Following
Lin~\cite{Lin2021}, $\CSPdeg{d}{\cF}$ denotes the restriction of
$\CSP{\cF}$ to instances in which every variable has occurrence degree
divisible by $d$.  The corresponding set of functions is
\begin{equation*}
  \WFdeg{d}
  =\left\{F_I^{[t]}:
    \begin{array}{l}
      I\text{ is a finite }\CSPdeg{d}{\cF}\text{ instance},\\
      1\le t\le n(I)
    \end{array}
    \right\}.
\end{equation*}
For $d=1$ this is $W_{\cF}$.  Isolated variables have degree zero and
are permitted.

\subsection{Rows and Block Orthogonality}

For every function $F:D^r\to\AlgNums$, write $\arity(F)=r$.  Following Cai and
Chen~\cite[Sections~2.3 and~2.5]{CaiChen2017}, suppose $r\ge2$ and regard
$F$ as the
$q^{r-1}\times q$ matrix whose row indexed by
$\mathbf x\in D^{r-1}$ is $F(\mathbf x,\ast)=\bigl(F(\mathbf x,z)\bigr)_{z\in D}$. The nonzero row indices carry the equivalence relation
\begin{equation}
  \mathbf x\sim_F\mathbf y
  \quad\Longleftrightarrow\quad
  F(\mathbf x,\ast)\text{ and }F(\mathbf y,\ast)
  \text{ are linearly dependent over }\mathbb C.
  \label{eq:complex-row-equivalence}
\end{equation}
Let $S_1,\ldots,S_s$ be its equivalence classes.  For each $j\in[s]$,
let $\mathbf v_j$ be the unique vector proportional to the rows in $S_j$
whose first nonzero entry, with respect to a fixed ordering of $D$, is $1$.
The family
\begin{equation*}
  \Row(F)=\{(S_1,\mathbf v_1),\ldots,(S_s,\mathbf v_s)\}
\end{equation*}
is the \emph{row representation} of $F$.  If $F$ is the zero function,
then $s=0$ and $\Row(F)=\varnothing$.

For a function or vector $H$, write $\abs H$ for its entrywise absolute value.
For a vector $\mathbf u\in\mathbb C^D$, write
$\supp(\mathbf u)=\{z\in D:u_z\ne0\}$.
The function $F$ is \emph{block-rank one} if the
nonnegative matrix $\abs F$ is block-rank one: for every two nonzero rows,
their magnitude vectors are positively proportional or their supports are
disjoint.  Equivalently, the nonzero support is a disjoint union of complete
row-by-column rectangles, and within each rectangle the magnitude rows are
proportional.  See Appendix~\ref{app:block-rank-one-example} for an example.

Let $\mathbf u,\mathbf v\in\mathbb C^D$ be nonzero with $\abs{\mathbf u}$
and $\abs{\mathbf v}$ proportional, and set
$T=\supp(\mathbf u)=\supp(\mathbf v)$.  Let
$m_1>\cdots>m_h>0$ be the distinct values of $\abs{u_z}$ on $T$, and
put $T_i=\{z\in T:\abs{u_z}=m_i\},\quad i\in[h]$. The same partition is obtained from $\mathbf v$.  The two vectors are
\emph{block-orthogonal}, written
$\mathbf u\perp_{\mathrm B}\mathbf v$, if
\begin{equation}
  \sum_{z\in T_i}u_z\overline{v_z}=0
  \qquad\text{for every }i\in[h].
  \label{eq:vector-block-orthogonality}
\end{equation}

\begin{definition}[Block-orthogonal function]
\label{def:block-orthogonal-function}
A function $F:D^r\to\AlgNums$, $r\ge2$, is \emph{block-orthogonal} if it
is block-rank one and, whenever two nonzero complex rows have proportional
magnitude vectors, those rows are either linearly dependent or
block-orthogonal in the sense of~\eqref{eq:vector-block-orthogonality}.
\end{definition}

See Appendix~\ref{app:block-orthogonality-example} for an example.

\subsection{Joint purification}

A function is
\emph{pure} if each of its values is a nonnegative integer times a root of
unity.  We use the notions of pure functions and generating sets, together
with the Purification Lemma, from Cai and
Chen~\cite[Definitions~2--3 and Lemma~7]{CaiChen2017}.

\begin{definition}[Generating set]
\label{def:generating-set}
Given a finite tuple
$\mathbf F=(F_1,\ldots,F_h)$ with $h\ge1$, let $A(\mathbf F)$ be the set of
its nonzero values.  Write $\mathbb Q(A(\mathbf F))$ for the smallest
subfield of $\AlgNums$ containing $\mathbb Q$ and $A(\mathbf F)$.
A \emph{generating set} for $A(\mathbf F)$ is a tuple
$g_1,\ldots,g_m$, where $m\ge0$, satisfying the following three conditions:
\begin{enumerate}
  \item\label{item:generators-field}
        $g_i\in\mathbb Q(A(\mathbf F))^\times$ for every $i\in[m]$.
  \item\label{item:generators-independent}
        For every nonzero
        $\mathbf k=(k_1,\ldots,k_m)\in\mathbb Z^m$,
        \[
          \prod_{i\in[m]} g_i^{k_i}\notin\mu_\infty.
        \]
  \item Every $a\in A(\mathbf F)$ has a unique exponent vector
        $e(a)=(e_1(a),\ldots,e_m(a))\in\mathbb Z^m$ such that
\begin{equation}
  a=\vartheta(a)\prod_{i\in[m]} g_i^{e_i(a)},
  \qquad \vartheta(a)\in\mu_\infty.
  \label{eq:multiplicative-decomposition}
\end{equation}
\end{enumerate}
\end{definition}

When $A(\mathbf F)$ is nonempty, such a generating set
exists~\cite[Lemma~7.3]{CaiChenLu2013}. See
also~\cite[Lemma~8]{CaiChen2017}.  Let $p_1,\ldots,p_m$ be the first $m$
positive primes.
For every nonzero value $a$, put
\[
  \varrho(a)=\prod_{i\in[m]} p_i^{e_i(a)}\in\mathbb Q_{>0}.
\]
The transformation displayed in the proof of the Purification
Lemma~\cite[Lemma~7]{CaiChen2017} replaces each nonzero $a$ by
$\vartheta(a)\varrho(a)$. Zeros remain zero.  Write $\widetilde F_i$ for
the resulting function corresponding to $F_i$.
Negative exponents in $\varrho(a)$ may introduce denominators.
To clear them, let $N_i$ be the least positive integer such that
$N_i\varrho(a)\in\mathbb Z_{>0}$ for every nonzero value $a$ of $F_i$, taking
$N_i=1$ when $F_i$ is zero, and put $F'_i=N_i\widetilde F_i$.
Every value of $F'_i$ is a nonnegative integer times a root of unity,
so $F'_i$ is pure.  We call $\widetilde F_i$ and $F'_i$ the
\emph{unnormalized} and \emph{normalized purifications}, respectively.
Multiplication by a positive scalar
changes neither support nor any clause of Block Orthogonality, so the unnormalized
and normalized forms define the same global condition.
Appendix~\ref{app:block-orthogonality-example} illustrates this purification.

Fix a legal generating-set choice whenever $A(\mathbf F)$ is nonempty and
define $\Pure(\mathbf F)=(F'_1,\ldots,F'_h)$, using the pure outputs,
normalized separately for each function.  If $A(\mathbf F)=\varnothing$,
define $\Pure(\mathbf F)$ to be the
all-zero tuple.  We call any unnormalized or normalized output obtained from a
legal generating-set choice a \emph{legal purification}.
\Cref{lem:purification-invariance} proves that the relevant structure is
independent of these choices.

\begin{definition}[Global Block Orthogonality condition]
\label{def:global-bo}
The set $\cF$ satisfies the \emph{Block Orthogonality condition} of Cai
and Chen~\cite[Section~3.1]{CaiChen2017}, denoted $\BO(\cF)$, if for every
nonempty finite tuple
$(F_1,\ldots,F_h)$ of pairwise distinct members of $W_{\cF}$, writing $\Pure(F_1,\ldots,F_h)=(F'_1,\ldots,F'_h)$, the function $F'_i$ is block-orthogonal for every $i$ with
$\arity(F_i)\ge2$.
\end{definition}

All functions in a finite tuple are purified using a common generating set,
with a separate normalization factor for each function.
\Cref{cor:singleton-purification} proves that checking one function at a
time gives the same condition.

\subsection{Types and derived relations}
\label{subsec:types-relations}

Following Cai and Chen~\cite[Sections~2.8--2.9 and~3.1]{CaiChen2017}, let
$F:D^r\to\AlgNums$, $r\ge2$, have row representation
$\{(S_j,\mathbf v_j):j\in[s]\}$.  For $0\le\ell\le r-1$ and
$\boldsymbol\alpha\in D^\ell$, define
\begin{equation*}
  \typeop_F(\boldsymbol\alpha)
  =\left\{j\in[s]:
       \text{some }\mathbf x\in S_j\text{ has prefix }
       \boldsymbol\alpha\right\}.
\end{equation*}
The map is a \emph{type-partition map} if
\begin{equation*}
  \begin{aligned}
  &\forall\,0\le\ell\le r-1\ \forall
    \boldsymbol\alpha,\boldsymbol\beta\in D^\ell,\\
  &\typeop_F(\boldsymbol\alpha)=\typeop_F(\boldsymbol\beta)
    \quad\text{or}\quad
    \typeop_F(\boldsymbol\alpha)
      \cap\typeop_F(\boldsymbol\beta)=\varnothing.
  \end{aligned}
\end{equation*}
Set $\typeop_F(\emptytuple)=[s]$.  If $F$ is the zero function, then $s=0$, so
every type is empty.
See Appendix~\ref{app:type-partition-example} for an example.

For every $F:D^r\to\AlgNums$, define its support relation $\SuppRel F=\{\mathbf a\in D^r:F(\mathbf a)\ne0\}$. For $r\ge2$, define the $2(r-1)$-ary row-equivalence relation
\begin{equation*}
  \RowEq F
  =\bigcup_{j\in[s]}(S_j\times S_j)
  =\{(\mathbf x,\mathbf y)\in D^{r-1}\times D^{r-1}:
      \mathbf x\sim_F\mathbf y\}.
\end{equation*}
If $F$ is the zero function, then $\SuppRel F=\varnothing$ and,
when $r\ge2$, $\RowEq F=\varnothing$.

An operation $\varphi:D^3\to D$ \emph{preserves} a relation $R\subseteq D^r$ if
coordinatewise application of $\varphi$ to any three tuples of $R$ returns a tuple
of $R$.  In this case $\varphi$ is a \emph{polymorphism} of $R$.
It is a \emph{Mal'tsev operation} if
\begin{equation*}
  \varphi(a,b,b)=\varphi(b,b,a)=a
  \qquad\text{for all }a,b\in D.
\end{equation*}
Set
\begin{equation*}
  \Lambda_{\cF}
  =\{\SuppRel F:F\in W_{\cF}\}
    \cup\{\RowEq F:F\in W_{\cF},\ \arity(F)\ge2\}.
\end{equation*}
$\Lambda_{\cF}$ is a family of relations of varying arities, and preservation
is required separately for every $\SuppRel F$ and every applicable
$\RowEq F$.

\begin{definition}[Type Partition and Mal'tsev polymorphism conditions]
\label{def:type-maltsev}
The set $\cF$ satisfies the \emph{Type Partition condition}, denoted
$\TP(\cF)$, if every function in $W_{\cF}$ of arity at least two has a
type-partition map.  It satisfies the \emph{Mal'tsev polymorphism condition},
denoted $\Mal(\cF)$, if
\begin{equation*}
  \exists \varphi:D^3\to D:\quad
  \varphi\text{ is a Mal'tsev operation}\quad\text{and}\quad
  \forall R\in\Lambda_{\cF},\ \varphi\text{ preserves }R.
\end{equation*}
Equivalently, $\Lambda_{\cF}$ has a common Mal'tsev polymorphism.
\end{definition}

Appendix~\ref{app:maltsev-example} illustrates how one Mal'tsev operation
preserves both the support and row-equivalence relations of a function.

For $d\ge1$, the conditions $\BO^d(\cF)$, $\TP^d(\cF)$, and
$\Mal^d(\cF)$ are given by \cref{def:global-bo,def:type-maltsev}
with $W_{\cF}$ replaced throughout
by $\WFdeg{d}$.

We now consider $\cF=\varnothing$.  In this case
$F_I\equiv1$, and every function $F=F_I^{[t]}$ in $W_{\cF}$ is the positive
constant function $F(a_1,\ldots,a_t)=q^{n(I)-t}$. Consequently $\SuppRel F=D^t$.  When $t\ge2$, all rows are nonzero and
proportional, there is exactly one complex row class, and
$\RowEq F=D^{2(t-1)}$.  Every legal purification is again a nonzero constant
function, so Block Orthogonality holds. Since $F$ itself is constant, Type
Partition holds as well.  The operation
\[
  \varphi_0(a,b,c)=
  \begin{cases}
    c,&a=b,\\
    a,&a\ne b
  \end{cases}
\]
is Mal'tsev and preserves every full relation.  Thus $\CC(\cF)$ holds. The
same argument applies to every $\CC^d(\cF)$ because constraint-free
instances have degree zero.  The decision procedures accept this case
immediately, and henceforth we may assume $\cF\ne\varnothing$.

\section{Purification Tests and Arity Compression}
\label{sec:purification-compression}

To bound the arity of a counterexample to the Block Orthogonality
condition, we will identify variables while preserving two rows that
witness its failure.  This requires the relevant tests to be independent
of the other functions purified jointly with the given function.  We first
establish this invariance and reduce to singleton purification, then
compress the arity and express the purified tests in the original values.

\subsection{Singleton purification}

Joint purification uses one generating set for the nonzero values of the
whole tuple.
We show that the tests defining Block Orthogonality for a given function
are independent of both the generator choice and the other functions in
the tuple.  This allows us to test each generated function separately and
later to transfer a failed test through variable identification.
We also establish the properties of the purification maps used in
\cref{subsec:intrinsic-purification-tests}.

\begin{lemma}
\label{lem:purification-invariance}
Let $F:D^r\to\AlgNums$, where $r\ge2$, occur in each of two finite
tuples of functions.  Fix an occurrence of $F$ in each tuple, and let
$F'$ and $F''$ be the corresponding outputs of two legal joint
purifications.  Then:
\begin{enumerate}
  \item\label{item:purification-support}
        $\SuppRel{F'}=\SuppRel{F''}$.
  \item\label{item:purification-magnitudes}
        For all $\mathbf x,\mathbf y\in D^{r-1}$,
        $F'(\mathbf x,\ast)$ and $F'(\mathbf y,\ast)$ have disjoint
        supports if and only if $F''(\mathbf x,\ast)$ and
        $F''(\mathbf y,\ast)$ do.  The same equivalence holds for
        positive proportionality of their magnitude vectors.
  \item\label{item:purification-row-equivalence}
        $\RowEq{F'}=\RowEq{F''}$.
  \item\label{item:purification-block-sums}
        If the rows indexed by $\mathbf x$ and $\mathbf y$ have
        proportional magnitude vectors in $F'$, then the partition of
        their common support into equal-magnitude blocks is the same for
        $F'$ and $F''$.  On each block, the inner product of these two
        rows restricted to that block is zero for $F'$ if and only if
        it is zero for $F''$.
\end{enumerate}
In particular, $F'$ is block-orthogonal if and only if $F''$ is
block-orthogonal.

For either chosen purification, let $e(a)$ and $\vartheta(a)$ be the
exponent vector and root-of-unity factor of a nonzero value $a$ of $F$
in~\eqref{eq:multiplicative-decomposition}.  With $m$ generators, the
purification replaces $a$ by
\[
  N\kappa(a),\qquad
  \kappa(a)=\vartheta(a)\prod_{i\in[m]}p_i^{e_i(a)},
\]
where $N>0$ is common to all values of $F$, with $N=1$ for an
unnormalized output.
Let $\Gamma$ be the multiplicative group generated by the nonzero
values of $F$.  These assignments extend to well-defined homomorphisms
$e:\Gamma\to\mathbb Z^m$, $\vartheta:\Gamma\to\mu_\infty$, and
$\kappa:\Gamma\to\mathbb C^\times$.
The extensions satisfy the same formula for $\kappa$ and have the
following properties:
\begin{itemize}
  \item $\ker e=\Gamma\cap\mu_\infty$.
  \item $\vartheta$ fixes every root of unity in $\Gamma$.
  \item $\kappa$ is injective, and
        $|\kappa(x)|=1$ if and only if $x\in\mu_\infty$.
\end{itemize}
For a singleton purification, the root-of-unity map also satisfies
$\vartheta(\Gamma)\subseteq\mu(K)$ whenever $K$ is a number field
containing the values of $F$.
\end{lemma}

\begin{proof}
If $F$ is the zero function, both purified functions are zero and all
assertions hold.  Assume that $F$ has a nonzero value, and write
$\Gamma=\langle a:a\text{ is a nonzero value of }F\rangle$.

We first prove that the assignments in the purification formula extend
to homomorphisms with the stated properties.
Fix one of the two joint purifications, and let $c_1,\ldots,c_t$ be the
nonzero values in its tuple,
with legal generators $g_1,\ldots,g_m$ and exponent vectors
$e(c_j)=(e_1(c_j),\ldots,e_m(c_j))$.
Put $H=\langle c_1,\ldots,c_t\rangle$.  Every nonzero value of $F$ occurs
among the $c_j$, so $\Gamma\le H$.
To extend the prescribed exponent vectors to $H$, set
\begin{equation}
  e\left(\prod_{j\in[t]} c_j^{n_j}\right)
  =\sum_{j\in[t]} n_j e(c_j)
  \qquad(n_1,\ldots,n_t\in\mathbb Z).
  \label{eq:extended-exponent-map}
\end{equation}
To check that this is well-defined, suppose
$\prod_{j\in[t]}c_j^{n_j}=1$.
Then~\eqref{eq:multiplicative-decomposition} gives
\begin{equation*}
  \prod_{i\in[m]}
       g_i^{\sum_{j\in[t]} n_j e_i(c_j)}
  =\prod_{j\in[t]}\vartheta(c_j)^{-n_j}\in\mu_\infty.
\end{equation*}
Condition~\ref{item:generators-independent} of \cref{def:generating-set} forces
$\sum_{j\in[t]}n_j e_i(c_j)=0$ for every $i\in[m]$.
If two products represent the same element of $H$, their quotient is
$1$, so this calculation shows that their exponent sums agree.
Thus~\eqref{eq:extended-exponent-map} is well-defined.
Multiplication of products adds exponent vectors, so $e$ is the unique
homomorphism on $H$ taking each $c_j$ to its prescribed vector.
Restricting it to $\Gamma$ gives $e:\Gamma\to\mathbb Z^m$.

We next determine its kernel.  For $x=\prod_j c_j^{n_j}\in\Gamma$,
if $e(x)=0$, the decompositions of the $c_j$ give
$x=\prod_{j\in[t]}\vartheta(c_j)^{n_j}\in\mu_\infty$.
Conversely, if $x^N=1$ for some $N\ge1$, then
$N e(x)=e(x^N)=0$, which implies $e(x)=0$ in $\mathbb Z^m$.
Hence $\ker e=\Gamma\cap\mu_\infty$, and $e$ induces an injective
homomorphism
$\nu:\Gamma/(\Gamma\cap\mu_\infty)\hookrightarrow\mathbb Z^m$,
given by $\nu\bigl(x(\Gamma\cap\mu_\infty)\bigr)=e(x)$.
For $x\in\Gamma$, define
\[
  \vartheta(x)=x\prod_{i\in[m]} g_i^{-e_i(x)}.
\]
The decomposition of $x$ shows that $\vartheta(x)\in\mu_\infty$, and
additivity of $e$ makes $\vartheta$ a homomorphism.
For $\xi\in\Gamma\cap\mu_\infty$, we have $e(\xi)=0$ and hence
$\vartheta(\xi)=\xi$, as in~\eqref{eq:fixes-torsion-definition}.
It follows that
\[
  \kappa(x)=\vartheta(x)\prod_{i\in[m]} p_i^{e_i(x)}
\]
is a homomorphism $\Gamma\to\mathbb C^\times$ that fixes every root of
unity in $\Gamma$.
Unique factorization in $\mathbb Q_{>0}$ gives the multiplicative
independence of the distinct primes: for integers $k_1,\ldots,k_m$,
$\prod_{i\in[m]}p_i^{k_i}=1$ implies $k_1=\cdots=k_m=0$.
Consequently, $\kappa$ is injective.  Indeed, if $\kappa(x)=1$, taking
absolute values gives $e(x)=0$.  Then $x\in\mu_\infty$ and
$\kappa(x)=\vartheta(x)=x$, so $x=1$.

The magnitude of $\kappa(a)$ is $\prod_i p_i^{e_i(a)}$, which may
have denominators when exponents are negative.  Since $F$ has finitely many values,
there is a positive integer $N_F$ such that
$N_F\prod_{i\in[m]}p_i^{e_i(a)}\in\mathbb Z_{>0}$ for every nonzero
value $a$ of $F$.
Thus replacing $a$ by $N_F\kappa(a)$ gives a pure function.
Each product of two entries and each block sum below is multiplied by
$N_F^2$.  This preserves equality of the products and whether a block
sum equals zero.

The magnitude comparisons below follow from
\begin{equation}
  \abs{\kappa(x)}=1
  \quad\Longleftrightarrow\quad
  e(x)=0
  \quad\Longleftrightarrow\quad
  x\in\mu_\infty.
  \label{eq:purification-unit-modulus}
\end{equation}
Indeed, $\vartheta(x)$ has modulus one, and unique factorization in
$\mathbb Q_{>0}$ gives $\prod_i p_i^{e_i(x)}=1$ exactly when every
exponent equals zero.

For a singleton purification of a $K$-valued function,
condition~\ref{item:generators-field} of \cref{def:generating-set}
places every $g_i$ in the field generated by its values,
and hence in $K$.
Since $\Gamma\subseteq K^\times$, the formula for $\vartheta$ gives
$\vartheta(\Gamma)\subseteq\mu(K)$.

We now compare the two purifications.
Applying this construction to each gives homomorphisms $\kappa_1$ and
$\kappa_2$ and positive factors $N_1$ and $N_2$.
For a value $a$ of $F$, write $\widehat a_1$ and $\widehat a_2$ for its
images in $F'$ and $F''$, respectively.  Both are zero when $a=0$.
For $a\ne0$, they satisfy $\widehat a_1=N_1\kappa_1(a)$ and
$\widehat a_2=N_2\kappa_2(a)$.
Since both homomorphisms fix roots of unity, for
$\xi\in\Gamma\cap\mu_\infty$ and $x\in\Gamma$ we have
\begin{equation}
  \kappa_1(\xi x)=\xi\kappa_1(x),
  \qquad
  \kappa_2(\xi x)=\xi\kappa_2(x).
  \label{eq:purification-torsion-covariance}
\end{equation}
Consequently, if both $a$ and $\xi a$ occur as values of $F$, then
$\widehat{\xi a}_{1}=\xi\widehat a_1$ and
$\widehat{\xi a}_{2}=\xi\widehat a_2$.
Thus root-of-unity ratios are preserved.  A root-of-unity value $\xi$
itself is replaced by $N_1\xi$ or $N_2\xi$.

We now prove~(\ref{item:purification-support})
and~(\ref{item:purification-magnitudes}), and establish the common
blocks required in~(\ref{item:purification-block-sums}).
For nonzero values $a,b,c,u$ of $F$, applying
\eqref{eq:purification-unit-modulus} to $a/b$ gives
\begin{equation}
  \abs{\widehat a_1}=\abs{\widehat b_1}
  \quad\Longleftrightarrow\quad a/b\in\mu_\infty
  \quad\Longleftrightarrow\quad
  \abs{\widehat a_2}=\abs{\widehat b_2}.
  \label{eq:pure-equal-magnitude}
\end{equation}
Applying it to $au/(bc)$ gives
\begin{equation*}
  \abs{\widehat a_1\widehat u_1}=\abs{\widehat b_1\widehat c_1}
  \quad\Longleftrightarrow\quad au/(bc)\in\mu_\infty
  \quad\Longleftrightarrow\quad
  \abs{\widehat a_2\widehat u_2}=\abs{\widehat b_2\widehat c_2}.
\end{equation*}
Zeros remain zero and nonzero values remain nonzero, proving
(\ref{item:purification-support}) and preservation of disjoint row supports.
A product with a zero factor remains zero, while a product of nonzero
entries remains nonzero.  Together with the second equivalence, this
preserves whether each $2\times2$ minor of the magnitude rows equals
zero.  These equalities and the unchanged supports preserve positive
proportionality of magnitude rows, completing~(\ref{item:purification-magnitudes}).
The first equivalence gives the common equal-magnitude blocks
in~(\ref{item:purification-block-sums}).

To prove~(\ref{item:purification-row-equivalence}), we use the fact that
complex linear dependence of two rows is determined by their
$2\times2$ minors.  Write the original rows as
$\mathbf a=(a_z)_{z\in D}$ and $\mathbf b=(b_z)_{z\in D}$.
For columns $z,w\in D$, the minor is $a_zb_w-a_wb_z$.
When these four entries are nonzero, injectivity and multiplicativity give
\[
  \widehat{a_z}_{1}\widehat{b_w}_{1}
    =\widehat{a_w}_{1}\widehat{b_z}_{1}
  \quad\Longleftrightarrow\quad a_zb_w=a_wb_z
  \quad\Longleftrightarrow\quad
  \widehat{a_z}_{2}\widehat{b_w}_{2}
    =\widehat{a_w}_{2}\widehat{b_z}_{2}.
\]
If an entry is zero, preservation of zero and nonzero products gives
the same equivalence.
Thus every minor equals zero under $F'$ exactly when it equals zero
under $F''$.  Together with the common support from
(\ref{item:purification-support}), this gives
$\RowEq{F'}=\RowEq{F''}$, proving~(\ref{item:purification-row-equivalence}).

We complete~(\ref{item:purification-block-sums}) by comparing the inner
products on the common blocks established above.
Suppose the two rows have proportional magnitudes after
purification.
Let $T$ be one of their common equal-magnitude blocks and fix $z_0\in T$.
By~\eqref{eq:pure-equal-magnitude}, there are roots of unity
$\xi_z,\eta_z\in\Gamma\cap\mu_\infty$ such that
$a_z=\xi_z a_{z_0}$ and $b_z=\eta_z b_{z_0}$ for $z\in T$.
Equation~\eqref{eq:purification-torsion-covariance} and the common scaling
factor within each purified function preserve these ratios.
Consequently,
\begin{align*}
  \sum_{z\in T}\widehat{a_z}_{1}\overline{\widehat{b_z}_{1}}
  &=\widehat{a_{z_0}}_{1}\overline{\widehat{b_{z_0}}_{1}}
    \sum_{z\in T}\xi_z\overline{\eta_z},\\
  \sum_{z\in T}\widehat{a_z}_{2}\overline{\widehat{b_z}_{2}}
  &=\widehat{a_{z_0}}_{2}\overline{\widehat{b_{z_0}}_{2}}
    \sum_{z\in T}\xi_z\overline{\eta_z}.
\end{align*}
The two prefactors are nonzero, so the two block sums equal zero
simultaneously, completing~(\ref{item:purification-block-sums}).
\end{proof}

The lemma makes each function's Block Orthogonality test independent of
the other functions purified with it.  The global condition therefore
reduces to the following singleton test.

\begin{corollary}[Singleton purification suffices]
\label{cor:singleton-purification}
The joint condition $\BO(\cF)$ holds if and only if every
$F\in W_{\cF}$ of arity at least two has a block-orthogonal singleton
purification.
\end{corollary}

\begin{proof}
If the joint condition holds, apply it to the one-element tuple $(F)$ for
each $F\in W_{\cF}$ of arity at least two.

Conversely, assume every $F\in W_{\cF}$ of arity at least two has a
block-orthogonal singleton purification.
Fix a nonempty finite tuple $\mathbf F=(F_1,\ldots,F_h)$ from $W_{\cF}$
and a legal joint purification $\Pure(\mathbf F)=(F'_1,\ldots,F'_h)$.
For each $i$ with $\arity(F_i)\ge2$, choose a block-orthogonal singleton
purification $F''_i$ of $F_i$.
Apply~\cref{lem:purification-invariance} to $F_i$ in the tuples $(F_i)$
and $(F_1,\ldots,F_h)$.
It gives that $F'_i$ is block-orthogonal if and only if $F''_i$ is
block-orthogonal, so every jointly purified component of arity at least
two is block-orthogonal.
Since $\mathbf F$ was arbitrary, the joint condition $\BO(\cF)$ holds.
\end{proof}

\subsection{Arity compression}

We now show that a failure of $\BO(\cF)$ has a witness of arity at most
$q^2+1$.  By~\cref{cor:singleton-purification}, we can start with one
generated function whose singleton purification fails Block Orthogonality.
The construction identifies variables while retaining the two rows
that witness this failure.  We first show that identifying retained
variables keeps the resulting function in $W_{\cF}$.

\begin{lemma}
\label{lem:diagonal-minors}
Let $F\in W_{\cF}$ have arity $r$, let $1\le s\le r$, and let
$\pi:[r]\twoheadrightarrow[s]$ be a surjection.  Then
\[
  H(y_1,\ldots,y_s)
  =F(y_{\pi(1)},\ldots,y_{\pi(r)})
\]
belongs to $W_{\cF}$.
\end{lemma}

\begin{proof}
Choose a $\CSP{\cF}$ instance $I$ with $F=F_I^{[r]}$ and retained
variables $x_1,\ldots,x_r$.
Introduce new retained variables $y_1,\ldots,y_s$ and replace every
occurrence of $x_i$ in a constraint by $y_{\pi(i)}$.
Keep all variables being summed out unchanged and distinct from the
$y_j$.  Repeated variables are allowed in constraint scopes, so this
gives a $\CSP{\cF}$ instance $I_\pi$.
Summing over the same hidden variables gives $F_{I_\pi}^{[s]}=H$, hence
$H\in W_{\cF}$.
\end{proof}

Appendix~\ref{app:variable-identification-example} illustrates this
construction.
To preserve two selected rows, we group their variable positions by the
ordered pairs of domain values assigned there.  There are at most $q^2$
such pairs.  Each group needs one retained variable, and the column
variable remains separate, giving the following bound.

\begin{lemma}
\label{lem:bo-arity-bound}
If the singleton Block Orthogonality condition fails on $W_{\cF}$,
then it fails for some function in $W_{\cF}$ of arity at most $q^2+1$.
\end{lemma}

\begin{proof}
Let $F\in W_{\cF}$ have arity $r\ge2$, and choose a legal singleton
purification $F'$ that fails Block Orthogonality.
There are two row indices $\mathbf x,\mathbf y\in D^{r-1}$ witnessing
this failure in $F'$.
Either their magnitude rows have overlapping supports but are not
positively proportional, or their magnitude rows are proportional
while the complex rows are neither linearly dependent nor
block-orthogonal.
Group the positions $i,j\in[r-1]$ by
\[
  i\equiv j
  \quad\Longleftrightarrow\quad
  (x_i,y_i)=(x_j,y_j).
\]
There is at least one class because $r\ge2$, and at most $q^2$ classes
because each corresponds to a distinct pair in $D^2$.
Identify the variables within each class and keep the column variable
separate.  By~\cref{lem:diagonal-minors}, the resulting function $H$
belongs to $W_{\cF}$ and satisfies $2\le\arity(H)\le q^2+1$.

For each new row variable, assign the first component of its class's
ordered pair to obtain one row of $H$, and the second component to
obtain the other.
These assignments recover $\mathbf x$ and $\mathbf y$ when expanded
to the original positions.  Since the column variable is unchanged,
the two rows of $H$ equal $F(\mathbf x,\ast)$ and $F(\mathbf y,\ast)$.

To transfer the failed test, take a legal joint unnormalized
purification $(\widetilde F,\widetilde H)$ of $(F,H)$.
By~\cref{lem:purification-invariance}, the test that fails in $F'$ also
fails in the corresponding rows of $\widetilde F$.
Joint unnormalized purification applies the same map to the values of
both functions, so the two row equalities persist in
$(\widetilde F,\widetilde H)$.
Thus the same test fails in $\widetilde H$.
Applying~\cref{lem:purification-invariance} to $H$ transfers this failure
to every legal singleton purification of $H$.
Hence $H$ is a counterexample with the required arity bound.
\end{proof}

Appendix~\ref{app:arity-compression-example} illustrates the grouping of
positions and the two recovered rows.

\subsection{Intrinsic purification tests}
\label{subsec:intrinsic-purification-tests}

\Cref{lem:bo-arity-bound} reduces the decision problem to finitely many
arities.  At each arity, we need equations in the original function values
that test Block Orthogonality after purification.  We first express
comparisons of purified magnitudes and products in this form.  A common
number field restricts the root-of-unity factors in these equations to a
finite set.

Fix the computable conjugation-stable number field
\begin{equation}
  K=\mathbb Q\left(
      \{f(\mathbf a),\overline{f(\mathbf a)}:
        f\in\cF,\ \mathbf a\in D^{r_f}\}
    \right),
  \label{eq:working-field}
\end{equation}
and let $\mu(K)$ be its finite, effectively computable group of roots of
unity~\cite[Algorithm~4.9.10]{Cohen1993}.  Every input value $f(\mathbf a)$ belongs to $K$.
By \cref{eq:instance-function,eq:partial-marginal}, each value of every
$F\in W_{\cF}$ is a finite sum of finite products of input values, so it
also belongs to $K$.

The next lemma gives the tests used in the certificates of
\cref{sec:bo-certificates}: magnitude comparisons become equations with a
root-of-unity factor, ordinary minors retain their original equations, and
root-of-unity ratios are preserved.

\Needspace{8\baselineskip}
\begin{lemma}
\label{lem:intrinsic-purification-tests}
Let $F$ be a $K$-valued function on a finite domain, let $F'$ be any legal singleton
purification of $F$, and, for every nonzero value $a$ occurring in $F$,
write $\widehat a$ for its purified value.  For nonzero values $a,b,c,u$
occurring in $F$,
\begin{align}
  \abs{\widehat a}=\abs{\widehat b}
  &\quad\Longleftrightarrow\quad
  a/b\in\mu(K),
  \label{eq:intrinsic-equal-magnitude}\\
  \abs{\widehat a}\abs{\widehat u}
     =\abs{\widehat b}\abs{\widehat c}
  &\quad\Longleftrightarrow\quad
  au/(bc)\in\mu(K)\notag\\
  &\quad\Longleftrightarrow\quad
  au=\zeta bc
  \text{ for some }\zeta\in\mu(K),
  \label{eq:intrinsic-magnitude-minor}\\
  \widehat a\,\widehat u=\widehat b\,\widehat c
  &\quad\Longleftrightarrow\quad
  au=bc.
  \label{eq:intrinsic-ordinary-minor}
\end{align}
Moreover, if $a=\rho b$ for $\rho\in\mu(K)$, then
\begin{equation}
  \widehat a=\rho\,\widehat b.
  \label{eq:purification-root-covariance}
\end{equation}
\end{lemma}

\begin{proof}
If $F$ is the zero function, the assertions are vacuous.  Assume that $F$ has a
nonzero entry.
Let $\Gamma\le K^\times$ be generated by the nonzero entries of $F$.
The construction in \cref{lem:purification-invariance} gives an injective
homomorphism $\kappa:\Gamma\to K^\times$ fixing every root of unity in
$\Gamma$, with $|\kappa(x)|=1$ exactly when $x\in\mu_\infty$, and a positive factor
$N$ such that $\widehat a=N\kappa(a)$ for every nonzero entry $a$.

We first prove \cref{eq:intrinsic-equal-magnitude,eq:intrinsic-magnitude-minor}.
Multiplicative independence of the replacement primes gives
\[
  \abs{\kappa(a/b)}=1
  \quad\Longleftrightarrow\quad
  a/b\in\Gamma\cap\mu(K).
\]
The common factor $N$ cancels when comparing
$\abs{\widehat a}$ and $\abs{\widehat b}$, giving
\eqref{eq:intrinsic-equal-magnitude}.  Applying the same argument to
$au/(bc)$, with the common factor $N^2$, gives
\eqref{eq:intrinsic-magnitude-minor}.

For \eqref{eq:intrinsic-ordinary-minor}, cancel the common factor $N^2$
and use injectivity of $\kappa$.  Finally, if $a=\rho b$, then
$\rho=a/b\in\Gamma$.  Since $\kappa$ fixes every root of unity in $\Gamma$
as in~\eqref{eq:fixes-torsion-definition}, multiplying by the common factor
$N$ gives~\eqref{eq:purification-root-covariance}.
\end{proof}

Together with support preservation, \eqref{eq:intrinsic-ordinary-minor}
shows that all ordinary $2\times2$ minors of two original rows equal zero
if and only if the corresponding purified minors do.  Thus row dependence
can be tested directly in the original values.  We now combine these
tests with the magnitude and root-of-unity tests to describe
Block Orthogonality by polynomial equations.

\section{Algebraic Certificates for Block Orthogonality}
\label{sec:bo-certificates}

At each fixed arity, we construct finitely many algebraic certificates
from the tests in \cref{lem:intrinsic-purification-tests}.  A function has
a block-orthogonal singleton purification exactly when it satisfies the
equations of at least one certificate.  We then combine these alternatives
into one finite set of polynomial identities, which will be tested on all
generated functions in \cref{sec:universal-identities}.  Throughout, $K$
is the field in~\eqref{eq:working-field}, and $\mu(K)$ is its group of
roots of unity.

\subsection{Finite certificates}
\label{subsec:finite-bo-certificates}

Fix an arity $k\ge2$ and put $\mathcal R_k=D^{k-1}$.  We use $A$ and $A'$
for the matrix forms of a $K$-valued $k$-ary function $F$ and its singleton
purification $F'$, respectively: their entries are
$a_{\mathbf x,c}=F(\mathbf x,c)$ and
$a'_{\mathbf x,c}=F'(\mathbf x,c)$ for
$\mathbf x\in\mathcal R_k$ and $c\in D$.  This identifies the original
function with a point of $K^{D^k}$.  Let
$\mathbf X=\{X_{\mathbf x,c}:\mathbf x\in\mathcal R_k,\ c\in D\}$ be
the formal variables corresponding to its entries.  Fix arbitrary total
orders on $\mathcal R_k$ and $D$.

\begin{definition}[Block Orthogonality certificate]
\label{def:bo-certificate}
A \emph{Block Orthogonality certificate} of arity $k$ consists of the
following choices and their associated polynomial equations.

\emph{Support rectangles.}
Choose a rectangle system $(R_1,C_1),\ldots,(R_t,C_t)$, where every
$R_i\subseteq\mathcal R_k$ and $C_i\subseteq D$ is nonempty,
the row sets $R_i$ are pairwise disjoint, and the column sets $C_i$ are
pairwise disjoint.  The empty system $t=0$ is allowed.  Impose the support
equations
\begin{equation}
  X_{\mathbf x,c}=0
  \qquad
  \text{whenever }
  (\mathbf x,c)\notin
  \bigcup_{i\in[t]}(R_i\times C_i).
  \label{eq:bo-certificate-support}
\end{equation}

\emph{Magnitude comparisons.}
For every $i$, every two distinct rows
$\mathbf x,\mathbf y\in R_i$, and every two distinct columns
$c,e\in C_i$, choose
$\tau_{\mathbf x,\mathbf y;c,e}\in\mu(K)$ and impose the twisted minor
\begin{equation}
  X_{\mathbf x,c}X_{\mathbf y,e}
  -\tau_{\mathbf x,\mathbf y;c,e}
   X_{\mathbf x,e}X_{\mathbf y,c}=0.
  \label{eq:bo-certificate-twisted-minor}
\end{equation}
\emph{Row comparisons.}
For every unordered pair $\mathbf x<\mathbf y$ in one $R_i$, choose one of
the following two modes.

\emph{Dependent mode.}
Impose every ordinary minor
\begin{equation}
  X_{\mathbf x,c}X_{\mathbf y,e}
  -X_{\mathbf x,e}X_{\mathbf y,c}=0
  \qquad(c,e\in C_i).
  \label{eq:bo-certificate-dependent}
\end{equation}

\emph{Orthogonal mode.}
Choose a set partition $\mathcal T_{\mathbf x,\mathbf y}$ of $C_i$.  For
every $T\in\mathcal T_{\mathbf x,\mathbf y}$, choose an anchor $c_T\in T$
and roots $\rho_{T,c},\sigma_{T,c}\in\mu(K) \quad(c\in T)$, with $\rho_{T,c_T}=\sigma_{T,c_T}=1$.  Retain this discrete choice only if
\begin{equation}
  \sum_{c\in T}\rho_{T,c}\overline{\sigma_{T,c}}=0
  \qquad
  \text{for every }T\in\mathcal T_{\mathbf x,\mathbf y}.
  \label{eq:bo-certificate-root-sum}
\end{equation}
For every $T\in\mathcal T_{\mathbf x,\mathbf y}$ and every
$c\in T\setminus\{c_T\}$, impose the anchor equations
\begin{align}
  X_{\mathbf x,c}-\rho_{T,c}X_{\mathbf x,c_T}&=0,
  \label{eq:bo-certificate-first-anchor}\\
  X_{\mathbf y,c}-\sigma_{T,c}X_{\mathbf y,c_T}&=0.
  \label{eq:bo-certificate-second-anchor}
\end{align}

For a retained certificate $\mathfrak c$, let
$E_{\mathfrak c}\subseteq K[\mathbf X]$ be the set of left-hand sides of
its selected equations among
\eqref{eq:bo-certificate-support}--\eqref{eq:bo-certificate-dependent},
\eqref{eq:bo-certificate-first-anchor}, and
\eqref{eq:bo-certificate-second-anchor}.  We use these polynomials to
describe the common solution set of the equations.
\end{definition}

The root sums in~\eqref{eq:bo-certificate-root-sum} are checked when
enumerating the choices in \cref{def:bo-certificate}.  Each retained
$E_{\mathfrak c}$ consists of homogeneous polynomials in $\mathbf X$:
support and anchor equations have degree one, and minor equations have
degree two.  Conjugation occurs only in the root sums over $\mu(K)$.

The row and column sets, their partitions, and $\mu(K)$ are finite, and
the root sums can be checked by exact number-field equality tests.
Hence the family $\mathfrak C_k$ of retained certificates is finite and
effectively enumerable.  Homogeneity makes the zero function a solution
of every certificate's equations.

For $E\subseteq K[\mathbf X]$, write
$V_K(E)=\{A\in K^{D^k}:e(A)=0\text{ for every }e\in E\}$.
Let $\mathcal B_k$ be the set of $K$-valued $k$-ary functions whose
singleton purifications are block-orthogonal, a property independent of
the purification chosen by \cref{lem:purification-invariance}.

The next theorem shows that these certificates describe $\mathcal B_k$.
The forward direction uses \cref{lem:intrinsic-purification-tests} to
choose the roots and row comparisons from a block-orthogonal purification.
For the reverse direction, the support and twisted-minor equations
determine the surviving support and give block-rank one after
purification.  The chosen mode then gives
linear dependence or block orthogonality for each pair of overlapping
nonzero rows.

\begin{theorem}[Algebraic description of the BO locus]
\label{thm:finite-bo-certificates}
For every $k\ge2$, the effectively enumerable certificate family
$\mathfrak C_k$ satisfies
\begin{equation}
  \mathcal B_k
  =\bigcup_{\mathfrak c\in\mathfrak C_k}V_K(E_{\mathfrak c}).
  \label{eq:bo-certificate-union}
\end{equation}
\end{theorem}

\begin{proof}
For the forward inclusion, let $A\in\mathcal B_k$ represent $F$, and fix a singleton purification
$F'$ with matrix $A'=(a'_{\mathbf x,c})$.  The support of $A'$ equals
that of $A$.
Because $\abs{A'}$ is block-rank one, the nonempty connected components of its
bipartite support graph are complete bipartite graphs $R_1\times C_1,\ldots,R_t\times C_t$, with pairwise disjoint row sets and pairwise disjoint column sets.  Choose
these components as the certificate rectangles.  This gives
\eqref{eq:bo-certificate-support}.

To choose the twisted minors, use proportionality of the purified
magnitude rows within each rectangle.  For
$\mathbf x,\mathbf y\in R_i$ and $c,e\in C_i$,
\[
  \abs{a'_{\mathbf x,c}}\abs{a'_{\mathbf y,e}}
  =\abs{a'_{\mathbf x,e}}\abs{a'_{\mathbf y,c}}.
\]
All four original entries are nonzero.  By
\eqref{eq:intrinsic-magnitude-minor}, there is a root
$\tau_{\mathbf x,\mathbf y;c,e}\in\mu(K)$ for which the corresponding
twisted minor equals zero.

It remains to choose a mode for each pair $\mathbf x<\mathbf y$ in $R_i$.
If their purified rows are linearly
dependent, then~\eqref{eq:intrinsic-ordinary-minor} shows that every ordinary
minor of the original rows equals zero.  We choose dependent mode.

For a linearly independent pair, choose orthogonal mode and partition
$C_i$ into the nonempty level sets
$T$ of $c\mapsto\abs{a'_{\mathbf x,c}}$.  Since the magnitude rows are
proportional, these are also the level sets for the other row.  Choose any
anchor $c_T\in T$.  For every $c\in T$,
\[
  \rho_{T,c}=\frac{a_{\mathbf x,c}}{a_{\mathbf x,c_T}},
  \qquad
  \sigma_{T,c}=\frac{a_{\mathbf y,c}}{a_{\mathbf y,c_T}}
\]
belong to $\mu(K)$ by~\eqref{eq:intrinsic-equal-magnitude}.  These roots
satisfy the anchor equations.  By
\eqref{eq:purification-root-covariance} and block orthogonality on $T$,
\[
  0=\sum_{c\in T}a'_{\mathbf x,c}\overline{a'_{\mathbf y,c}}
   =a'_{\mathbf x,c_T}\overline{a'_{\mathbf y,c_T}}
    \sum_{c\in T}\rho_{T,c}\overline{\sigma_{T,c}}.
\]
The prefactor is nonzero, so
\eqref{eq:bo-certificate-root-sum} holds.  Thus the choices form a retained
certificate $\mathfrak c$ with $A\in V_K(E_{\mathfrak c})$.

For the reverse inclusion, fix a certificate $\mathfrak c$, let
$A\in V_K(E_{\mathfrak c})$, and let $A'$ be the matrix of a singleton
purification of the corresponding function $F$.  We verify block-rank one
first, then the two row-comparison modes of \cref{def:bo-certificate}.

The support equations place every nonzero row in a unique certified row
set $R_i$, and its
support is contained in $C_i$.  Rows belonging to different certified
rectangles have disjoint supports because the corresponding column sets are
disjoint.

Consider two nonzero rows $\mathbf x,\mathbf y\in R_i$.  If their supports
are disjoint, they already satisfy the block-rank-one alternative.  Suppose
instead that they overlap, and choose $e\in C_i$ with
$a_{\mathbf x,e}a_{\mathbf y,e}\ne0$.  For every $c\in C_i\setminus\{e\}$,
the twisted minor gives
\[
  a_{\mathbf x,c}a_{\mathbf y,e}
  =\tau_{\mathbf x,\mathbf y;c,e}
   a_{\mathbf x,e}a_{\mathbf y,c}.
\]
Consequently,
$a_{\mathbf x,c}\ne0$ if and only if $a_{\mathbf y,c}\ne0$.
Thus overlapping rows have the same support.  On that common support,
\eqref{eq:intrinsic-magnitude-minor} applied to the same twisted equations
gives
\[
  \abs{a'_{\mathbf x,c}}\abs{a'_{\mathbf y,e}}
  =\abs{a'_{\mathbf x,e}}\abs{a'_{\mathbf y,c}},
\]
so their purified magnitude rows are positively proportional.  Hence
$\abs{A'}$ is block-rank one, including when a chosen rectangle has zero
support or its nonzero support splits into smaller rectangles.

For an overlapping nonzero pair in dependent mode,
\eqref{eq:bo-certificate-dependent} and
\eqref{eq:intrinsic-ordinary-minor} show that all purified ordinary minors
equal zero.  Since both rows are nonzero, they are linearly dependent.

Now consider an overlapping nonzero pair in orthogonal mode.  The anchor
equations make each certified group $T$ entirely nonzero or entirely zero
in each row.  We first show that the groups with nonzero entries are the
same for both rows.  If $T$ survives only in row
$\mathbf x$, if $U$ survives in both, and if $c\in T$ and $e\in U$, then the
twisted minor has a nonzero left product and a zero right product.  The
symmetric case is identical.  Since the two rows overlap, at least one
common surviving group exists. Hence every surviving group is common to
both rows.

For every common surviving group $T$,
\cref{eq:purification-root-covariance,eq:bo-certificate-root-sum} give
\begin{equation}
  \sum_{c\in T}a'_{\mathbf x,c}\overline{a'_{\mathbf y,c}}
  =a'_{\mathbf x,c_T}\overline{a'_{\mathbf y,c_T}}
    \sum_{c\in T}\rho_{T,c}\overline{\sigma_{T,c}}=0.
  \label{eq:certified-group-remains-orthogonal}
\end{equation}
It remains to pass from the certified groups to the actual magnitude
blocks.  The anchor equations give all purified entries of each row within
$T$ the same magnitude, so each surviving group lies wholly in one
magnitude block.  Distinct groups may have the same purified magnitude.
Because the two magnitude rows are proportional, every magnitude block is
a disjoint union of whole common surviving groups.  Summing
\eqref{eq:certified-group-remains-orthogonal} over those groups shows that
the Hermitian sum on the merged block is still zero.

Thus every overlapping nonzero pair is linearly dependent or
block-orthogonal.  Together with block-rank one, this proves that $A'$ is
block-orthogonal by \cref{def:block-orthogonal-function}, so
$A\in\mathcal B_k$.
\end{proof}

\begin{remark}[The role of the working field]
\label{rem:bo-certificates-k-points}
Equation~\eqref{eq:bo-certificate-union} holds in $K^{D^k}$.  Every
generated function is $K$-valued, as established in
\cref{subsec:intrinsic-purification-tests}, so this equality applies to
all generated functions of arity $k$.  Their root-of-unity ratios belong
to the finite group $\mu(K)$ used in the certificates.
\end{remark}

\subsection{Polynomial identities}

\Cref{thm:finite-bo-certificates} expresses membership in $\mathcal B_k$
as satisfaction of at least one certificate's equations.  We now express
the same property by the simultaneous equations $P(A)=0$ for $P$ in a
single finite set $Q_k$.  Each polynomial in $Q_k$ takes one equation from every
certificate and multiplies their left-hand sides.

Enumerate the certificates as $\mathfrak C_k=\{\mathfrak c_1,\ldots,\mathfrak c_s\}$ and write $E_j=E_{\mathfrak c_j}$, after deleting syntactically zero
polynomials.  If some $E_j$ is empty, then
$V_K(E_j)=K^{D^k}$ and $\mathcal B_k=K^{D^k}$. In that case define
$Q_k=\varnothing$.  Otherwise define the finite set
\begin{equation*}
  Q_k=
  \left\{
    \prod_{j\in[s]} e_j:
    (e_1,\ldots,e_s)\in E_1\times\cdots\times E_s
  \right\}
  \subseteq K[\mathbf X].
\end{equation*}

\begin{corollary}[Polynomial identities for the BO locus]
\label{cor:bo-product-identities}
For every $A\in K^{D^k}$,
\begin{equation}
  A\in\mathcal B_k
  \quad\Longleftrightarrow\quad
  P(A)=0\text{ for every }P\in Q_k.
  \label{eq:bo-locus-by-identities}
\end{equation}
The finite set $Q_k$ is effectively computable from $D$, $k$, and $\mu(K)$.
\end{corollary}

\begin{proof}
For the forward implication, suppose $A\in V_K(E_j)$ for some $j$.
Every polynomial $P\in Q_k$ contains a
factor from $E_j$, and that factor equals zero at $A$.  Hence every $P(A)$ is
zero.

For the converse, suppose that $A\notin V_K(E_j)$ for every $j$.  For each $j$,
choose $e_j\in E_j$ with $e_j(A)\ne0$.  Since $K$ is a field, $\left(\prod_{j\in[s]} e_j\right)(A)\ne0$. The resulting product belongs to $Q_k$, contradicting all these products being zero.
Together with \cref{thm:finite-bo-certificates}, this proves
\eqref{eq:bo-locus-by-identities}.

For effective construction, enumerate the certificates in
\cref{def:bo-certificate} and their sets $E_j$, then form the products
defining $Q_k$.  All choices range over finite explicitly represented sets.
\end{proof}

For fixed $k\ge2$, write $W_{\cF}^{(k)}$ for the $k$-ary members of $W_{\cF}$.
By \cref{thm:finite-bo-certificates,cor:bo-product-identities}, all these
functions have block-orthogonal singleton purifications exactly when every
polynomial in $Q_k$ equals zero on $W_{\cF}^{(k)}$.  The next section gives
an exact test for such universal identities.

\section{Universal Identities on \texorpdfstring{$W_{\cF}$}{W(F)}}
\label{sec:universal-identities}

The polynomial conditions from \cref{sec:bo-certificates} must hold on
every generated function of a fixed arity, with no bound on the size of
its representing instance.  We give a finite test for these identities.
Labelled gluing makes each monomial a
multiplicative character on instances.  We compare these characters by
representing them as pinned partition values and applying an isomorphism
theorem.  Linear independence then reduces each identity to finitely many
coefficient sums.  The generated values lie in the field $K$
from~\eqref{eq:working-field}.

\subsection{Labelled instances and characters}
\label{subsec:labelled-instances}

We represent generated functions by instances with labelled retained
variables, then construct a product of instances that realizes pointwise
multiplication of these functions.

Fix $k\ge1$.  A $k$-labelled $\cF$-instance is a finite $\cF$-instance
with $k$ distinct variables carrying the labels $1,\ldots,k$.
Unlabelled isolated variables are permitted.  For such an instance $I$,
define its pinned partition vector $\zvec I\in K^{D^k}$ by
\begin{equation}
  \zvec I(\mathbf a)
  =\sum_{\substack{\sigma:V(I)\to D\\
                    \sigma\text{ assigns }\mathbf a\text{ to the labels}}}
       \prod_{\substack{(f;v_1,\ldots,v_{r_f})\\
                         \in\mathcal C_I}}
          f\bigl(\sigma(v_1),\ldots,\sigma(v_{r_f})\bigr).
  \label{eq:pinned-partition-vector}
\end{equation}
When the labelled variables are ordered first,
$\zvec I(\mathbf a)=F_I^{[k]}(\mathbf a)$.  Thus $\zvec I$ records the
values of the generated function $F_I^{[k]}$, indexed by the assignments
$\mathbf a\in D^k$ to the labels.

To multiply two generated functions at the same input, we glue their
representing instances along the labelled variables.  For $k$-labelled
instances $I$ and $J$, form $IJ$ by taking their disjoint union and
identifying variables with the same label.  All unlabelled variables
remain separate, and all constraint occurrences are retained.
The common labels specify the same input tuple in both instances.
\Cref{lem:pinned-realization} verifies that this construction multiplies
the corresponding function values.

Up to label-preserving isomorphism, this product is associative and
commutative.  Its identity is the instance consisting of the $k$ labelled
variables and no constraints.  Thus the label-preserving isomorphism
classes form a \emph{commutative monoid}, denoted $\InstMonoid{k}$,
with multiplication $IJ$.
All monoids below are commutative.  Unless displayed explicitly in
additive notation, they are written multiplicatively.

\begin{lemma}
\label{lem:pinned-realization}
Let $W_{\cF}^{(k)}$ be the set of $k$-ary members of $W_{\cF}$.  Then
\[
  \{\zvec I:I\in\InstMonoid{k}\}=W_{\cF}^{(k)}.
\]
Moreover, for $I,J\in\InstMonoid{k}$, we have
\[
  \zvec{IJ}=\zvec I\odot\zvec J,
\]
where $\odot$ is entrywise multiplication.
\end{lemma}

\begin{proof}
For the realization statement, take $F=F_I^{[k]}\in W_{\cF}^{(k)}$ and label
the retained variables of its representing instance.  The vector in
\eqref{eq:pinned-partition-vector} is then $F$.  Conversely, order the
labelled variables of a labelled instance first and sum all unlabelled
variables to obtain a member of $W_{\cF}^{(k)}$.

For the product statement, the unlabelled variables of the two
factors remain independent after the common labels are pinned, which gives
the entrywise product formula.
\end{proof}

For each fixed label assignment $\mathbf a$, the function
$I\mapsto\zvec I(\mathbf a)$ is multiplicative by
\cref{lem:pinned-realization} and equals $1$ on the unit instance.
A unital multiplicative \emph{character} from a monoid $M$ to a field $L$ is
a map
\[
  \chi:M\to L,
  \qquad
  \chi(1)=1,
  \qquad
  \chi(xy)=\chi(x)\chi(y).
\]
Thus each map $I\mapsto\zvec I(\mathbf a)$ is a character on
$\InstMonoid{k}$.  Characters may take the value zero.  Pointwise products
of characters are characters, as are complex conjugates of complex-valued
characters.  Consequently, monomials in the values $\zvec I(\mathbf a)$
and their conjugates define characters on $\InstMonoid{k}$, with the
input tuples and exponents fixed by the monomial.
Equality of characters means agreement on every element of the monoid,
so different monomial expressions can define the same character.

\subsection{Pinned constraint-function isomorphism}
\label{subsec:pinned-isomorphism}

To compare monomial characters, we will represent them as pinned partition
values of tensorized constraint functions.  \Cref{thm:young-pinned}
replaces agreement on all labelled instances by the existence of a domain
isomorphism satisfying conditions on the pins.  Over finite domains with
exactly represented algebraic values, this becomes a finite search over
bijections.

We use the following terminology from
Young~\cite[Definitions~3, 5--7 and Section~4.1]{Young2025}.  A constraint-function
family $\mathcal H$ is defined over a common finite domain, denoted
$D_{\mathcal H}$.  Two finite families $\mathcal H$ and $\mathcal G$ are
\emph{similar} if they are equipped with an arity-preserving bijection between
their functions.  A domain bijection $\sigma:D_{\mathcal H}\to D_{\mathcal G}$
is an \emph{isomorphism} of similar families if, for every corresponding pair
$h\leftrightarrow g$ of arity $r$,
\[
  h(a_1,\ldots,a_r)=g(\sigma(a_1),\ldots,\sigma(a_r))
  \qquad(a_1,\ldots,a_r\in D_{\mathcal H}).
\]
Elements $u,v\in D_{\mathcal H}$ are \emph{twins in $\mathcal H$} if,
for every $h\in\mathcal H$ of arity $r$, every $j\in[r]$, and every
$(a_1,\ldots,a_r)\in D_{\mathcal H}^r$,
\begin{equation*}
 h(a_1,\ldots,a_{j-1},u,a_{j+1},\ldots,a_r)
 =h(a_1,\ldots,a_{j-1},v,a_{j+1},\ldots,a_r).
\end{equation*}
For a $k$-labelled $\mathcal H$-instance $I$ and a pin
$\beta:[k]\to D_{\mathcal H}$, write $Z_{\beta}^{\mathcal H}(I)$ for its
pinned partition value.  In particular, for the original family $\cF$
and the pin $\beta(i)=a_i$, we have $Z_\beta^{\cF}(I)=\zvec I(\mathbf a)$.
The superscript records the function family used to evaluate the instance.
If $\mathcal H$ and $\mathcal G$ are similar, let
$I_{\mathcal H\to\mathcal G}$ be the $\mathcal G$-instance obtained by
replacing every function of $\mathcal H$ by its corresponding function of
$\mathcal G$, while retaining the variables, labels, and constraint scopes.
The following theorem characterizes equality of the two pinned partition
functions on all labelled instances.

\begin{theorem}[Pinned constraint-function isomorphism]
\label{thm:young-pinned}
Let $\mathcal H$ and $\mathcal G$ be similar finite sets of constraint
functions over a field of characteristic zero.  Fix $k\ge1$ and pins
$\beta:[k]\to D_{\mathcal H}$ and $\gamma:[k]\to D_{\mathcal G}$.
If
\begin{equation*}
  Z_{\beta}^{\mathcal H}(I)
  =Z_{\gamma}^{\mathcal G}(I_{\mathcal H\to\mathcal G})
\end{equation*}
for every $k$-labelled $\mathcal H$-instance $I$, then the two
domains have the same size and there is an isomorphism
$\sigma:D_{\mathcal H}\to D_{\mathcal G}$ of the two families such that
$\sigma(\beta(i))$ and $\gamma(i)$ are twins in $\mathcal G$ for every
$i\in[k]$.

Conversely, such a domain bijection and twin replacements preserve pinned
partition values on every $k$-labelled instance.
\end{theorem}

\begin{proof}
Let $L$ be the coefficient field.  For any $\ell\ge k$, call two pin tuples
$a:[\ell]\to D_{\mathcal H}$ and $b:[\ell]\to D_{\mathcal G}$ \emph{indistinguishable} if their
pinned partition values agree on every finite $\ell$-labelled instance,
with the functions retargeted from $\mathcal H$ to $\mathcal G$ on the
second side.  The hypothesis says that $\beta$ and $\gamma$ are
indistinguishable.  We prove that pins can be extended while preserving
this property, then use extensions covering both domains to recover the
function values and the sizes of twin classes.

We show that a prescribed new pin on either side can be matched on the
other side while preserving indistinguishability.
Suppose $a$ and $b$ are indistinguishable.  For $x\in D_{\mathcal H}$ and $y\in D_{\mathcal G}$,
define their extension profiles on finite $(\ell+1)$-labelled
$\mathcal H$-instances $J$ by
\[
  p_x(J)=Z^{\mathcal H}_{(a,x)}(J),
  \qquad
  \widetilde p_y(J)=Z^{\mathcal G}_{(b,y)}(J_{\mathcal H\to\mathcal G}).
\]
Here $(a,x)$ extends $a$ by pinning the new label to $x$.  Let
$\mathcal P=\{p_x:x\in D_{\mathcal H}\}\cup\{\widetilde p_y:y\in D_{\mathcal G}\}$ be the finite set of
\emph{distinct} profiles. Each profile is a function on all such $J$.
For a fixed $\rho\in\mathcal P$ and every
$\tau\in\mathcal P\setminus\{\rho\}$, choose an instance
$J_{\rho,\tau}$ for which
$\rho(J_{\rho,\tau})\ne\tau(J_{\rho,\tau})$.  Such an instance exists
because the profiles differ as functions.  On the profiles in $\mathcal P$,
the polynomial expression
\[
  Q_\rho(\eta)
  =\prod_{\tau\in\mathcal P\setminus\{\rho\}}
    \frac{\eta(J_{\rho,\tau})-\tau(J_{\rho,\tau})}
         {\rho(J_{\rho,\tau})-\tau(J_{\rho,\tau})}
\]
is $1$ when $\eta=\rho$ and $0$ otherwise.  All denominators are nonzero.
If $\mathcal P=\{\rho\}$, the empty product is $1$.

For a fixed pin tuple, labelled gluing multiplies partition values, and
the constraint-free instance with just the labelled variables has value
$1$.  Expanding $Q_\rho$ therefore expresses it as a finite $L$-linear
combination of partition values of $(\ell+1)$-labelled instances formed
by gluing.  The coefficients multiply the resulting values, and each
instance uses the original constraint family.
If $J^-$ is obtained from $J$ by forgetting the last label while retaining
its variable, indistinguishability gives
\[
  \sum_{x\in D_{\mathcal H}}p_x(J)
  =Z_a^{\mathcal H}(J^-)
  =Z_b^{\mathcal G}((J^-)_{\mathcal H\to\mathcal G})
  =\sum_{y\in D_{\mathcal G}}\widetilde p_y(J).
\]
Applying this equality to the linear combination defining $Q_\rho$ yields
\[
  \bigl|\{x\in D_{\mathcal H}:p_x=\rho\}\bigr|\,1_L
  =\sum_{x\in D_{\mathcal H}}Q_\rho(p_x)
  =\sum_{y\in D_{\mathcal G}}Q_\rho(\widetilde p_y)
  =\bigl|\{y\in D_{\mathcal G}:\widetilde p_y=\rho\}\bigr|\,1_L.
\]
Since $L$ has characteristic zero, the two cardinalities are equal as
integers.  In particular, for every prescribed $x\in D_{\mathcal H}$ there is a
$y\in D_{\mathcal G}$ with $p_x=\widetilde p_y$: the extended tuples $(a,x)$ and $(b,y)$ remain
indistinguishable.  The same argument works in the reverse direction.

We use the extension step to make every domain element accessible through
a label, then recover function values with single-constraint instances.
Starting with $(\beta,\gamma)$, apply the extension step once for each
element of $D_{\mathcal H}$.  Then apply it in the reverse direction once for each
element of $D_{\mathcal G}$.  We obtain indistinguishable tuples
$a:[\ell]\to D_{\mathcal H}$ and $b:[\ell]\to D_{\mathcal G}$ that extend the original pins and
are both surjective.  Each extension adds a distinct labelled variable,
whose pinned value may repeat a value already used.
For corresponding functions $h\leftrightarrow g$ of arity $r$, the
instance consisting of the $\ell$ labelled variables and the single
constraint $(h;i_1,\ldots,i_r)$ gives
\begin{equation}
  h(a(i_1),\ldots,a(i_r))
  =g(b(i_1),\ldots,b(i_r))
  \qquad(i_1,\ldots,i_r\in[\ell]).
  \label{eq:pinned-atomic-diagram}
\end{equation}
Thus corresponding constraints have equal values on the paired pins,
including when scope positions repeat.

We first match the twin classes by their function values, then use
interpolation to show that corresponding classes have equal size.
For $x\in D_{\mathcal H}$, form the finite vector $\theta_{\mathcal H}(x)$ whose entries are
\[
  h\bigl(a(i_1),\ldots,a(i_{j-1}),x,
          a(i_{j+1}),\ldots,a(i_r)\bigr),
\]
indexed by $h\in\mathcal H$, its positions $j\in[r]$, and all choices
of the other $r-1$ indices in $[\ell]$.  Define $\theta_{\mathcal G}(y)$ with the
same index set, replacing $h$ and $a$ by $g$ and $b$.
Surjectivity of $a$ implies that
$\theta_{\mathcal H}(x)=\theta_{\mathcal H}(x')$ exactly when $x$ and $x'$ are twins in
$\mathcal H$: the entries test every function position with every
possible assignment to the other positions.  The corresponding statement
holds for $\theta_{\mathcal G}$ because $b$ is surjective.
Equation~\eqref{eq:pinned-atomic-diagram} gives
$\theta_{\mathcal H}(a(i))=\theta_{\mathcal G}(b(i))$ for every $i$.  Hence the two maps
have the same finite range, and their common values specify a bijection
between the twin classes of the two families.

To obtain a bijection between domain elements, we next show that
corresponding twin classes have equal cardinality.
For each vector $u$ in this common range, choose, for every
other vector $v$, an entry on which $u$ and $v$ differ.  The same
product of normalized linear factors used to extend pins gives a polynomial
$R_u$ that is $1$ at $u$ and $0$ at every other vector in the range.
Each entry of $\theta_{\mathcal H}(x)$, with its corresponding entry of
$\theta_{\mathcal G}(y)$, is the partition value of a single-constraint instance
whose new label occupies position $j$ and whose other positions are old
labels.  Labelled gluing realizes products of these entries.  After
expanding $R_u$ and forgetting the new label, indistinguishability of
$a,b$ therefore gives
\[
  \sum_{x\in D_{\mathcal H}}R_u(\theta_{\mathcal H}(x))
  =\sum_{y\in D_{\mathcal G}}R_u(\theta_{\mathcal G}(y)).
\]
Both sides count the elements of the corresponding twin class in $L$.
Characteristic zero again makes their integer cardinalities equal.  This
also handles an empty constraint family: the vectors are empty,
there is one class on each side, and the constant polynomial $1$ counts
the two domains by means of one isolated unlabelled variable.

Choose a bijection between each pair of corresponding twin classes.
Their union is a bijection $\sigma:D_{\mathcal H}\to D_{\mathcal G}$, and
$\sigma(a(i))$ is a twin of $b(i)$ for every $i\in[\ell]$.
To check preservation of function values, take $x_1,\ldots,x_r\in D_{\mathcal H}$ and choose
$i_1,\ldots,i_r$ with $a(i_j)=x_j$.  By
\eqref{eq:pinned-atomic-diagram} and successive replacement of input values
by twins,
\[
  h(x_1,\ldots,x_r)
  =g(b(i_1),\ldots,b(i_r))
  =g(\sigma(x_1),\ldots,\sigma(x_r)).
\]
Thus $\sigma$ is an isomorphism, and its restriction to the original pins
satisfies the asserted twin condition.

For the converse, an isomorphism $\sigma$ bijectively transports assignments
with pins $\beta$ to assignments with pins $\sigma\circ\beta$, preserving
every constraint value.  Replacing these pins by the prescribed twins
$\gamma$ also preserves every constraint value, one input position at a time,
including when a labelled variable occurs repeatedly in a scope.  Hence
\[
  Z_\beta^{\mathcal H}(I)
  =Z_{\sigma\circ\beta}^{\mathcal G}(I_{\mathcal H\to\mathcal G})
  =Z_\gamma^{\mathcal G}(I_{\mathcal H\to\mathcal G})
\]
for every finite labelled instance, as required.
\end{proof}

For exactly represented algebraic values, \cref{thm:young-pinned} gives a
finite decision procedure for equality on all labelled instances.  If the domains
have different sizes, equality fails.  Otherwise, enumerate their finitely
many bijections and test preservation of all function values and the
prescribed twin relations exactly.  Equality holds precisely when one
bijection passes these tests.

\subsection{Tensor realization and character comparison}
\label{subsec:tensor-comparison}

We first construct an instance whose pinned partition value is a given
monomial in generated function values and their conjugates.  We then
compare the constructions for two monomials using
\cref{thm:young-pinned}.

Fix $k\ge1$.  For each assignment to the labels
$\mathbf a=(a_1,\ldots,a_k)\in D^k$, introduce two independent formal
variables $X_{\mathbf a}$ and $Y_{\mathbf a}$.  When a polynomial in these
variables is evaluated at the pinned partition vector of an instance $I$,
the substitutions are
\begin{equation}
  X_{\mathbf a}\longmapsto\zvec I(\mathbf a),
  \qquad
  Y_{\mathbf a}\longmapsto\overline{\zvec I(\mathbf a)}.
  \label{eq:formal-coordinate-substitution}
\end{equation}
Let $p_+,p_-\ge0$.  A monomial of $X$-degree $p_+$ and $Y$-degree
$p_-$ can be written as
\begin{equation*}
  M=\prod_{i\in[p_+]}X_{\mathbf a^{(i)}}
    \prod_{j\in[p_-]}Y_{\mathbf b^{(j)}}.
\end{equation*}
For $I\in\InstMonoid{k}$, define
\begin{equation*}
  \chi_M(I):=M(\zvec I,\overline{\zvec I}),
\end{equation*}
using the simultaneous substitution in
\eqref{eq:formal-coordinate-substitution}.  The input tuples and their
multiplicities are fixed by $M$.  As $I$ varies, these evaluations define
the function $\chi_M$ on $\InstMonoid{k}$.  Thus $\chi_M(I)$ is a
product of $p_+$ generated function values and $p_-$ conjugated values,
counted with multiplicity.
The possibly repeated indices
\[
\begin{aligned}
  \mathbf a^{(i)}&=(a_1^{(i)},\ldots,a_k^{(i)})\in D^k
    &&(i\in[p_+]),\\
  \mathbf b^{(j)}&=(b_1^{(j)},\ldots,b_k^{(j)})\in D^k
    &&(j\in[p_-])
\end{aligned}
\]
list the factors of $M$.  The parenthesized superscript numbers a factor:
$a_t^{(i)}$ is the entry in position $t$ of the tuple indexing
the $i$th $X$-factor.

To represent this product by one partition value, let each instance
variable take a tuple in $D^{p_++p_-}$.  The tuple records the values
assigned to that variable for the ordinary factors, followed by those
for the conjugated factors.  For each constraint, multiply
the evaluations of its original function on corresponding entries of
these tuples, conjugating the last $p_-$ evaluations.
For each original constraint function $f:D^{r_f}\to K$, define
$T_{p_+,p_-}(f):(D^{p_++p_-})^{r_f}\to K$ by
\begin{equation*}
  T_{p_+,p_-}(f)(\mathbf u_1,\ldots,\mathbf u_{r_f})
  =\prod_{\ell\in[p_+]}f(u_{1,\ell},\ldots,u_{r_f,\ell})
   \cdot\prod_{\ell=p_++1}^{p_++p_-}
      \overline{f(u_{1,\ell},\ldots,u_{r_f,\ell})},
\end{equation*}
where $\mathbf u_j\in D^{p_++p_-}$.  When $p_+=p_-=0$, the domain is the
singleton $D^0$ and $T_{0,0}(f)$ is the constant-one function of arity
$r_f$ on that domain.  Write $T_{p_+,p_-}(\cF)=\bigl(T_{p_+,p_-}(f):f\in\cF\bigr)$ for this $f$-indexed tensorized family.

The degrees determine the transformed constraint functions.  The
specific input tuples occurring in $M$ determine the pins: at each label,
we collect its prescribed values from all the factors into one tuple.

\begin{definition}[Tensor realization]
\label{def:tensor-realization}
The pin $\beta_M:[k]\to D^{p_++p_-}$ is obtained by transposing these $p_++p_-$
assignments to the labels.  For label $t\in[k]$, set
\[
  \beta_M(t)
  =\bigl(a_t^{(1)},\ldots,a_t^{(p_+)},
          b_t^{(1)},\ldots,b_t^{(p_-)}\bigr)\in D^{p_++p_-}.
\]
Equivalently, regarding the values $\beta_M(t)$ as columns,
\begin{equation*}
  \bigl(\beta_M(1)\ \cdots\ \beta_M(k)\bigr)
  =
  \begin{pmatrix}
    a_1^{(1)}&\cdots&a_k^{(1)}\\
    \vdots&&\vdots\\
    a_1^{(p_+)}&\cdots&a_k^{(p_+)}\\
    b_1^{(1)}&\cdots&b_k^{(1)}\\
    \vdots&&\vdots\\
    b_1^{(p_-)}&\cdots&b_k^{(p_-)}
  \end{pmatrix}.
\end{equation*}
The block of $\mathbf a$-rows or $\mathbf b$-rows is omitted when $p_+=0$ or
$p_-=0$, respectively.
For $I\in\InstMonoid{k}$, let $I^M$ be the $k$-labelled
$T_{p_+,p_-}(\cF)$-instance obtained by replacing each occurrence of
$f\in\cF$ by $T_{p_+,p_-}(f)$ while retaining all variables, labels, and scopes.
\end{definition}

The next lemma verifies that these transformed constraints and pins
realize the product defining $\chi_M(I)$.
Appendix~\ref{app:tensor-realization-example} illustrates the construction.

\begin{lemma}
\label{lem:tensor-character}
For $I\in\InstMonoid{k}$, the instance $I^M$ from
\cref{def:tensor-realization} satisfies
\begin{equation*}
  Z_{\beta_M}^{T_{p_+,p_-}(\cF)}(I^M)
  =\chi_M(I)
  =\prod_{i\in[p_+]}\zvec I(\mathbf a^{(i)})
    \prod_{j\in[p_-]}\overline{\zvec I(\mathbf b^{(j)})}.
\end{equation*}
Thus $\chi_M:\InstMonoid{k}\to K$ is a unital multiplicative character.
\end{lemma}

\begin{proof}

An assignment $\eta:V(I)\to D^{p_++p_-}$ is uniquely a tuple
$(\sigma_1,\ldots,\sigma_{p_++p_-})$ of assignments
$\sigma_\ell:V(I)\to D$.  By the definition of $\beta_M$ in
\cref{def:tensor-realization}, the assignment
$\sigma_i$ pins the $k$ labels to $\mathbf a^{(i)}$ for $i\in[p_+]$, while
$\sigma_{p_++j}$ pins them to $\mathbf b^{(j)}$ for $j\in[p_-]$.
For each component assignment $\sigma_\ell$, the product of its
constraint weights is the instance value $F_I(\sigma_\ell)$ from
\eqref{eq:instance-function}.  Expanding the tensorized constraint
functions therefore gives
\begin{align*}
  Z_{\beta_M}^{T_{p_+,p_-}(\cF)}(I^M)
  &=\sum_{\sigma_1,\ldots,\sigma_{p_++p_-}}
      \prod_{i\in[p_+]} F_I(\sigma_i)
      \prod_{j\in[p_-]}\overline{F_I(\sigma_{p_++j})}\\
  &=\prod_{i\in[p_+]}\zvec I(\mathbf a^{(i)})
    \prod_{j\in[p_-]}\overline{\zvec I(\mathbf b^{(j)})}.
\end{align*}
The first sum ranges over the component assignments with the pins just
specified.  The second equality holds because these component assignments
are independent, so the finite sum factors. Finite sums commute with complex
conjugation.  When $p_+=p_-=0$, the domain is $D^0$ and the value is the empty
product $1$.

Finally, \cref{lem:pinned-realization} gives
$\chi_M(IJ)=\chi_M(I)\chi_M(J)$ and $\chi_M(1)=1$, so $\chi_M$ is a unital
multiplicative character.
\end{proof}

Now consider a second monomial $N$ of $(X,Y)$-degrees $(r_+,r_-)$.
For each original function $f\in\cF$, write
\[
  h_f:=T_{p_+,p_-}(f),\qquad
  g_f:=T_{r_+,r_-}(f).
\]
These are the transformed versions of $f$ for $M$ and $N$, respectively.
Both have arity $r_f$, on the respective domains
$D_M=D^{p_++p_-}$ and $D_N=D^{r_++r_-}$.
Use \cref{def:tensor-realization} to construct the pins $\beta_M,\beta_N$
and the instances $I^M,I^N$ by replacing each occurrence of $f$ with
$h_f,g_f$, respectively.  By \cref{lem:tensor-character}, their pinned
partition values are $\chi_M(I)$ and $\chi_N(I)$.

To apply \cref{thm:young-pinned}, the correspondence $h_f\leftrightarrow g_f$
must give a bijection between the two sets of functions.
Tensorization can identify distinct original functions:
for example, $T_{1,1}(\mathrm{i}f)=T_{1,1}(f)$, where $\mathrm{i}^2=-1$.
We therefore multiply the functions by positive rational scalars to make
the functions within each family distinct at each arity.  Choosing the same scalar
$\lambda_f$ for $h_f$ and $g_f$ multiplies both partition values by a
common nonzero factor, preserving their equality.
The next lemma makes this simultaneous choice for two indexed families.

\begin{lemma}
\label{lem:scalar-tags}
Let $(h_f:f\in\cF)$ and $(g_f:f\in\cF)$ be indexed families over finite
domains $D_M$ and $D_N$, respectively, where
$h_f:D_M^{r_f}\to\AlgNums$ and $g_f:D_N^{r_f}\to\AlgNums$ are given by their
exact values.  Suppose that, within
each family and each arity, at most one indexed function is identically zero.
One can effectively choose positive rational numbers $\lambda_f$, $f\in\cF$,
such that the functions within each arity are pairwise distinct in both
indexed families $(\lambda_f h_f:f\in\cF)$ and
$(\lambda_f g_f:f\in\cF)$.
\end{lemma}

\begin{proof}
Choose the tags sequentially, avoiding collisions with functions already
treated.  For a previously treated index $e$ of the
same arity, the equality
$\lambda h_f=\lambda_e h_e$ excludes at most one positive rational value of
$\lambda$, since the hypothesis ensures that at least one of $h_f,h_e$
is nonzero.  The same argument applies to the $g$-family.  Thus only
finitely many rational values are excluded at each stage.  Enumerating
the positive integers and testing the resulting function values exactly
finds a valid tag.
\end{proof}

The tensorized families associated with monomial characters satisfy the
hypothesis of \cref{lem:scalar-tags}.  At positive tensor degree,
$T_{p_+,p_-}(f)$ is zero exactly when $f$ is zero, and exact deduplication of the
input leaves at most one zero function of each arity.  At degree zero every
function is the nonzero constant-one function.

For the two tensorized families, equality of their pinned partition
values is exactly the desired equality $\chi_M=\chi_N$ by
\cref{lem:tensor-character}.  The next lemma decides this equality using
the common scalar tags and \cref{thm:young-pinned}.
We state the comparison for general indexed families so that it also
applies to the degree-restricted construction in \cref{sec:degree-extension}.

\begin{lemma}
\label{lem:character-comparison}
Let $\mathcal H_M=(h_f:f\in\cF)$ and $\mathcal H_N=(g_f:f\in\cF)$ be
finite indexed families of exactly represented algebraic-valued functions
on finite nonempty domains $D_M$ and $D_N$.  For each $f$, both $h_f$ and
$g_f$ have arity $r_f$.  Assume each family contains at most one zero
function of each arity.  Fix $k\ge1$ and pins $\beta_M:[k]\to D_M$ and
$\beta_N:[k]\to D_N$.
For an original $\cF$-instance $I$, form $I^M$ and $I^N$ by replacing
each occurrence of $f$ by $h_f$ and $g_f$, respectively.  It is decidable
whether
\[
 Z_{\beta_M}^{\mathcal H_M}(I^M)=Z_{\beta_N}^{\mathcal H_N}(I^N)
 \quad\text{for every }I\in\InstMonoid{k}.
\]
In particular, equality of two monomial characters on $\InstMonoid{k}$
is decidable.
\end{lemma}

\begin{proof}
Abbreviate the two pinned partition values by $\chi_M(I)$ and $\chi_N(I)$:
\begin{equation}
 Z_{\beta_M}^{\mathcal H_M}(I^M)=\chi_M(I),\qquad
 Z_{\beta_N}^{\mathcal H_N}(I^N)=\chi_N(I).
 \label{eq:two-monomial-realizations}
\end{equation}
Thus $\chi_M$ and $\chi_N$ are functions on the same original
labelled-instance monoid $\InstMonoid{k}$.

Apply
\cref{lem:scalar-tags} to choose common positive rational tags $\lambda_f$
that make both indexed families duplicate-free, and put
\[
\begin{aligned}
  \widehat h_f&=\lambda_fh_f,
  &\widehat g_f&=\lambda_fg_f,\\
  \widehat{\mathcal H}_M&=(\widehat h_f:f\in\cF),
  &\widehat{\mathcal H}_N&=(\widehat g_f:f\in\cF).
\end{aligned}
\]
If $N_f(I)$ is
the number of occurrences of $f$ in the original instance $I$, interpreting
that instance with the tagged functions multiplies both values in
\eqref{eq:two-monomial-realizations} by the same factor
\begin{equation*}
  \tau(I)=\prod_{f\in\cF}\lambda_f^{N_f(I)}.
\end{equation*}
This factor is nonzero.  Denoting the two tagged interpretations by
$\widehat I^M$ and $\widehat I^N$, define
\begin{align*}
  \widehat Z_M(I)
  &:=Z_{\beta_M}^{\widehat{\mathcal H}_M}(\widehat I^M)
    =\tau(I)\chi_M(I),\\
  \widehat Z_N(I)
  &:=Z_{\beta_N}^{\widehat{\mathcal H}_N}(\widehat I^N)
    =\tau(I)\chi_N(I).
\end{align*}
Since the same nonzero factor occurs on both sides, we have
\begin{equation}
  \bigl(\forall I\in\InstMonoid{k},\ \chi_M(I)=\chi_N(I)\bigr)
  \ \Longleftrightarrow\ 
  \bigl(\forall I\in\InstMonoid{k},\ \widehat Z_M(I)=\widehat Z_N(I)\bigr).
  \label{eq:character-agreement-tagged-agreement}
\end{equation}

The tagged families are sets of size $\abs{\cF}$, and
$\widehat h_f\longleftrightarrow\widehat g_f \quad(f\in\cF)$ is an
arity-preserving bijection.  They are therefore similar in the sense of
\cref{subsec:pinned-isomorphism}.  Because the
functions are duplicate-free, every labelled
$\widehat{\mathcal H}_M$-instance corresponds uniquely to an original
$\cF$-instance, and replacing $\widehat h_f$ by
$\widehat g_f$ produces its $\widehat{\mathcal H}_N$ interpretation.

If the right-hand side of
\eqref{eq:character-agreement-tagged-agreement} holds, then by
\cref{thm:young-pinned}, the
domains have equal size and there is a bijection
\begin{equation}
  \sigma:D_M\to D_N
  \label{eq:character-comparison-domain-bijection}
\end{equation}
that is an isomorphism of the two tagged families and satisfies
\begin{equation}
  \sigma(\beta_M(t))
  \text{ and }\beta_N(t)
  \text{ are twins in }\widehat{\mathcal H}_N
  \qquad(t\in[k]).
  \label{eq:character-comparison-pin-twins}
\end{equation}
Conversely, such a bijection maps every assignment of
$\widehat I^M$ to an assignment of $\widehat I^N$, preserves each
constraint weight, and preserves the pins through the twin conditions.
The converse direction of \cref{thm:young-pinned} therefore gives equality of
the two tagged pinned values
for every $I$.  Since $\tau(I)\ne0$, division by the common factor gives
$\chi_M(I)=\chi_N(I)$ for every $I\in\InstMonoid{k}$.

If $\abs{D_M}\ne\abs{D_N}$, no
bijection in \eqref{eq:character-comparison-domain-bijection} exists and
the characters are different.  Otherwise enumerate every bijection
$\sigma:D_M\to D_N$.  For each one, test exactly whether all corresponding
tagged functions have equal values under $\sigma$ and whether
\eqref{eq:character-comparison-pin-twins} holds.  The characters agree if
and only if some bijection passes.  All domains and sets of functions are finite, so
the enumeration terminates.

Finally, apply the criterion to the families $T_{p_+,p_-}(\cF)$ and
$T_{r_+,r_-}(\cF)$ determined by $M,N$, with the pins from
\cref{def:tensor-realization}.  \Cref{lem:tensor-character} identifies
their pinned partition values with $\chi_M(I)$ and $\chi_N(I)$,
respectively.  These families satisfy the zero-function hypothesis of
\cref{lem:scalar-tags}, so the criterion decides equality of the two
monomial characters.
\end{proof}

\subsection{The universal identity oracle}
\label{subsec:identity-oracle}

Equality of the monomial characters is now decidable by
\cref{lem:character-comparison}.  We group monomials that define the same
character and add their coefficients.  To obtain an identity test, it
remains to show that distinct
characters cannot cancel in a nontrivial linear combination.  The next
lemma supplies this fact, which \cref{thm:universal-identity} turns into
finite coefficient tests.

\begin{lemma}
\label{lem:dedekind-independence}
Let $M$ be a monoid, let $L$ be a field of characteristic zero, and let $\chi_1,\ldots,\chi_s:M\to L$ be pairwise distinct unital multiplicative characters.  If
$c_1,\ldots,c_s\in L$ satisfy
\begin{equation}
  \sum_{i\in[s]} c_i\chi_i(x)=0
  \qquad\text{for every }x\in M,
  \label{eq:dedekind-assumed-relation}
\end{equation}
then $c_1=\cdots=c_s=0$.  Equivalently,
$\chi_1,\ldots,\chi_s$ are linearly independent in the vector space $L^M$.
\end{lemma}

\begin{proof}

Suppose the conclusion is false, and choose a nontrivial relation of the
form~\eqref{eq:dedekind-assumed-relation} involving the minimum possible
number $s$ of characters.  After deleting zero terms and relabelling, we may
assume
\begin{equation}
  c_i\ne0
  \qquad(i\in[s]).
  \label{eq:dedekind-nonzero-coefficients}
\end{equation}
If $s=1$, evaluation at the monoid identity gives $0=c_1\chi_1(1)=c_1$, contrary to \eqref{eq:dedekind-nonzero-coefficients}.  Hence $s\ge2$.

Because $\chi_1$ and $\chi_s$ are distinct functions on $M$, choose and fix
$u\in M$ such that
\begin{equation}
  \chi_1(u)\ne\chi_s(u).
  \label{eq:dedekind-fixed-witness}
\end{equation}
For an arbitrary $v\in M$, evaluate the assumed relation at $uv$ and use
multiplicativity:
\begin{equation}
  0=\sum_{i\in[s]}c_i\chi_i(uv)
   =\sum_{i\in[s]}c_i\chi_i(u)\chi_i(v).
  \label{eq:dedekind-relation-at-uv}
\end{equation}
Evaluation at $v$, followed by multiplication by the fixed scalar
$\chi_s(u)$, gives
\begin{equation}
  0
  =\chi_s(u)\sum_{i\in[s]} c_i\chi_i(v)
  =\sum_{i\in[s]} c_i\chi_s(u)\chi_i(v).
  \label{eq:dedekind-scaled-relation-at-v}
\end{equation}
Subtracting \eqref{eq:dedekind-scaled-relation-at-v} from
\eqref{eq:dedekind-relation-at-uv} yields
\begin{equation}
  0=\sum_{i\in[s-1]}
      c_i\bigl(\chi_i(u)-\chi_s(u)\bigr)\chi_i(v)
  \qquad\text{for every }v\in M.
  \label{eq:dedekind-shorter-relation}
\end{equation}

The term indexed by $s$ has coefficient zero.  With $u$ fixed, the
coefficients $b_i=c_i\bigl(\chi_i(u)-\chi_s(u)\bigr) \quad(i\in[s-1])$
lie in $L$ and are independent of $v$.  Moreover,
$b_1=c_1\bigl(\chi_1(u)-\chi_s(u)\bigr)\ne0$ by
\cref{eq:dedekind-nonzero-coefficients,eq:dedekind-fixed-witness}.
Consequently, \eqref{eq:dedekind-shorter-relation} is a nontrivial linear
relation involving at most $s-1$ of the characters, contradicting the
minimality of $s$.
\end{proof}

\begin{theorem}[Universal identity oracle]
\label{thm:universal-identity}
Fix $k\ge1$.  The input is an arbitrary commutative polynomial $P\in\AlgNums[X_{\mathbf a},Y_{\mathbf a}:\mathbf a\in D^k]$, given as a finite list of monomials with exactly encoded algebraic-complex
coefficients.  The $2q^k$ symbols
$X_{\mathbf a},Y_{\mathbf a}$ are independent formal variables.
For $I\in\InstMonoid{k}$, the notation
$P(\zvec I,\overline{\zvec I})$ means the simultaneous substitution $X_{\mathbf a}=\zvec I(\mathbf a), \quad Y_{\mathbf a}=\overline{\zvec I(\mathbf a)} \quad(\mathbf a\in D^k)$, as in \eqref{eq:formal-coordinate-substitution}.  There is a total exact
algorithm deciding whether
\begin{equation}
  P(\zvec I,\overline{\zvec I})=0
  \qquad\text{for every }I\in\InstMonoid{k}.
  \label{eq:universal-star-identity}
\end{equation}
\end{theorem}

\begin{proof}

Let $L$ be the number field generated over $K$ by the finitely many
coefficients of $P$.  It is effectively constructible from their exact
encodings.  Every monomial character takes values in $K\subseteq L$, so we
may regard all of them as $L$-valued characters.

Expand $P$ into finitely many monomials and combine identical monomials.
By \cref{lem:character-comparison}, we can partition these monomials into
classes, placing $M$ and $N$ in the same class exactly when
$\chi_M=\chi_N$.  Let $\mathcal E_{\mathrm{char}}$ be the finite set of
these classes, and let $\chi_E$ denote the common character defined by the
monomials in $E$.  Then, for every $I\in\InstMonoid{k}$,
\[
  P(\zvec I,\overline{\zvec I})
  =\sum_{E\in\mathcal E_{\mathrm{char}}}
     \left(\sum_{M\in E}c_M\right)\chi_E(I),
\]
where $c_M\in L$ is the coefficient of $M$.

The characters $\chi_E$ are
pairwise distinct, so \cref{lem:dedekind-independence} makes them linearly
independent over $L$.  Therefore~\eqref{eq:universal-star-identity} holds if
and only if
\[
  \sum_{M\in E}c_M=0
  \qquad\text{for every }E\in\mathcal E_{\mathrm{char}}.
\]
These are finitely many exact algebraic-number equality tests.

Constructing the tensorized functions, choosing scalar tags, enumerating
domain bijections, and testing preservation of function values, twins, and
coefficient sums all require only finitely many exact computations.  If
one class sum is nonzero, \cref{lem:dedekind-independence} says that the
displayed linear combination is not the zero function on $\InstMonoid{k}$.
Hence a finite labelled instance $I$ satisfies
$P(\zvec I,\overline{\zvec I})\ne0$, and the algorithm rejects.
\end{proof}

\section{The Global Block Orthogonality Decision Procedure}
\label{sec:bo-algorithm}

To decide the joint condition $\BO(\cF)$ in \cref{def:global-bo},
\cref{cor:singleton-purification,lem:bo-arity-bound} reduce the problem
to singleton purifications of generated functions of arity between two
and $q^2+1$.  For each such arity $k$,
\cref{cor:bo-product-identities} constructs a finite set
$Q_k\subseteq K[X_{\mathbf a}:\mathbf a\in D^k]$.
By~\eqref{eq:bo-locus-by-identities}, a $k$-ary generated function $F$
has a block-orthogonal singleton purification exactly when $P(F)=0$
for every $P\in Q_k$.

The remaining task is to test these equations on all generated functions
of each tested arity.  By~\cref{lem:pinned-realization}, these functions
are precisely the $\zvec I$ with $I\in\InstMonoid{k}$.
Each $P\in Q_k$ is a polynomial in the variables $X_{\mathbf a}$ alone,
so \cref{thm:universal-identity} decides whether $P(\zvec I)=0$ for
every such instance.  Thus finitely many identity tests decide the
condition on the entire set $W_{\cF}$.

\begin{procedure}[Global Block Orthogonality]
\label{proc:global-bo}
On input $q\ge1$ and a finite set $\cF$ of algebraic
complex-valued functions on $D$, perform the following steps.
\begin{enumerate}
  \item Identify functions with the same arity and the same values.  If $\cF=\varnothing$, return
        \YES.  Otherwise construct the field $K$
        in~\eqref{eq:working-field} and compute $\mu(K)$.
  \item For every integer $2\le k\le q^2+1$, enumerate the Block
        Orthogonality certificates in \cref{def:bo-certificate} of arity $k$ and
        construct $Q_k$ as in \cref{cor:bo-product-identities}.
  \item For every tested $k$ and every $P\in Q_k$, use
        \cref{thm:universal-identity} to decide whether
        \[
          P(\zvec I)=0
          \qquad\text{for every }I\in\InstMonoid{k}.
        \]
        Return \NO{} as soon as one identity fails.
  \item If all tests succeed, return \YES.
\end{enumerate}
\end{procedure}

\begin{theorem}[Exact decision of global BO]
\label{thm:global-bo-decidable}
\Cref{proc:global-bo} is a total exact algorithm that returns \YES{}
on $(q,\cF)$ if and only if $\BO(\cF)$ holds.
\end{theorem}

\begin{proof}
The case $\cF=\varnothing$ was verified in \cref{subsec:types-relations}.  Assume
$\cF\ne\varnothing$.

For termination, the arity loop is finite.  At each tested arity,
the row and column sets are
finite, $\mu(K)$ is finite, and the certificate construction ranges over
finitely many rectangle systems, row-pair modes, set partitions, anchors,
and root-of-unity labels.  Every root-sum test is an exact equality test in
$K$.  Each certificate has finitely many equations, and hence $Q_k$ is
finite and effectively computable.  If a certificate's equations all have
syntactically zero left-hand sides, this is detected when forming $Q_k$
and gives $Q_k=\varnothing$.

Finally, \cref{thm:universal-identity} is total and exact for each queried
polynomial.  Consequently every loop in \cref{proc:global-bo} is finite.

Suppose the procedure returns \YES.  Fix
$2\le k\le q^2+1$ and a $k$-ary $F\in W_{\cF}$.  By
\cref{lem:pinned-realization}, there is a finite
$I\in\InstMonoid{k}$ such that $F=\zvec I$.  Acceptance of all queries
at arity $k$ gives $P(F)=0$ for every $P\in Q_k$.
By~\eqref{eq:bo-locus-by-identities}, $F\in\mathcal B_k$.
Thus every function in $W_{\cF}$ of arity between two and $q^2+1$
has a block-orthogonal singleton purification.

If a function in $W_{\cF}$ of larger arity had a non-block-orthogonal singleton
purification, \cref{lem:bo-arity-bound} would produce such a member of
$W_{\cF}$ at a tested arity, a contradiction.  Hence every member of
$W_{\cF}$ of arity at least two has a block-orthogonal singleton
purification.  By \cref{cor:singleton-purification}, the joint condition
$\BO(\cF)$ holds.

Conversely, suppose $\BO(\cF)$ holds.  By~\cref{cor:singleton-purification},
every $F\in W_{\cF}$ of arity at
least two has a block-orthogonal singleton purification.  At every tested
arity $k$, \cref{lem:pinned-realization} therefore gives
$\zvec I\in\mathcal B_k$ for every $I\in\InstMonoid{k}$.
Equation~\eqref{eq:bo-locus-by-identities} gives $P(\zvec I)=0$
for every $P\in Q_k$ and every $I\in\InstMonoid{k}$.
The universal identity oracle accepts every query, and the procedure returns
\YES.

If the procedure returns \NO, a failed identity implies that there is a
finite labelled instance $I$ with $P(\zvec I)\ne0$ by
\cref{thm:universal-identity}.  Equation~\eqref{eq:bo-locus-by-identities}
then gives $\zvec I\notin\mathcal B_k$, a witness to failure of the
singleton condition.  The equivalence for \NO{} follows from termination
and the two directions already proved for \YES.
\end{proof}

\section{Generated Supports and a Common Mal'tsev Polymorphism}
\label{sec:generated-supports}

Assume $\BO(\cF)$.  We first construct a common Mal'tsev polymorphism
for a finite family of generated supports, then use the finiteness of
ternary operations on $D$ to obtain one for all generated supports.
For the finite family, we express preservation by constraints on the
operation's values and impose the Mal'tsev identities one input triple
at a time.  Block Orthogonality will allow each new prescribed value to
be added while retaining the earlier ones.  To apply it, we first show
that the solution relation of these constraints and its projections
are generated supports.

\subsection{Equality-free support realization}

A product of function values is nonzero exactly when every factor is
nonzero, so multiplication realizes conjunctions of support constraints.
Projection requires summing over omitted variables, where nonzero
contributions may cancel.  We will raise each contribution to a common
positive power before summing.  The next lemma supplies an exponent for
which every nonempty sum is nonzero, simultaneously for all retained
assignments.

\Needspace{8\baselineskip}
\begin{lemma}
\label{lem:simultaneous-power-sums}
Let $\mathcal A_1,\ldots,\mathcal A_t$ be nonempty finite multisets of
elements of $\AlgNums^\times$.  There exist infinitely many integers $m\ge1$
such that
\begin{equation}
  \sum_{a\in\mathcal A_i}a^m\ne0
  \qquad\text{for every }i\in[t].
  \label{eq:simultaneous-power-sums}
\end{equation}
Moreover, such an $m$ can be found by enumerating the positive integers and
performing exact algebraic zero tests.
\end{lemma}

\begin{proof}
We will show that all sufficiently large multiples of one positive integer
satisfy \eqref{eq:simultaneous-power-sums}.
We use the characteristic-zero form of the Skolem--Mahler--Lech
theorem: the zero set of a sequence satisfying a homogeneous linear
recurrence with constant coefficients is a finite set together with
finitely many arithmetic progressions~\cite{Lech1953}.
For $i\in[t]$, put $s_i(n)=\sum_{a\in\mathcal A_i}a^n$.
Let $E$ be the shift operator, $(Es)(n)=s(n+1)$.  The sequence $a^n$
is annihilated by $E-a$, so
\[
 \prod_{a\in\operatorname{supp}(\mathcal A_i)}(E-a)s_i=0.
\]
Expanding this product gives the required constant-coefficient recurrence
over the number field generated by $\mathcal A_i$.  The theorem therefore
applies to each $s_i$.

After absorbing finitely many initial terms into the finite exceptional
set, write each progression in the zero set as
$\{n\ge N:n\equiv r\pmod p\}, \quad 0\le r<p$.
We claim that its residue $r$ is nonzero.  Otherwise, for some $p\ge1$
and $N\ge0$,
\begin{equation}
  s_i(pn)=0
  \qquad(n\ge N).
  \label{eq:zero-progression-multiples}
\end{equation}
Group the elements of the multiset $\mathcal A_i$ according to their $p$th
powers.  Let $\beta_1,\ldots,\beta_h$ be the distinct resulting nonzero
values, and let $\nu_j\ge1$ be the multiplicity of $\beta_j$.  Then
\[
  s_i(pn)
  =\sum_{a\in\mathcal A_i}(a^p)^n
  =\sum_{j\in[h]}\nu_j\beta_j^n.
\]
Applying~\eqref{eq:zero-progression-multiples} at
$n=N,N+1,\ldots,N+h-1$ gives
\begin{equation}
  \begin{pmatrix}
    1&\cdots&1\\
    \beta_1&\cdots&\beta_h\\
    \vdots&&\vdots\\
    \beta_1^{h-1}&\cdots&\beta_h^{h-1}
  \end{pmatrix}
  \begin{pmatrix}
    \nu_1\beta_1^N\\
    \vdots\\
    \nu_h\beta_h^N
  \end{pmatrix}
  =\mathbf0.
  \label{eq:zero-progression-vandermonde}
\end{equation}
The determinant of the first matrix is $\prod_{1\le j<k\le h}(\beta_k-\beta_j)\ne0$, so~\eqref{eq:zero-progression-vandermonde} forces
$\nu_j\beta_j^N=0$ for every $j$.  Since every $\beta_j$ is nonzero, this
contradicts $\nu_j\ge1$.

To choose a common exponent, let $M$ be a common multiple of all
progression periods occurring for all $s_i$, taking $M=1$ if there are
no such progressions.  Every multiple of $M$ has residue zero modulo
each period, so it lies in none of these progressions.  All sufficiently
large multiples also avoid the finitely many exceptional zeros.
Thus infinitely many exponents satisfy~\eqref{eq:simultaneous-power-sums},
and enumerating $m=1,2,\ldots$ with exact zero tests finds one.
\end{proof}

We use this common exponent to realize the relation defined by an
equality-free primitive-positive (pp) formula.  Such a formula is formed
from relation atoms by conjunction and existential quantification.
The same variable may occur more than once, and equality atoms between
distinct variables are excluded.

\begin{lemma}
\label{lem:equality-free-support-realization}
Let $h\ge1$ and $F_1,\ldots,F_h\in W_{\cF}$.  Let
$\mathbf x=(x_1,\ldots,x_\ell)$, with $\ell\ge1$, be the tuple of free
variables, and let $\mathbf y=(y_1,\ldots,y_s)$, with $s\ge0$, be the tuple
of existential variables.  Let $\mathcal E$ be a finite set indexing the
atoms.  For every $e\in\mathcal E$, choose
$i(e)\in[h]$ and an ordered scope $\mathbf z_e$ of length
$\arity(F_{i(e)})$ whose entries are variables from $\mathbf x$ and
$\mathbf y$.  The same logical variable may occur repeatedly within one
scope or across several scopes.  Define $R\subseteq D^\ell$ by
\begin{equation*}
  R(\mathbf x)
  \quad\Longleftrightarrow\quad
  \exists\mathbf y\;
  \bigwedge_{e\in\mathcal E}
       \SuppRel{F_{i(e)}}(\mathbf z_e).
\end{equation*}
Then there exists $H\in W_{\cF}$ with $\SuppRel H=R$.  Given instances
$I_i$ with $F_i=F_{I_i}^{[r_i]}$, where $r_i=\arity(F_i)$, such an $H$
can be obtained by a terminating exact search.
\end{lemma}

\begin{proof}
For $m\ge1$, define
\begin{equation}
  H_m(\mathbf x)
  =\sum_{\mathbf y\in D^s}
       \prod_{e\in\mathcal E}F_{i(e)}(\mathbf z_e)^m.
  \label{eq:powered-support-table}
\end{equation}
The empty product is $1$.  For fixed $\mathbf a\in D^\ell$, let
$\mathcal A_{\mathbf a}$ be the multiset of all nonzero values $\prod_{e\in\mathcal E} F_{i(e)}(\mathbf z_e(\mathbf a,\mathbf b)), \quad \mathbf b\in D^s$. Then $\mathcal A_{\mathbf a}$ is nonempty exactly when $\mathbf a\in R$,
and
\begin{equation}
  H_m(\mathbf a)=\sum_{w\in\mathcal A_{\mathbf a}}w^m,
  \label{eq:witness-power-sum}
\end{equation}
where an empty sum is zero.  Thus tuples outside $R$ give zero term by term
for every $m$.  Apply \cref{lem:simultaneous-power-sums} to the finitely many
nonempty multisets $\mathcal A_{\mathbf a}$, $\mathbf a\in R$.  The resulting
common exponent makes~\eqref{eq:witness-power-sum} nonzero for every tuple in
$R$.  If $R=\varnothing$, any $m$ works.  Hence $\SuppRel{H_m}=R$.

We next realize $H_m$ by an instance over $\cF$.  For each $F_i$, choose a $\CSP{\cF}$ instance $I_i$ such
that $F_i=F_{I_i}^{[r_i]}$, where $r_i=\arity(F_i)$.  For every
occurrence $e$ and every
copy index in $[m]$, take a fresh copy of $I_{i(e)}$, identify its
retained variables with those in $\mathbf z_e$, and keep all of
its hidden variables fresh.  For fixed $(\mathbf x,\mathbf y)$,
summing over the copied hidden variables gives $\prod_{e\in\mathcal E}F_{i(e)}(\mathbf z_e)^m$. Summing over $\mathbf y$ then gives~\eqref{eq:powered-support-table}.
Repeated variables in $\mathbf z_e$ are legal scope identifications, and the
retained variables may be put in the displayed order.  Therefore
$H_m\in W_{\cF}$.
From the given instances $I_i$, all values of the $F_i$ and the finite
multisets $\mathcal A_{\mathbf a}$ can be computed exactly.
The search for $m$ terminates by \cref{lem:simultaneous-power-sums},
so the construction is effective.
\Cref{app:support-realization-example} illustrates it for two binary
functions.
\end{proof}

By \cref{lem:equality-free-support-realization}, Block Orthogonality
applies to relations defined by these formulas, including their
projections.  It forces any two nonempty row supports to be equal or
disjoint.  In the construction of a Mal'tsev operation, this property
will let us add one prescribed value without changing those already fixed.

\begin{corollary}[Singleton rectangularity]
\label{cor:generated-singleton-rectangularity}
Assume $\BO(\cF)$.  Let $R\subseteq D^{r-1}\times D$, with $r\ge2$, have an
equality-free pp definition from finitely many relations $\SuppRel F$,
$F\in W_{\cF}$.  For $\mathbf u\in D^{r-1}$, put $R_{\mathbf u}=\{a\in D:(\mathbf u,a)\in R\}$. Then any two nonempty sets $R_{\mathbf u}$ and $R_{\mathbf v}$ are equal or
disjoint.
\end{corollary}

\begin{proof}
By \cref{lem:equality-free-support-realization}, with the retained variables in the displayed order there is an $H\in W_{\cF}$ with $\SuppRel H=R$.  Apply $\BO(\cF)$ to
the singleton set $\{H\}$ and write $\Pure(H)=H'$.  Purification
preserves support, and $\abs{H'}$ is block-rank one.  Hence two
nonzero rows that overlap have equal support.
\end{proof}

\subsection{A common Mal'tsev polymorphism}

For a finite family of generated supports, introduce variables for the
outputs of a candidate ternary operation.  The constraints require that
applying it to the corresponding entries of any three tuples in each
relation produces another tuple in that relation.
The third projection preserves every relation and already satisfies
one Mal'tsev identity.  Starting from its values, we use
\cref{cor:generated-singleton-rectangularity} to impose the other identity
one input triple at a time.  At each step, the prescribed values will
extend to an operation preserving the entire family.

\begin{lemma}
\label{lem:finite-support-maltsev}
Assume $\BO(\cF)$.  Every finite family of generated support relations has a
common Mal'tsev polymorphism.
\end{lemma}

\begin{proof}
Let $\Delta\subseteq\{\SuppRel F:F\in W_{\cF}\}$ be finite.  Empty
relations may be removed because every operation preserves them.  If nothing
remains, fix $a_0\in D$ and use
\[
  \varphi_0(a,b,c)=
  \begin{cases}
    c,&a=b,\\
    a,&b=c,\\
    a_0,&\text{otherwise}.
  \end{cases}
\]
Thus assume that $\Delta$ contains a nonempty relation.  For every
$R\in\Delta$, fix $F_R\in W_{\cF}$ such that $R=\SuppRel{F_R}$, and write
$r_R$ for its arity.

The preservation constraints use only the input triples
$(u^1_j,u^2_j,u^3_j)\in D^3$, where $R\in\Delta$,
$\mathbf u^1,\mathbf u^2,\mathbf u^3\in R$, and $j\in[r_R]$.
Let $T\subseteq D^3$ be the set of these triples.  For every $\tau\in T$,
introduce a variable $x_\tau$ whose value in $D$ is the candidate
operation's output at $\tau$.  An assignment
$\mathbf p=(p_\tau)_{\tau\in T}\in D^T$ therefore specifies a partial
operation $p:T\to D$, with $x_\tau=p_\tau=p(\tau)$.

For every $R\in\Delta$ and
$\mathbf U=(\mathbf u^1,\mathbf u^2,\mathbf u^3)\in R^3$, use the
scope $\mathbf x[\mathbf U] =\bigl(x_{(u^1_j,u^2_j,u^3_j)}\bigr)_{j\in[r_R]}$.
The constraints requiring preservation of every relation in $\Delta$ are
\begin{align}
  \Psi_\Delta((x_\tau)_{\tau\in T})
  &=\bigwedge_{R\in\Delta}
    \ \bigwedge_{\mathbf U\in R^3}
    \SuppRel{F_R}\bigl(\mathbf x[\mathbf U]\bigr).
  \label{eq:universal-polymorphism-formula}
\end{align}
For a partial operation $p:T\to D$,
\begin{equation*}
  \Psi_\Delta(\mathbf p)\ \Longleftrightarrow\ 
  \forall R\in\Delta\ \forall\mathbf u^1,\mathbf u^2,\mathbf u^3\in R,
  \quad\bigl(p(u^1_j,u^2_j,u^3_j)\bigr)_{j\in[r_R]}\in R.
\end{equation*}
These constraints require one common operation to preserve all relations
in $\Delta$.  Every repeated input triple uses the same variable, so
\eqref{eq:universal-polymorphism-formula} is equality-free.
An example appears in \cref{app:preservation-constraints-example}.

Let $U_\Delta\subseteq D^T$ be the solution relation of $\Psi_\Delta$.
Every common ternary polymorphism of $\Delta$ restricts to an element
of $U_\Delta$.  Conversely, any extension of a solution from $T$ to
$D^3$ preserves every relation in $\Delta$: all input triples used in
each preservation test belong to $T$, and the corresponding conjunct
requires the resulting tuple to lie in the relation.
Thus $U_\Delta$ consists exactly of these restrictions.

The Mal'tsev identities prescribe values on the following subsets of $T$:
\begin{equation*}
  \begin{aligned}
    A&=\{(a,a,b)\in T:a,b\in D\},\\
    C&=\{(a,b,b)\in T:a,b\in D,\ a\ne b\},
  \end{aligned}
\end{equation*}
Let $P=\pr_{A\cup C}(U_\Delta)$ be the relation obtained by restricting
each solution to the inputs in $A\cup C$.  The sets $A$ and $C$ are
disjoint, and every constant triple in $T$ lies in $A$.
It suffices to find an element of $P$ assigning $b$ to each $(a,a,b)\in A$
and $a$ to each $(a,b,b)\in C$, since membership in $P$ guarantees an
extension to a common polymorphism.
For each $k\in[3]$,
let $\pi_k:D^3\to D, \quad \pi_k(a_1,a_2,a_3)=a_k$, be the $k$th projection.  It preserves every relation:
for $R\in\Delta$ and
$\mathbf u^1,\mathbf u^2,\mathbf u^3\in R$,
\[
  \pi_k(\mathbf u^1,\mathbf u^2,\mathbf u^3)
  =\bigl(\pi_k(u^1_j,u^2_j,u^3_j)\bigr)_{j\in[r_R]}
  =\mathbf u^k\in R.
\]
Hence
\begin{equation}
  \pi_k|_T\in U_\Delta,
  \qquad
  \pi_k|_{A\cup C}\in P
  \qquad(k\in[3]).
  \label{eq:projection-tables-in-P}
\end{equation}
Also, $A\ne\varnothing$: repeating any tuple of a nonempty
positive-arity relation three times produces a constant triple in $T$.

Enumerate $C=\{c_1,\ldots,c_s\}$.  For $0\le j\le s$, we inductively
construct a partial operation $p_j:A\cup C\to D$ satisfying the three
invariants
\begin{align}
  p_j&\in P,
    \notag\\
  p_j(a,a,b)&=b
    &&\text{for every }(a,a,b)\in A,
    \notag\\
  p_j(c_i)&=\pi_1(c_i)
    &&\text{for every }i\in[j].
    \label{eq:partial-maltsev-c}
\end{align}
For $j=0$, take $p_0=\pi_3|_{A\cup C}$.  Membership follows from
\eqref{eq:projection-tables-in-P}, the first value condition follows from
$\pi_3(a,a,b)=b$, and the second is vacuous.

Suppose $0\le j<s$ and the three invariants hold for $p_j$.
To combine the values already prescribed by $p_j$ with the required value
at $c_{j+1}$, we first construct a common polymorphism of $\Delta$ that
agrees with $\pi_1$ on $A$ and at $c_1,\ldots,c_j$, and with $p_j$ at
$c_{j+1}$.  Because $p_j\in P=\pr_{A\cup C}(U_\Delta)$, there exists a lift $\widehat p_j\in U_\Delta$ satisfying $\widehat p_j|_{A\cup C}=p_j$. This lift is effective: enumerate all maps
$T\setminus(A\cup C)\to D$, combine each with the fixed values of $p_j$,
and test~\eqref{eq:universal-polymorphism-formula}.  There are at most
$q^{\abs{T\setminus(A\cup C)}}$ candidates, and membership in $P$ guarantees
that one succeeds.  Extend $\widehat p_j$ arbitrarily to a map $t_j:D^3\to D$.
The characterization of $U_\Delta$ shows that $t_j$ is a common ternary
polymorphism of $\Delta$.

The minor $\widehat t_j(x,y,z)=t_j(x,y,y)$ is also a common ternary
polymorphism: applying $t_j$ to a first tuple and two copies of a second
tuple in a relation still gives a tuple in that relation.  Let
$K_j=A\cup\{c_i:i\in[j]\}$.  On $K_j$, $\widehat t_j$ agrees with $\pi_1$.
For $c_i$ with $i\in[j]$, the second and third entries coincide, so
this follows from~\eqref{eq:partial-maltsev-c}.
If $(x,x,z)\in A$, then $(x,x,x)\in T$ by repeating the tuple that witnesses
the occurrence of $x$.  Since $(x,x,x)\in A$, $\widehat t_j(x,x,z)=t_j(x,x,x)=x=\pi_1(x,x,z)$.

We now impose the required value at $c_{j+1}$ while keeping all values
on $K_j$.  Write $c_{j+1}=(a,b,b)$ and set
\[
\begin{aligned}
  X&=\pi_1|_{K_j}=\widehat t_j|_{K_j},
  &Y&=p_j|_{K_j},\\
  v&=p_j(c_{j+1})=\widehat t_j(c_{j+1}).
\end{aligned}
\]
To apply \cref{cor:generated-singleton-rectangularity} with assignments on
$K_j$ as rows, let $Q_j=\pr_{K_j\cup\{c_{j+1}\}}(P)$, ordering the variable
indexed by $c_{j+1}$ last.  This projection is defined by existentially
quantifying the omitted variables in~\eqref{eq:universal-polymorphism-formula},
so its pp definition is equality-free.  Since $A\ne\varnothing$, its
arity is at least two.  The minor $\widehat t_j$ is a common polymorphism, so its
restriction belongs to $P$. The same is true of $p_j$ by the induction
invariant and of $\pi_1$ by~\eqref{eq:projection-tables-in-P}.  Restricting
$\widehat t_j$, $p_j$, and $\pi_1$ to $K_j\cup\{c_{j+1}\}$ therefore gives,
respectively, the three corners
\begin{equation*}
  (X,v),\qquad(Y,v),\qquad(X,a)\in Q_j.
\end{equation*}
\Cref{fig:finite-support-induction} summarizes both the finite lifting step
and the missing-corner argument at this point of the induction.

\begin{figure}[!t]
  \centering
  \resizebox{\linewidth}{!}{%
  \begin{tikzpicture}[
      x=1cm,y=1cm,
      title/.style={font=\small\bfseries},
      panel/.style={draw,rounded corners=2pt,line width=.45pt},
      rel/.style={draw,rounded corners=2pt,minimum width=2.75cm,
                  minimum height=7mm,font=\scriptsize,align=center},
      elem/.style={draw,double,double distance=.45pt,rounded corners=2pt,
                   minimum width=2.75cm,minimum height=7mm,
                   font=\scriptsize,align=center},
      cell/.style={draw,minimum width=2.25cm,minimum height=9mm,
                   font=\scriptsize,align=center},
      forcedcell/.style={cell,double,double distance=.6pt,line width=.65pt},
      arrow/.style={-{Latex[length=2mm]},line width=.55pt},
      lift/.style={-{Latex[length=2mm]},line width=.55pt,dashed}
    ]
    % Panel (a): restrictions, a lift, extension, and a minor.
    \draw[panel] (-.20,5.35) rectangle (14.20,10.75);
    \node[title] at (7.00,10.40) {(a) projection and finite lifting};

    \node[rel,minimum width=3.0cm] (U) at (1.80,9.30)
      {$U_\Delta\subseteq D^T$};
    \node[rel,minimum width=3.4cm] (P) at (7.00,9.30)
      {$P=\pr_{A\cup C}(U_\Delta)$};
    \node[rel,minimum width=3.6cm] (Q) at (12.20,9.30)
      {$Q_j=\pr_{K_j\cup\{c_{j+1}\}}(P)$};
    \draw[arrow] (U)--(P)
      node[midway,above=2pt,font=\scriptsize] {restriction};
    \draw[arrow] (P)--(Q)
      node[midway,above=2pt,font=\scriptsize] {restriction};

    \node[elem,minimum width=3.4cm] (hatp) at (1.80,7.77)
      {$\widehat p_j\in U_\Delta$\\$\widehat p_j|_{A\cup C}=p_j$};
    \node[elem,minimum width=3.0cm] (pj) at (7.00,7.77) {$p_j\in P$};
    \draw[lift] (pj)--(hatp)
      node[midway,above=2pt,font=\scriptsize] {finite lift};
    \draw[arrow] (hatp)--(U);
    \draw[arrow] (pj)--(P);

    \node[elem,minimum width=3.4cm] (tj) at (1.80,6.22)
      {$t_j:D^3\to D$\\extends $\widehat p_j$};
    \node[elem,minimum width=4.2cm] (qj) at (7.20,6.22)
      {$\widehat t_j(x,y,z)=t_j(x,y,y)$\\common polymorphism,
       $\widehat t_j|_{A\cup C}\in P$};
    \draw[arrow] (hatp)--(tj)
      node[midway,right=2pt,font=\scriptsize] {extend};
    \draw[arrow] (tj)--(qj)
      node[midway,above=2pt,font=\scriptsize] {minor};
    \node[font=\scriptsize,align=center,text width=3.4cm] at (11.95,7.77)
      {restrictions of $p_j$, $\widehat t_j$, and $\pi_1$ give the three corners};

    % Panel (b): singleton rectangularity fills the missing corner.
    \draw[panel] (-.20,-.40) rectangle (14.20,5.05);
    \node[title] at (7.00,4.70) {(b) the missing-corner step in $Q_j$};
    \node[font=\scriptsize] at (7.70,4.13)
      {last input $c_{j+1}$};
    \node[font=\scriptsize] at (6.45,3.73) {$v$};
    \node[font=\scriptsize] at (9.15,3.73) {$a$};
    \node[font=\scriptsize] at (4.10,3.18) {row $X$};
    \node[font=\scriptsize] at (4.10,2.13) {row $Y$};

    \node[cell,minimum width=2.5cm] (Xv) at (6.45,3.18)
      {$(X,v)$\\from $\widehat t_j$};
    \node[cell,minimum width=2.5cm] (Xa) at (9.15,3.18)
      {$(X,a)$\\from $\pi_1$};
    \node[cell,minimum width=2.5cm] (Yv) at (6.45,2.13)
      {$(Y,v)$\\from $p_j$};
    \node[forcedcell,minimum width=2.5cm] (Ya) at (9.15,2.13)
      {$(Y,a)$\\forced into $Q_j$};

    \node[font=\scriptsize] at (6.25,1.47)
      {the two rows meet at $v$, so rectangularity fills the fourth corner};
    \node[elem,text width=7.0cm] (next) at (7.00,.48)
      {lift $(Y,a)\in Q_j$ to $p_{j+1}\in P$ with
       $p_{j+1}|_{K_j}=Y$ and $p_{j+1}(c_{j+1})=a$};
    % Route the dashed lift down the right side so it does not cover the text.
    \draw[lift] (Ya.south east) to[out=-70,in=70] (next.north east);
  \end{tikzpicture}%
  }
  \caption{One induction step in
    \cref{lem:finite-support-maltsev}.  Top: membership $p_j\in P$
    guarantees a lift $\widehat p_j\in U_\Delta$. Because all sets are
    finite, it is found by exhaustive search.  Extending that lift and taking
    the displayed minor supplies the common polymorphism $\widehat t_j$.  Solid arrows
    are restrictions, extensions, or minors, whereas dashed arrows are finite
    lifts.  Bottom: in the row view of $Q_j$, the rows $X$ and $Y$ intersect at
    $v$.  Singleton rectangularity makes their supports equal, so the
    double-bordered fourth corner $(Y,a)$ must exist.  Lifting that corner
    produces the next partial operation $p_{j+1}$.}
  \label{fig:finite-support-induction}
\end{figure}

By \cref{cor:generated-singleton-rectangularity}, the two nonempty rows
indexed by $X$ and $Y$ have equal support because they intersect at $v$.
Consequently $(Y,a)\in Q_j$.  Lifting this tuple to $P$ gives
$p_{j+1}\in P$ such that
\begin{equation*}
\begin{aligned}
  p_{j+1}|_{K_j}&=Y=p_j|_{K_j},\\
  p_{j+1}(c_{j+1})&=a=\pi_1(c_{j+1}).
\end{aligned}
\end{equation*}
Thus $p_{j+1}$ retains the required values on $A,c_1,\ldots,c_j$ and
has the required value at $c_{j+1}$, establishing all three invariants.
This second lift is also a finite exact construction: enumerate
$\widehat p\in D^T$ subject to
\[
  \widehat p|_{K_j}=Y,
  \qquad
  \widehat p(c_{j+1})=a,
  \qquad
  \Psi_\Delta(\widehat p),
\]
and take $p_{j+1}=\widehat p|_{A\cup C}$.  Membership
$(Y,a)\in Q_j$ guarantees that this search terminates.

At $j=s$, choose a lift $\widehat p_s\in U_\Delta$ with
$\widehat p_s|_{A\cup C}=p_s$, using the same finite search.
Its values on $A\cup C$ satisfy both Mal'tsev identities wherever their
input triples lie in $T$.  Extend it to $D^3$ by assigning the required
values at triples of the forms $(a,a,b)$ and $(a,b,b)$ outside $T$,
and arbitrary values elsewhere.  The two prescriptions agree on their
overlap, the constant triples.  Denote the resulting operation by
$\varphi$.  Since every preservation constraint uses only inputs from
$T$, $\varphi$ preserves every relation in $\Delta$.  It satisfies
\[
  \varphi(a,b,b)=\varphi(b,b,a)=a
  \qquad(a,b\in D),
\]
and is therefore a common Mal'tsev polymorphism of $\Delta$.
\end{proof}

\subsection{All generated supports}

The operation in \cref{lem:finite-support-maltsev} may depend on the
finite family.  To obtain one operation for all generated supports,
we use the finiteness of the set of Mal'tsev operations on $D$:
if no candidate preserved all supports, finitely many supports would
exclude every candidate.

\begin{theorem}[Common Mal'tsev polymorphism of generated supports]
\label{thm:common-support-maltsev}
If $\BO(\cF)$ holds, then
\begin{equation}
  \exists\varphi:D^3\to D\quad\forall F\in W_{\cF},\quad
  \varphi\text{ is a Mal'tsev polymorphism of }\SuppRel F.
  \label{eq:global-support-maltsev}
\end{equation}
\end{theorem}

\begin{proof}
Let $\mathcal M_q$ be the finite set of all Mal'tsev operations on
$D$.  By \cref{lem:finite-support-maltsev}, every finite family of generated
supports is preserved by some member of $\mathcal M_q$.

Suppose, toward a contradiction, that no member of $\mathcal M_q$ preserves
all generated supports.  For every $\varphi\in\mathcal M_q$, choose
$F_\varphi\in W_{\cF}$ such that $\varphi$ fails to preserve $\SuppRel{F_\varphi}$.
The finite family $\Delta=\{\SuppRel{F_\varphi}:\varphi\in\mathcal M_q\}$ has a common Mal'tsev polymorphism $\varphi^\star\in\mathcal M_q$ by
\cref{lem:finite-support-maltsev}.  But $\Delta$ contains
$\SuppRel{F_{\varphi^\star}}$, which was chosen not to be preserved by
$\varphi^\star$, a contradiction.  This proves~\eqref{eq:global-support-maltsev}.
\end{proof}

\section{Realizing Row Equivalence as a Generated Support}
\label{sec:omega-realization}

Assume $\BO(\cF)$.  \Cref{sec:generated-supports} gives one Mal'tsev
operation preserving all generated supports.  To show that the same
operation preserves row equivalence, we construct, for each $F$, a
generated function whose support is $\RowEq F$.  Its value at
$(\mathbf x,\mathbf y)$ will be a sum of products of positive powers of
entries in the two rows indexed by $\mathbf x$ and $\mathbf y$.  We
first establish that powers remain generated and that the relative row
phases are controlled by roots of unity.  These facts, together with a
power-sum argument, let us choose exponents for which the sum is nonzero
exactly on pairs of nonzero, linearly dependent rows.

Let $e_K$ be the exponent of the finite group $\mu(K)$, and define $L_q=\lcm(1,\ldots,q)$. Since $-1\in\mu(K)$ has order two, $e_K$ is even and $e_K\ge2$.
Both integers are computable from the input, and
\begin{equation}
  \xi^{e_K}=1
  \qquad\text{for every }\xi\in\mu(K).
  \label{eq:Kmu-kills-torsion}
\end{equation}
For a function or vector $H$ and $h\ge1$, let $\Pow{H}{h}$ denote its entrywise
$h$th power: each value $H(\mathbf a)$ is replaced by
$H(\mathbf a)^h$.  The angle brackets distinguish entrywise powers from
matrix powers when a function is represented as a matrix.

\begin{lemma}
\label{lem:generated-entrywise-powers}
If $F\in W_{\cF}$ and $h\ge1$, then $\Pow{F}{h}\in W_{\cF}$.
\end{lemma}

\begin{proof}
Choose a finite $\cF$-instance $I$ with retained variables $\mathbf a$
and hidden variables $\mathbf w$ such that $F(\mathbf a)=\sum_{\mathbf w\in D^s}F_I(\mathbf a,\mathbf w)$, where $F_I$ is its instance function from \eqref{eq:instance-function}.
Take $h$ copies of that instance, identify their retained variables
$\mathbf a$, and make their internal tuples
$\mathbf w^{(1)},\ldots,\mathbf w^{(h)}$ pairwise fresh.  Summing over them
gives
\[
  \sum_{\mathbf w^{(1)},\ldots,\mathbf w^{(h)}}
     \prod_{j\in[h]}F_I(\mathbf a,\mathbf w^{(j)})
  =F(\mathbf a)^h.
\]
Thus $\Pow{F}{h}\in W_{\cF}$.
\end{proof}

The closure under powers will keep the construction inside $W_{\cF}$.
We next compare the ratios between two overlapping rows.  Dividing by the
ratio at one nonzero entry separates a common scalar from the variation
across entries.  The next lemma shows that the remaining factors are roots
of unity whose orders divide $L_q$.  We can then choose powers that
preserve these factors.

\begin{lemma}
\label{lem:finite-row-phases}
Assume $\BO(\cF)$.  Let $F:D^r\to K$, $r\ge2$, belong to $W_{\cF}$, and
let $\mathbf u=F(\mathbf x,\ast)$ and
$\mathbf v=F(\mathbf y,\ast)$ be nonzero rows with overlapping supports.
Then their supports are equal.  Fix any anchor $z_0$ in their common support
$S$, and define
\begin{equation}
  \lambda:=\frac{v_{z_0}}{u_{z_0}},
  \qquad
  \zeta_z:=
    \frac{v_z/u_z}{v_{z_0}/u_{z_0}}
    =\frac{v_zu_{z_0}}{u_zv_{z_0}}
  \qquad(z\in S).
  \label{eq:anchored-row-phases}
\end{equation}
Then $\lambda\in K^\times$ and
$(\zeta_z)_{z\in S}\in\mu(K)^S$.  For every $z\in S$,
\begin{equation}
  v_z=\lambda\zeta_z u_z,
  \qquad
  \zeta_{z_0}=1,
  \qquad
  \zeta_z^{L_q}=1.
  \label{eq:finite-row-phases}
\end{equation}
\end{lemma}

\begin{proof}
By \cref{cor:singleton-purification} and Corollary~1 of the Purification Lemma
of Cai and Chen~\cite[Lemma~7 and Corollary~1]{CaiChen2017}, the original
function $F$ is block-orthogonal.
Hence overlapping nonzero rows have equal
supports and proportional magnitude vectors.

By \cref{lem:generated-entrywise-powers}, $\Pow{F}{h}\in W_{\cF}$ for every
$h\ge1$.  Applying the same corollary again shows that each original function
$\Pow{F}{h}$ is block-orthogonal.  We show that some $1\le h\le q$ makes
the two powered rows linearly dependent, so their entry ratios have equal
$h$th powers.  Suppose, toward a contradiction, that
$\Pow{\mathbf u}{h}$ and $\Pow{\mathbf v}{h}$ are linearly independent
for every $1\le h\le q$.  Fix a nonempty magnitude block $B$ of $\mathbf u$.
Positive powering is injective on positive real numbers, so $B$ remains the
same magnitude block in every $\Pow{\mathbf u}{h}$.  Block orthogonality of
$\Pow{F}{h}$ gives the following power sums.  Put $b=\abs B, \quad x_z=u_z\overline{v_z}\quad(z\in B)$. Because $B$ lies in the common support, every $x_z$ is nonzero.  Also
$b\le q$, and hence
\begin{equation*}
  p_h:=\sum_{z\in B}x_z^h
  =\sum_{z\in B}u_z^h\overline{v_z^h}=0,
  \qquad h=1,\ldots,b.
\end{equation*}
Enumerate the $x_z$, $z\in B$, as $x_1,\ldots,x_b$, and let
\[
  e_k=\sum_{1\le i_1<\cdots<i_k\le b}
       x_{i_1}\cdots x_{i_k},
  \qquad e_0=1,
\]
be their elementary symmetric polynomials.  Newton's identities state
\begin{equation}
  k e_k
  =\sum_{j=1}^k(-1)^{j-1}e_{k-j}p_j
  \qquad(k=1,\ldots,b).
  \label{eq:newton-block-identities}
\end{equation}
Since the ambient field has characteristic zero, each integer
$1\le k\le b$ is invertible.  As $p_1=\cdots=p_b=0$, induction on $k$ in
\eqref{eq:newton-block-identities} gives
$e_1=\cdots=e_b=0$.  In particular,
\[
  0=e_b=x_1\cdots x_b
  =\prod_{z\in B}u_z\overline{v_z},
\]
contradicting the fact that every factor is nonzero.

Consequently, $\Pow{\mathbf u}{h}$ and $\Pow{\mathbf v}{h}$ are dependent
for some $1\le h\le q$.  Thus $(v_z/u_z)^h$ is independent of $z\in S$.  By
\eqref{eq:anchored-row-phases}, for every $z\in S$,
\[
  \zeta_z^h
  =\frac{(v_z/u_z)^h}{(v_{z_0}/u_{z_0})^h}=1.
\]
Since $h$ divides $L_q$, \eqref{eq:finite-row-phases} follows.  All ratios
in~\eqref{eq:anchored-row-phases} belong to $K^\times$, so each
$\zeta_z\in\mu(K)$. The same display gives
$\lambda\in K^\times$ and $\zeta_{z_0}=1$.
\end{proof}

The remaining ingredient concerns cancellation on linearly dependent
rows.  After grouping equal exponential bases in the sum, we will use
the following consequence of the Skolem--Mahler--Lech theorem to ensure
that the resulting sum is nonzero for all sufficiently large values of
its integer parameter.

\begin{lemma}
\label{lem:nondegenerate-exponential-sums}
Let $s\ge1$, let $a_1,\ldots,a_s,b_1,\ldots,b_s$ be nonzero algebraic
numbers, and
suppose $b_i/b_j$ is not a root of unity whenever $i\ne j$.  Then
\[
  p(t)=\sum_{i\in[s]} a_ib_i^t
\]
equals zero for only finitely many integers $t\ge0$.
\end{lemma}

\begin{proof}
Let $E$ be the shift operator from the proof of
\cref{lem:simultaneous-power-sums}.  The operator $\prod_i(E-b_i)$
annihilates $p(t)$.  The Skolem--Mahler--Lech theorem
therefore says that
its zero set is a finite set together with finitely many arithmetic
progressions.  Suppose one such progression is
$\{r+mn:n\ge N\}$, where $m\ge1$.  Then
\[
  0=p(r+mn)=\sum_{i\in[s]}(a_i b_i^r)(b_i^m)^n
\]
for $n=N,N+1,\ldots,N+s-1$.  The bases $b_i^m$ are pairwise distinct because
no $b_i/b_j$ is a root of unity.  The resulting Vandermonde system is
invertible and forces
$(a_i b_i^r)(b_i^m)^N=0$ for every $i$, contradicting the nonzero hypotheses.
Hence no infinite progression occurs.
\end{proof}

We now combine these facts in a function whose support is $\RowEq F$.
Taking $n=1+tL_q$, with $t\ge0$, preserves the relative factors
$\zeta_z$ in \eqref{eq:finite-row-phases}.  The two row exponents below
sum to $e_Kn$, which removes root-of-unity ratios within each purified
magnitude block.  Block Orthogonality then makes the sum zero on
independent rows.  For nonzero dependent rows,
\cref{lem:nondegenerate-exponential-sums} will make it nonzero for all
sufficiently large $t$.  Copies of the original instance realize this
function.  We call the resulting function $C_n$ the
\emph{row-equivalence detector}.

\begin{theorem}[Row-equivalence realization]
\label{thm:omega-support-realization}
Assume $\BO(\cF)$.  Let $F:D^r\to K$, $r\ge2$, be any member of
$W_{\cF}$.  There is an integer $t\ge0$ such that, for $n=1+tL_q$, the
$2(r-1)$-ary function
\begin{equation}
  C_n(\mathbf x,\mathbf y)
  =\sum_{z\in D}
       F(\mathbf x,z)^{(e_K-1)n}F(\mathbf y,z)^n
  \label{eq:omega-detector}
\end{equation}
belongs to $W_{\cF}$ and satisfies
\begin{equation}
  \SuppRel{C_n}=\RowEq F.
  \label{eq:omega-is-support}
\end{equation}
Moreover, there is an algorithm that, given the exact values of $F$ and
an instance $I$ with $F=F_I^{[r]}$, computes such integers $t,n$ and an
instance $J$ satisfying $C_n=F_J^{[2(r-1)]}$, with retained variables
ordered as $(\mathbf x,\mathbf y)$.
\end{theorem}

\begin{proof}
Fix an allowed exponent $n=1+tL_q$.  Both exponents
in~\eqref{eq:omega-detector} are positive.  Let
$\mathbf u=F(\mathbf x,\ast)$ and $\mathbf v=F(\mathbf y,\ast)$.  If one row
is zero, or if their supports are disjoint, every summand is zero.
Suppose henceforth that the rows are nonzero and overlap.  The global Block
Orthogonality condition and support invariance under purification imply that their
supports agree and that their purified magnitude rows are proportional.
Write their common support as $S=\supp(\mathbf u)=\supp(\mathbf v)$.

First suppose that $\mathbf u$ and $\mathbf v$ are linearly independent.
Fix a legal singleton purification of $F$, and let $B$ be one magnitude block
of the purified row corresponding to $\mathbf u$.  The two purified
magnitude rows are proportional, so $B$ is simultaneously a magnitude block
of the purified row corresponding to $\mathbf v$.  Choose $z_B\in B$.  By
\cref{lem:intrinsic-purification-tests}, the anchor ratios
\begin{equation}
  \eta_z=u_z/u_{z_B},
  \qquad
  \theta_z=v_z/v_{z_B}
  \qquad(z\in B)
  \label{eq:block-anchor-ratios}
\end{equation}
belong to $\mu(K)$.  Purification preserves complex row dependence, so the
two purified rows remain independent and are block-orthogonal.  Root
covariance then turns the Hermitian sum on $B$ into
\begin{equation}
  \sum_{z\in B}\eta_z\overline{\theta_z}=0.
  \label{eq:purified-root-sum}
\end{equation}

Apply \cref{lem:finite-row-phases} with the same anchor $z_B$:
$v_z=\lambda\zeta_z u_z$, where $\zeta_{z_B}=1$ and
$\zeta_z^{L_q}=1$.  Taking ratios gives $\theta_z=\zeta_z\eta_z$.
Substitution in~\eqref{eq:purified-root-sum}, followed by conjugation, yields
\begin{equation*}
  \sum_{z\in B}\zeta_z=0.
\end{equation*}
Equation~\eqref{eq:block-anchor-ratios} and
\eqref{eq:Kmu-kills-torsion} give
$u_z^{e_K}=u_{z_B}^{e_K}$.  Since $n\equiv1\pmod{L_q}$, also
$\zeta_z^n=\zeta_z$.  The contribution of $B$ to
$C_n(\mathbf x,\mathbf y)$ is therefore
\begin{equation}
  \sum_{z\in B}u_z^{(e_K-1)n}v_z^n
  =\lambda^n\sum_{z\in B}\zeta_z^n u_z^{e_K n}
  =\lambda^n u_{z_B}^{e_K n}\sum_{z\in B}\zeta_z=0.
  \label{eq:independent-block-stays-zero}
\end{equation}
The purified magnitude blocks partition the common support. Write $S=B_1\mathbin{\dot\cup}\cdots\mathbin{\dot\cup}B_N$. Terms outside $S$ are zero, and
\eqref{eq:independent-block-stays-zero} applies to every $B_i$.  Therefore
\begin{equation*}
  C_n(\mathbf x,\mathbf y)
  =\sum_{z\in S}u_z^{(e_K-1)n}v_z^n
  =\sum_{i\in[N]}\sum_{z\in B_i}u_z^{(e_K-1)n}v_z^n=0.
\end{equation*}
Thus the detector equals zero on every overlapping independent ordered row pair
for every allowed $n$.

It remains to treat a nonzero dependent pair, say
$\mathbf v=\lambda\mathbf u$ with $\lambda\ne0$.  Using the common support
$S$, we have
\begin{equation*}
  C_n(\mathbf x,\mathbf y)
   =\lambda^n\sum_{z\in S}u_z^{e_K n}.
\end{equation*}
Along $n=1+tL_q$, put $a_z=u_z^{e_K}$ and
$b_z=u_z^{e_K L_q}$.  The remaining sum is
\begin{equation}
  p_{\mathbf x,\mathbf y}(t)=\sum_{z\in S}a_zb_z^t.
  \label{eq:dependent-exponential-sum}
\end{equation}
To apply \cref{lem:nondegenerate-exponential-sums}, group the indices
$z\in S$ having the same base $b_z$.  We check that the grouped coefficients
are nonzero and that the distinct bases have no root-of-unity quotients.
If $b_z=b_w$, then
$u_z/u_w\in\mu(K)$, so $a_z=a_w$ by
\eqref{eq:Kmu-kills-torsion}.  Thus each grouped coefficient is a positive
integer multiple of one nonzero $a_z$.  Moreover, the quotient of two
distinct grouped bases is not a root of unity: otherwise a power of
$u_z/u_w$ would be one, so $u_z/u_w\in\mu(K)$ and the two bases would have
been equal.  After grouping,
\cref{lem:nondegenerate-exponential-sums} shows that
\eqref{eq:dependent-exponential-sum} has only finitely many zeros.

There are finitely many ordered dependent nonzero row pairs in $F$.  Choose
$t$ outside the union of their finite exceptional sets.  Then the detector
is nonzero on every member of $\RowEq F$, whereas the zero-row,
disjoint-support, and independent cases show that it equals zero on every
nonmember.  This proves~\eqref{eq:omega-is-support}.

The witness $t$ is found by exact search.  Compute $\RowEq F$ exactly: a row
is nonzero if one entry is nonzero, and two
nonzero rows are dependent exactly when
\begin{equation*}
  F(\mathbf x,z)F(\mathbf y,w)
  -F(\mathbf x,w)F(\mathbf y,z)=0
  \qquad(z,w\in D).
\end{equation*}
Enumerate $t=0,1,2,\ldots$, evaluate~\eqref{eq:omega-detector} by exact
algebraic arithmetic, and stop when its support equals $\RowEq F$.  The
finite-exception argument proves termination.

Finally, we realize the two powers in \eqref{eq:omega-detector} by
copies of an instance representing $F$.  Choose a finite $\cF$-instance $I$
with retained variables $(\mathbf x,z)$ and hidden variables $\mathbf h$
such that $F(\mathbf x,z)=\sum_{\mathbf h\in D^s}F_I(\mathbf x,z,\mathbf h)$,
where $F_I$ is its instance function.  Take $(e_K-1)n$ copies sharing
$(\mathbf x,z)$ and $n$ copies sharing $(\mathbf y,z)$, with all copied
hidden variables distinct.  For fixed $(\mathbf x,\mathbf y,z)$, summing
over these two collections of hidden variables gives
\begin{align*}
  \left(\sum_{\mathbf h}F_I(\mathbf x,z,\mathbf h)\right)^{(e_K-1)n}
    &=F(\mathbf x,z)^{(e_K-1)n},\\
  \left(\sum_{\mathbf k}F_I(\mathbf y,z,\mathbf k)\right)^n
    &=F(\mathbf y,z)^n.
\end{align*}
Because the collections are disjoint, their contributions multiply.
Summing over $z$ now gives exactly $C_n(\mathbf x,\mathbf y)$.
Thus $C_n\in W_{\cF}$.
\end{proof}

The realization theorem allows us to use the operation from
\cref{sec:generated-supports} for row equivalence as well.  The following
consequence applies to any operation preserving all generated supports.

\begin{corollary}[Support preservation extends to all row relations]
\label{cor:omega-reduces-to-supports}
Assume $\BO(\cF)$.  If one operation $\varphi:D^3\to D$ preserves $\SuppRel H$
for every $H\in W_{\cF}$, then the same operation preserves $\RowEq F$ for
every $F\in W_{\cF}$ of arity at least two.
\end{corollary}

\begin{proof}
Fix the operation $\varphi$ from the hypothesis, and let
$F\in W_{\cF}$ have arity at least two.  By
\cref{thm:omega-support-realization}, choose a row-equivalence detector
$C_n\in W_{\cF}$ with $\SuppRel{C_n}=\RowEq F$.  By hypothesis,
$\varphi$ preserves $\SuppRel{C_n}$, and hence $\RowEq F$.
Since $F$ was arbitrary, this same $\varphi$ preserves every generated
row-equivalence relation.
\end{proof}

\section{Type Partition and the Structural Collapse}
\label{sec:type-collapse}

Under $\BO(\cF)$, we have a common Mal'tsev operation preserving all
generated supports, and \cref{cor:omega-reduces-to-supports} shows that it
also preserves all row-equivalence relations.  It remains to derive the
Type Partition condition.  The next lemma does this for one function:
preservation of row equivalence and the Mal'tsev identities transfer a
row-equivalence class reached from one prefix to any same-length prefix
whose type overlaps it.

For equal-length tuples $\mathbf u,\mathbf v,\mathbf w$, we write
$\varphi(\mathbf u,\mathbf v,\mathbf w)$ for coordinatewise application of $\varphi$.

\begin{lemma}
\label{lem:omega-implies-type}
Let $F:D^r\to\AlgNums$, $r\ge2$, and let $\varphi:D^3\to D$ be a Mal'tsev
polymorphism of $\RowEq F$.  Then the type map of $F$ is a
type-partition map.
\end{lemma}

\begin{proof}
Write $\Row(F)=\{(S_1,\mathbf v_1),\ldots,(S_s,\mathbf v_s)\}$. Fix $0\le\ell\le r-1$ and
$\boldsymbol\alpha,\boldsymbol\beta\in D^\ell$.  Suppose their types
intersect.  Choose $j\in\typeop_F(\boldsymbol\alpha) \cap\typeop_F(\boldsymbol\beta)$ and rows $\mathbf a,\mathbf b\in S_j$ extending
$\boldsymbol\alpha,\boldsymbol\beta$, respectively.

Let $h\in\typeop_F(\boldsymbol\alpha)$ be arbitrary.  To show that
$h\in\typeop_F(\boldsymbol\beta)$, we construct a row index in $S_h$
extending $\boldsymbol\beta$.  Choose $\mathbf c\in S_h$ extending
$\boldsymbol\alpha$.  The three
$2(r-1)$-tuples
\[
  (\mathbf c,\mathbf c),\qquad
  (\mathbf a,\mathbf a),\qquad
  (\mathbf b,\mathbf a)
\]
belong to $\RowEq F$: the first two are diagonal pairs of nonzero rows, and
the third belongs because $\mathbf a$ and $\mathbf b$ lie in the same row
class $S_j$.  Put $\mathbf d=\varphi(\mathbf c,\mathbf a,\mathbf b)$.  Applying
$\varphi$ coordinatewise to these three members of $\RowEq F$ gives
\[
  \varphi\bigl((\mathbf c,\mathbf c),(\mathbf a,\mathbf a),(\mathbf b,\mathbf a)\bigr)
  =\bigl(\varphi(\mathbf c,\mathbf a,\mathbf b),\varphi(\mathbf c,\mathbf a,\mathbf a)\bigr)
  =(\mathbf d,\mathbf c),
\]
where the last equality applies the Mal'tsev identity $\varphi(x,y,y)=x$ at every
tuple position.  Preservation therefore gives $(\mathbf d,\mathbf c)\in\RowEq F$. For every $i\in[\ell]$, entry $i$ of the first tuple is $\varphi(\alpha_i,\alpha_i,\beta_i)=\beta_i$. Thus the row indexed by $\mathbf d$ extends $\boldsymbol\beta$.
Membership in $\RowEq F$ also says that this row is
nonzero and belongs to the same row class $S_h$ as
$\mathbf c$.  Therefore $h\in\typeop_F(\boldsymbol\beta)$.  Since $h$ was
arbitrary, $\typeop_F(\boldsymbol\alpha) \subseteq\typeop_F(\boldsymbol\beta)$. Interchanging the prefixes proves the reverse inclusion.  Thus any two types
are equal or disjoint.  Empty types and the zero function are covered as well.
\end{proof}

Applying \cref{lem:omega-implies-type} with the common operation from
\cref{thm:common-support-maltsev,cor:omega-reduces-to-supports} gives the Type Partition
condition for every generated function.  The same operation satisfies the
Mal'tsev polymorphism condition, yielding the following implication.

\begin{theorem}[Structural collapse]
\label{thm:structural-collapse}
For every $q\ge1$ and every finite set $\cF$ of
algebraic complex-valued functions on $D$,
\[
  \BO(\cF)\quad\Longrightarrow\quad\TP(\cF)\wedge\Mal(\cF).
\]
More explicitly, under $\BO(\cF)$ there is one Mal'tsev operation
$\varphi:D^3\to D$ such that, for every $F\in W_{\cF}$,
\begin{equation*}
  \varphi\text{ preserves }\SuppRel F,
  \qquad \arity(F)\ge2\ \Longrightarrow\ \varphi\text{ preserves }\RowEq F.
\end{equation*}
Consequently,
\[
  \CC(\cF)\quad\Longleftrightarrow\quad\BO(\cF).
\]
\end{theorem}

\begin{proof}
The case $\cF=\varnothing$ was handled in \cref{subsec:types-relations}. Assume
$\cF\ne\varnothing$.
Assume $\BO(\cF)$.  By \cref{thm:common-support-maltsev}, there is a common
Mal'tsev polymorphism $\varphi:D^3\to D$ of the relations $\SuppRel H$,
$H\in W_{\cF}$.

By \cref{cor:omega-reduces-to-supports}, the same $\varphi$ preserves every
$\RowEq F$ with $F\in W_{\cF}$ of arity at least two.  It therefore preserves
all relations in $\Lambda_{\cF}$, proving $\Mal(\cF)$.

Applying \cref{lem:omega-implies-type} to every function in $W_{\cF}$ of
arity at least two proves $\TP(\cF)$.  The argument applies in every
ordering of variables. Relabelling the retained variables preserves membership
in $W_{\cF}$.

We have proved $\BO(\cF)\Rightarrow\TP(\cF)\wedge\Mal(\cF)$.  The reverse
implication in the final equivalence is immediate because $\BO(\cF)$ is a
conjunct of $\CC(\cF)$.
\end{proof}

\subsection{Proof of uniform decidability}

The structural implication identifies the full tractability criterion
with the Block Orthogonality condition.  We can therefore apply the
decision procedure from \cref{sec:bo-algorithm} to decide that criterion,
completing the proof of \cref{thm:main}.

\begin{proof}[Proof of \cref{thm:main}]
On input $(q,\cF)$, run \cref{proc:global-bo}.  By
\cref{thm:global-bo-decidable}, this is a total exact procedure and returns
\YES{} if and only if $\BO(\cF)$ holds.  By
\cref{thm:structural-collapse}, $\BO(\cF)\quad\Longleftrightarrow\quad\CC(\cF)$. Consequently, the same procedure halts on every input and returns
\YES{} exactly when $\CC(\cF)$ holds, proving both uniform
decidability and~\eqref{eq:main-equivalence}.
\end{proof}

\section{Extension to \texorpdfstring{$\#\mathrm{CSP}^{d}$}{degree-restricted counting CSP}}
\label{sec:degree-extension}

We now prove \cref{cor:degree-multiple-main}, using the definitions of
$\CSPdeg{d}{\cF}$, $\WFdeg{d}$, and the restricted tractability conditions
from \cref{subsec:constraint-instances,subsec:types-relations}.
There are two tasks.  First, we need an
identity test over the degree-restricted instances. A root-of-unity filter
will reduce it to the character comparisons from
\cref{sec:universal-identities}.  Second, we must verify that arity
compression and the support and row-equivalence constructions preserve
the degree restriction.  These two steps extend the decision procedure
and the structural implication to $\WFdeg{d}$.

\subsection{Universal identities on \texorpdfstring{$\WFdeg{d}$}{W(F, d)}}
\label{subsec:degree-multiple-identities}

Fix $k,d\ge1$.  Let $\mathfrak M_{k,d}$ be the submonoid of
$\InstMonoid{k}$ consisting of the labelled instances in which every
variable has occurrence degree divisible by $d$.  It is a submonoid:
gluing labelled instances adds the degrees of identified variables,
and the constraint-free labelled unit has degree zero at every variable.
The same labelling and summation argument as in
\cref{lem:pinned-realization} therefore gives
\begin{equation*}
  \{\zvec I:I\in\mathfrak M_{k,d}\}
  =\{F\in\WFdeg{d}:\arity(F)=k\}.
\end{equation*}

To enforce the restriction, we will use an auxiliary sum that is zero
whenever some variable has degree not divisible by $d$.  The following
identity supplies this test for each variable.  Write
$C_d=\mathbb Z/d\mathbb Z$ and choose a primitive
$d$th root of unity $\omega_d$ (with $\omega_1=1$).  For fixed
$n\in\mathbb Z$, the map $s\mapsto\omega_d^{sn}$ is a character of the
additive group $C_d$.  It is the constant-one character exactly when
$d\mid n$. Otherwise its values form a nontrivial geometric progression
whose sum is zero.  Thus
\begin{equation}
  \sum_{s\in C_d}\omega_d^{sn}
  =
  \begin{cases}
    d,&d\mid n,\\
    0,&d\nmid n.
  \end{cases}
  \label{eq:cyclic-character-filter}
\end{equation}

To compare monomial characters after restriction to
$\mathfrak M_{k,d}$, we use a root-of-unity filter that reduces the
problem to finitely many ordinary character comparisons on
$\InstMonoid{k}$.

\begin{lemma}
\label{lem:degree-filtered-character-comparison}
Fix $k,d\ge1$.  Given two monomials $M,N$ in the formal variables
$\{X_{\mathbf a},Y_{\mathbf a}:\mathbf a\in D^k\}$ and their associated
characters $\chi_M,\chi_N:\InstMonoid{k}\to K$, it is decidable whether
\begin{equation}
  \chi_M(I)=\chi_N(I)
  \qquad\text{for every }I\in\mathfrak M_{k,d}.
  \label{eq:degree-filtered-character-goal}
\end{equation}
\end{lemma}

\begin{proof}
We construct two finite sums of characters on $\InstMonoid{k}$.
At each $I\in\mathfrak M_{k,d}$, their values will be $\chi_M(I)$ and
$\chi_N(I)$ multiplied by the same nonzero factor.  Both sums will be zero
on the remaining instances.  Their equality on the full monoid will then
decide the desired equality on the restricted instances.

Use the tensorized families
$\mathcal H_M=(h_f:f\in\cF)$ and $\mathcal H_N=(g_f:f\in\cF)$
constructed from $M$ and $N$ in \cref{subsec:tensor-comparison}, on the
domains $D_M$ and $D_N$.  Take the respective pins
$\beta_M:[k]\to D_M$ and $\beta_N:[k]\to D_N$ from
\cref{def:tensor-realization}.  For an
original labelled $\cF$-instance $I\in\InstMonoid{k}$, let $I^M$
and $I^N$ denote its interpretations by these two families.  By
\cref{lem:tensor-character},
\begin{equation}
  Z_{\beta_M}^{\mathcal H_M}(I^M)=\chi_M(I),
  \qquad
  Z_{\beta_N}^{\mathcal H_N}(I^N)=\chi_N(I).
  \label{eq:degree-filter-monomial-realizations}
\end{equation}

Attach an auxiliary value $s\in C_d$ to each variable.  Each occurrence
of that variable will contribute a factor $\omega_d^s$, so multiplying
over constraints will record its occurrence degree in the exponent.
Let $K_d=K(\omega_d)$.  For an input function $f\in\cF$ of arity $r_f$,
define the everywhere-nonzero function
\begin{equation*}
  u_f(s_1,\ldots,s_{r_f})
  =\omega_d^{s_1+\cdots+s_{r_f}}.
\end{equation*}
Augment $h_f$ and $g_f$ to functions over
$D_M\times C_d$ and $D_N\times C_d$, respectively, by
\begin{equation*}
  h_f^{(d)}\bigl((a_1,s_1),\ldots,(a_{r_f},s_{r_f})\bigr)
  =h_f(a_1,\ldots,a_{r_f})u_f(s_1,\ldots,s_{r_f}).
\end{equation*}
Define $g_f^{(d)}$ by the same formula with $g_f$ in place of $h_f$.
The superscript $(d)$ records the added cyclic factor.
For $\mathbf s=(s_1,\ldots,s_k)\in C_d^k$, interpret an original
instance $I$ using the $h_f^{(d)}$-family and pin label $t$ to
$(\beta_M(t),s_t)$.  Denote the resulting pinned value by $\rho^M_{\mathbf s}(I)$. Define $\rho^N_{\mathbf s}(I)$ analogously using the $g_f^{(d)}$-family
and the pins $(\beta_N(t),s_t)$.  Gluing labelled instances multiplies
these pinned partition values, and the constraint-free labelled unit has
value one.  Thus every $\rho^M_{\mathbf s},\rho^N_{\mathbf s}:\InstMonoid{k}\to K_d$
is a unital multiplicative character.

We now compute the effect of summing over the auxiliary values.
Fix $I\in\InstMonoid{k}$.  For an assignment
$\theta:V(I)\to C_d$ to the cyclic factor, the product of all
auxiliary factors is
\begin{equation}
\begin{aligned}
  &\prod_{(f;v_1,\ldots,v_{r_f})\in\mathcal C_I}
    u_f\bigl(\theta(v_1),\ldots,\theta(v_{r_f})\bigr)\\
  &\quad={}
    \omega_d^{
      \sum_{(f;v_1,\ldots,v_{r_f})\in\mathcal C_I}
      \sum_{j\in[r_f]}\theta(v_j)}\\
  &\quad={}
    \omega_d^{
      \sum_{v\in V(I)}\deg_I(v)\theta(v)}.
  \label{eq:degree-filter-assignment-weight}
\end{aligned}
\end{equation}
Each occurrence of $v$ contributes $\theta(v)$ to the exponent, including
repeated occurrences within one scope.  Its total contribution is therefore
$\deg_I(v)\theta(v)$, where $\deg_I(v)$ is its occurrence degree.

Summing $\rho^M_{\mathbf s}(I)$ over all cyclic label pins
$\mathbf s\in C_d^k$ is equivalent to summing $\theta$ over all cyclic
assignments to $V(I)$.  The assignments to the tensorized domain and to the cyclic group
are independent, so
\cref{eq:degree-filter-monomial-realizations,eq:degree-filter-assignment-weight}
give the following identity, where $\mathbf 1_E$ is $1$ when condition
$E$ holds and $0$ otherwise:
\begin{equation}
\begin{aligned}
  \sum_{\mathbf s\in C_d^k}\rho^M_{\mathbf s}(I)
  &=\chi_M(I)
    \sum_{\theta:V(I)\to C_d}
      \omega_d^{\sum_v\deg_I(v)\theta(v)}\\
  &=\chi_M(I)
    \prod_{v\in V(I)}
      \left(
        \sum_{t\in C_d}
          \omega_d^{t\deg_I(v)}
      \right)\\
  &=d^{\abs{V(I)}}
    \mathbf 1_{\{I\in\mathfrak M_{k,d}\}}
    \chi_M(I).
  \label{eq:degree-filter-H}
\end{aligned}
\end{equation}
The last equality applies \eqref{eq:cyclic-character-filter} separately at
every variable.  The identical calculation for $N$ gives
\begin{equation}
  \sum_{\mathbf s\in C_d^k}\rho^N_{\mathbf s}(I)
  =d^{\abs{V(I)}}
    \mathbf 1_{\{I\in\mathfrak M_{k,d}\}}
    \chi_N(I).
  \label{eq:degree-filter-G}
\end{equation}

These two sums have the required behavior on and off
$\mathfrak M_{k,d}$.  Define them as functions on the full monoid
$\InstMonoid{k}$ by
\begin{equation}
  R_M=\sum_{\mathbf s\in C_d^k}\rho^M_{\mathbf s},
  \qquad
  R_N=\sum_{\mathbf s\in C_d^k}\rho^N_{\mathbf s}.
  \label{eq:degree-filter-character-sums}
\end{equation}
If $I\notin\mathfrak M_{k,d}$, both functions equal zero at $I$ by
\cref{eq:degree-filter-H,eq:degree-filter-G}.  If
$I\in\mathfrak M_{k,d}$, the same equations give
\[
  R_M(I)=d^{\abs{V(I)}}\chi_M(I),
  \qquad
  R_N(I)=d^{\abs{V(I)}}\chi_N(I).
\]
The common factor is nonzero because $K_d$ has characteristic zero.
Consequently,
\begin{equation}
  \chi_M=\chi_N\text{ on }\mathfrak M_{k,d}
  \quad\Longleftrightarrow\quad
  R_M=R_N\text{ on }\InstMonoid{k}.
  \label{eq:degree-filter-reduction}
\end{equation}

It remains to compare the two sums on $\InstMonoid{k}$.
There are $2d^k$ character terms in
\eqref{eq:degree-filter-character-sums}.  Apply
\cref{lem:character-comparison} to their augmented indexed families and
fixed pins.  At positive tensor degree, $T_{p_+,p_-}(f)\equiv0$ if and only
if $f\equiv0$. At degree zero it is the constant-one function.  Each
tensorized family therefore has at most one zero function of each arity.
Since $u_f$ is everywhere nonzero, augmenting the functions preserves
whether each is identically zero.  The hypotheses of the comparison lemma
hold, and it decides every pairwise equality.

Group these character terms by equality as functions on
$\InstMonoid{k}$, retaining their multiplicities in the two sums.  Let
$\mathcal E_{\mathrm{char}}$ be the resulting finite set of classes.
For each class $E$, count the cyclic pins contributing a character in $E$
on each side:
\begin{equation*}
\begin{aligned}
  m_M(E)&=\abs{\{\mathbf s\in C_d^k:
                    \rho^M_{\mathbf s}\in E\}},\\
  m_N(E)&=\abs{\{\mathbf s\in C_d^k:
                    \rho^N_{\mathbf s}\in E\}}.
\end{aligned}
\end{equation*}
Choose one representative character $\rho_E$ from each
$E\in\mathcal E_{\mathrm{char}}$.
Then, as functions on $\InstMonoid{k}$,
\begin{equation*}
  R_M-R_N
  =\sum_{E\in\mathcal E_{\mathrm{char}}}
     \bigl(m_M(E)-m_N(E)\bigr)\rho_E.
\end{equation*}
The representatives $\rho_E$, $E\in\mathcal E_{\mathrm{char}}$, are
pairwise distinct.  By
\cref{lem:dedekind-independence},
\begin{equation}
  R_M=R_N\quad\Longleftrightarrow\quad
  m_M(E)=m_N(E)\text{ for every }E\in\mathcal E_{\mathrm{char}}.
  \label{eq:degree-filter-acceptance-test}
\end{equation}
Combining
\cref{eq:degree-filter-reduction,eq:degree-filter-acceptance-test}
gives a finite exact decision procedure
for \eqref{eq:degree-filtered-character-goal}.
\end{proof}

We can now group the monomials of a polynomial according to equality of
their characters on $\mathfrak M_{k,d}$.  As in
\cref{thm:universal-identity}, linear independence reduces the identity
test on $\WFdeg{d}$ to testing the coefficient sum of each group.

\begin{theorem}[Universal identity oracle for $\WFdeg{d}$]
\label{thm:degree-multiple-universal-identity}
For every $k,d\ge1$ and every polynomial $P\in\AlgNums[X_{\mathbf a},Y_{\mathbf a}:\mathbf a\in D^k]$, given as a finite list of monomials with exactly encoded algebraic-complex
coefficients, it is decidable whether
\[
  P(\zvec I,\overline{\zvec I})=0
  \qquad\text{for every }I\in\mathfrak M_{k,d}.
\]
\end{theorem}

\begin{proof}
Each monomial of $P$ defines a unital multiplicative character on
$\mathfrak M_{k,d}$.  By \cref{lem:degree-filtered-character-comparison},
equality of any two such characters is decidable.  Group the finitely many
monomials according to this equality and sum their coefficients in the number
field generated by the input coefficients and $K$.  Applying
\cref{lem:dedekind-independence} to the monoid $\mathfrak M_{k,d}$ shows
that the identity
holds exactly when every class sum is zero.  All steps are finite and exact.
\end{proof}

\subsection{Proof of uniform decidability}
\label{subsec:degree-tractability}

The restricted identity oracle supplies the new algorithmic ingredient.
We now check that the remaining steps of the Block Orthogonality decision
procedure and the constructions establishing the other two conditions
stay within $\WFdeg{d}$.  In each instance construction, fresh internal
variables retain their degrees and identified retained variables receive
the sum of their degrees.

\begin{proof}[Proof of \cref{cor:degree-multiple-main}]
Fix $d\ge1$.  The proof of \cref{cor:singleton-purification} restricts
verbatim to $\WFdeg{d}$: $\BO^d(\cF)$ holds if and only if every
member of $\WFdeg{d}$ of arity at least two has a block-orthogonal
singleton purification.  At arity $k$, use the labelled submonoid
$\mathfrak M_{k,d}$ and the oracle in
\cref{thm:degree-multiple-universal-identity}.  The purification-invariance
and finite-certificate arguments concern one function on a finite domain and are unchanged.
The arity compression also remains valid: identifying retained
variables replaces their occurrence degrees by their sum, which is still
divisible by $d$.  Hence \cref{proc:global-bo}, using the
oracle for $\WFdeg{d}$, is a total exact decision procedure for
$\BO^d(\cF)$.

The structural proof likewise remains inside $\WFdeg{d}$.  More
explicitly, every instance substitution used in
\cref{lem:equality-free-support-realization,lem:generated-entrywise-powers,thm:omega-support-realization}
takes copies of $\CSPdeg{d}{\cF}$ instances, keeps their hidden
variables fresh, and identifies some retained variables.  A fresh variable retains a
degree divisible by $d$, while an identified variable receives a sum
of such degrees.  Labelled products, diagonal minors, variable
permutations, and repeated retained variables obey the same rule.  Thus all
functions realizing supports and row-equivalence detectors used in
\crefrange{sec:generated-supports}{sec:type-collapse} are members of
$\WFdeg{d}$.  Applying the arguments of
\cref{thm:common-support-maltsev,thm:structural-collapse} to these
restricted constructions gives
\[
  \BO^d(\cF)
  \Longrightarrow
  \TP^d(\cF)\wedge\Mal^d(\cF),
\]
with a common Mal'tsev polymorphism of all support and row-equivalence
relations associated with $\WFdeg{d}$.  Since $\BO^d(\cF)$ is a conjunct of
$\CC^d(\cF)$, this proves the claimed equivalence and uniform decidability
of $\CC^d(\cF)$.  Together with the dichotomy theorem
for $\#\mathrm{CSP}^{d}$ proved by Lin~\cite[Theorem~3]{Lin2021}, this
decides the polynomial-time versus $\#\mathsf P$-hard alternative for
$\CSPdeg{d}{\cF}$.
\end{proof}

\section{Conclusion}
\label{sec:conclusion}

We have given a total exact decision procedure for the three tractability
conditions characterizing the polynomial-time side of the complex-weighted
$\#\mathrm{CSP}$ dichotomy theorem of Cai and
Chen~\cite[Theorem~1 and Section~3.1]{CaiChen2017}.  The procedure decides
these conditions on $W_{\cF}$.
Its finite reduction combines the arity bound $q^2+1$ with universal
identity testing over all finite labelled instances.

The global Block Orthogonality condition also implies the Type Partition
and Mal'tsev polymorphism conditions, with one common Mal'tsev operation
preserving all generated support and row-equivalence relations.
The support realization and row-equivalence constructions stay within
$W_{\cF}$ and control cancellation so that their supports are exactly the
intended relations.

A cyclic root-of-unity degree filter and degree-preserving structural
constructions extend uniform decidability and structural collapse to the
degree-divisibility restriction $\#\mathrm{CSP}^{d}$ studied by
Lin~\cite[Definition~2 and Theorem~3]{Lin2021}.  Thus $\CC^d(\cF)$, the
tractability criterion in that dichotomy theorem, is uniformly decidable
for every positive degree modulus $d$.

\clearpage
\appendix
\renewcommand{\thesection}{F.\Alph{section}}
\renewcommand{\theHsection}{fullappendix.\Alph{section}}
\section{Examples}
\label{app:tractability-examples}

The following examples illustrate the tractability conditions in
\cref{sec:preliminaries}, the variable-identification constructions in
\cref{lem:diagonal-minors,lem:bo-arity-bound}, and the tensor
realization used in \cref{sec:universal-identities}.  The final two
examples illustrate support realization and the constraints defining
preserving operations in \cref{sec:generated-supports}.

\subsection{Block-rank one}
\label{app:block-rank-one-example}

Take $q=4$ and let the binary function $A$ be
\[
\begin{array}{c|rrrr}
  x\backslash z&0&1&2&3\\ \hline
  0&1&2&0&0\\
  1&3&6&0&0\\
  2&0&0&1&4\\
  3&0&0&2&8
\end{array}.
\]
For $\cF_A=\{A\}$, the instance consisting of the single
constraint $A(x,z)$ has $F_I^{[2]}=A$, so $A\in W_{\cF_A}$.  The first two
magnitude rows are proportional, as are the last two, while rows from
different pairs have disjoint supports.  Equivalently, the support is the
disjoint union of the rectangles $\{0,1\}\times\{0,1\} \quad\text{and}\quad \{2,3\}\times\{2,3\}$. Thus $A$ is block-rank one.  In this example every pair of
magnitude-proportional complex rows is already linearly dependent, so $A$ is block-orthogonal.

\subsection{Block Orthogonality and purification}
\label{app:block-orthogonality-example}

On the same domain, let
\[
B=
\begin{array}{c|rrrr}
  x\backslash z&0&1&2&3\\ \hline
  0&\sqrt3&\sqrt3&1/\sqrt3&1/\sqrt3\\
  1&\sqrt3&-\sqrt3&1/\sqrt3&-1/\sqrt3\\
  2&0&0&0&0\\
  3&0&0&0&0
\end{array}.
\]
Again $B\in W_{\{B\}}$ by the one-constraint instance.  Its two nonzero
rows $\mathbf u,\mathbf v$ are not linearly dependent, but their magnitude
vectors are both $(\sqrt3,\sqrt3,1/\sqrt3,1/\sqrt3)$.
Their equal-magnitude blocks are $T_1=\{0,1\}$ and $T_2=\{2,3\}$, and
cancellation holds separately on both blocks:
\[
  \sum_{z\in T_1}u_z\overline{v_z}=3-3=0,
  \qquad
  \sum_{z\in T_2}u_z\overline{v_z}=\tfrac13-\tfrac13=0.
\]
Thus $\mathbf u\perp_{\mathrm B}\mathbf v$ and $B$ is block-orthogonal.

For singleton purification, choose $g_1=\sqrt3$.  Its nonzero integer
powers are not roots of unity, and every nonzero value of $B$ has the form
$\varepsilon g_1^e$ with $\varepsilon\in\{1,-1\}$ and $e\in\{1,-1\}$.
The construction replaces $g_1$ by the first prime $p_1=2$, while keeping
the root-of-unity factor $\varepsilon$.  Thus
$\pm\sqrt3\mapsto\pm2$ and $\pm1/\sqrt3\mapsto\pm\tfrac12$.
The resulting function $\widetilde B$ has denominator-clearing factor
$N=2$, so its normalized purification is
\[
B'=2\widetilde B=
\begin{array}{c|rrrr}
  x\backslash z&0&1&2&3\\ \hline
  0&4&4&1&1\\
  1&4&-4&1&-1\\
  2&0&0&0&0\\
  3&0&0&0&0
\end{array}.
\]
The values have changed, while zeros, signs, and the two magnitude blocks
are preserved.  Within these blocks, the two nonzero rows of $B'$ give
$16-16=0$ and $1-1=0$, respectively.  Hence $B'$ is pure and
block-orthogonal.

\subsection{Type Partition}
\label{app:type-partition-example}

Take $q=3$, with addition on $D$ modulo three, and define
$F:D^3\to\AlgNums$ by
\[
\begin{array}{c|rrr}
  (x_1,x_2)\backslash z&0&1&2\\ \hline
  (0,0)&1&0&0\\
  (0,1)&0&1&0\\
  (0,2)&0&0&1\\
  (1,0)&0&2&0\\
  (1,1)&0&0&2\\
  (1,2)&2&0&0\\
  (2,0)&0&0&4\\
  (2,1)&4&0&0\\
  (2,2)&0&4&0
\end{array}.
\]
Equivalently,
\begin{equation}
  F(x_1,x_2,z)\ne0
  \quad\Longleftrightarrow\quad
  z\equiv x_1+x_2\pmod 3,
  \label{eq:type-example-modular-support}
\end{equation}
and the unique nonzero entry in a row has value $2^{x_1}$.
The one-constraint instance gives $F\in W_{\{F\}}$.  Its three nonzero
complex row classes are
\[
\begin{aligned}
  S_1&=\{(0,0),(1,2),(2,1)\},& \mathbf v_1&=(1,0,0),\\
  S_2&=\{(0,1),(1,0),(2,2)\},& \mathbf v_2&=(0,1,0),\\
  S_3&=\{(0,2),(1,1),(2,0)\},& \mathbf v_3&=(0,0,1).
\end{aligned}
\]
For prefixes of length one, $\typeop_F(0)=\typeop_F(1)=\typeop_F(2)=[3]$. Thus one type contains three row-class labels and occurs at all three
prefixes.  For prefixes of length two, if $a+b\equiv j\pmod3$ with
$j\in\{0,1,2\}$, then $\typeop_F(a,b)=\{j+1\}$. Hence, at each fixed prefix length, any two types are equal or disjoint.
The same function illustrates the Mal'tsev polymorphism condition in
\cref{app:maltsev-example}.

\subsection{A Mal'tsev polymorphism}
\label{app:maltsev-example}

Continue with the same domain $D$, with addition modulo three and the same function $F$
from Appendix~\ref{app:type-partition-example}.  Consider the affine operation $\varphi(a,b,c)=a-b+c\pmod3$. It is Mal'tsev because
\[
  \varphi(a,b,b)=a-b+b=a,
  \qquad
  \varphi(b,b,a)=b-b+a=a.
\]
By~\eqref{eq:type-example-modular-support}, the support relation is the
nontrivial affine relation $\SuppRel F =\{(x_1,x_2,z)\in D^3:z\equiv x_1+x_2\pmod3\}$. Take any three tuples
$\mathbf u^{(i)}=(x_1^{(i)},x_2^{(i)},z^{(i)})\in\SuppRel F$,
$i\in[3]$, and let
$\mathbf u=(u_1,u_2,u_3)$ be their coordinatewise image under $\varphi$.  Then
\[
\begin{aligned}
  u_3
  &\equiv z^{(1)}-z^{(2)}+z^{(3)}\\
  &\equiv
    \bigl(x_1^{(1)}-x_1^{(2)}+x_1^{(3)}\bigr)
    +\bigl(x_2^{(1)}-x_2^{(2)}+x_2^{(3)}\bigr)\\
  &\equiv u_1+u_2\pmod3.
\end{aligned}
\]
Thus $\mathbf u\in\SuppRel F$, so $\varphi$ preserves the support.

The row classes in \cref{app:type-partition-example} also give
\[
  \RowEq F
  =
  \left\{
  \begin{array}{l}
    (x_1,x_2,y_1,y_2)\in D^4:\\
    x_1+x_2\equiv y_1+y_2\pmod3
  \end{array}
  \right\}.
\]
Indeed, two rows are proportional exactly when their unique nonzero entries
occur in the same column.  The same linearity calculation shows that
$\varphi$ preserves the equation $x_1+x_2-y_1-y_2=0$ modulo three, and hence
$\RowEq F$.  Thus this one operation preserves both relations.

\subsection{Variable identification}
\label{app:variable-identification-example}

\Cref{fig:diagonal-minor-example} illustrates the construction in
\cref{lem:diagonal-minors}.
On a finite domain $D$, consider a $\CSP{\cF}$ instance $I$ with retained
variables $x_1,\ldots,x_5$ and hidden variables $z_1,z_2$.
Set $F=F_I^{[5]}$ and take the constraint multiset
\begin{equation*}
\begin{aligned}
  \mathcal C_I=\{&f(x_1,z_1,x_3),
                  g(z_1,x_2,x_5),\\
                 &h(x_4,z_2,x_3),
                  \ell(z_2,x_5)\}.
\end{aligned}
\end{equation*}
For the map $\pi:[5]\twoheadrightarrow[3]$ given by
$\pi=(1,2,1,3,2)$, replace each $x_i$ by $y_{\pi(i)}$.
The resulting instance $I_\pi$ has constraint multiset
\[
\begin{aligned}
  \mathcal C_{I_\pi}=\{&f(y_1,z_1,y_1),
                        g(z_1,y_2,y_2),\\
                       &h(y_3,z_2,y_1),
                        \ell(z_2,y_2)\}.
\end{aligned}
\]

\begin{figure}[!htbp]
  \centering
  \resizebox{0.94\linewidth}{!}{%
  \begin{tikzpicture}[
      x=1cm,y=1cm,
      var/.style={circle,draw,minimum size=5.8mm,inner sep=0pt,
                  font=\scriptsize},
      hidden/.style={var,dashed},
      factor/.style={draw,rounded corners=1pt,minimum width=6mm,
                     minimum height=5.5mm,inner sep=1pt,
                     double,double distance=.45pt,font=\scriptsize},
      edge/.style={line width=.42pt},
      title/.style={font=\small\bfseries},
      scopes/.style={font=\scriptsize,align=center}
    ]
    \begin{scope}[shift={(2.1,0)}]
      \node[title] at (1.2,2.2) {original instance $I$};
      \node[factor] (f) at (0,1) {$f$};
      \node[factor] (g) at (2.4,1) {$g$};
      \node[factor] (h) at (0,-1) {$h$};
      \node[factor] (l) at (2.4,-1) {$\ell$};

      \node[hidden] (z1) at (1.2,1) {$z_1$};
      \node[hidden] (z2) at (1.2,-1) {$z_2$};
      \node[var] (x3) at (0,0) {$x_3$};
      \node[var] (x5) at (2.4,0) {$x_5$};
      \node[var] (x1) at (-1.15,1.62) {$x_1$};
      \node[var] (x2) at (3.55,1.62) {$x_2$};
      \node[var] (x4) at (-1.15,-1.62) {$x_4$};

      \draw[edge] (f)--(x1) (f)--(z1) (f)--(x3);
      \draw[edge] (g)--(z1) (g)--(x2) (g)--(x5);
      \draw[edge] (h)--(x4) (h)--(z2) (h)--(x3);
      \draw[edge] (l)--(z2) (l)--(x5);

      \node[scopes] at (1.2,-2.25)
        {$f(x_1,z_1,x_3),\ g(z_1,x_2,x_5)$\\
         $h(x_4,z_2,x_3),\ \ell(z_2,x_5)$};
    \end{scope}

    \draw[-{Latex[length=2mm]},line width=.55pt]
      (6.05,0)--(7.65,0)
      node[midway,above=2pt,font=\scriptsize,align=center]
      {$x_1,x_3\mapsto y_1$\\$x_2,x_5\mapsto y_2$\\$x_4\mapsto y_3$};

    \begin{scope}[shift={(9.5,0)}]
      \node[title] at (1.2,2.2) {constructed instance $I_\pi$};
      \node[factor] (fp) at (0,1) {$f$};
      \node[factor] (gp) at (2.4,1) {$g$};
      \node[factor] (hp) at (0,-1) {$h$};
      \node[factor] (lp) at (2.4,-1) {$\ell$};

      \node[hidden] (z1p) at (1.2,1) {$z_1$};
      \node[hidden] (z2p) at (1.2,-1) {$z_2$};
      \node[var] (y1) at (0,0) {$y_1$};
      \node[var] (y2) at (2.4,0) {$y_2$};
      \node[var] (y3) at (-1.15,-1.62) {$y_3$};

      \draw[edge] (fp)--(z1p);
      \draw[edge] (gp)--(z1p);
      \draw[edge] (hp)--(z2p);
      \draw[edge] (lp)--(z2p);
      \draw[edge] (fp) to[bend left=17] (y1);
      \draw[edge] (fp) to[bend right=17] (y1);
      \draw[edge] (hp)--(y1);
      \draw[edge] (gp) to[bend left=17] (y2);
      \draw[edge] (gp) to[bend right=17] (y2);
      \draw[edge] (lp)--(y2);
      \draw[edge] (hp)--(y3);

      \node[scopes] at (1.2,-2.25)
        {$f(y_1,z_1,y_1),\ g(z_1,y_2,y_2)$\\
         $h(y_3,z_2,y_1),\ \ell(z_2,y_2)$};
    \end{scope}
  \end{tikzpicture}%
  }
  \caption{Construction in \cref{lem:diagonal-minors}.  Circles are
  variables, boxes are constraints, and dashed circles are hidden variables.
  Constraint nodes and hidden variables are unchanged. Only incidences of
  retained variables are relabelled.  The parallel edges at $f$ and $g$ in
  $I_\pi$ record repeated occurrences created by identification.  Summing
  over the unchanged $z_1,z_2$ gives
  $F_{I_\pi}^{[3]}(y_1,y_2,y_3)
   =F(y_1,y_2,y_1,y_3,y_2)=H(y_1,y_2,y_3)$.}
  \label{fig:diagonal-minor-example}
\end{figure}

\subsection{Arity compression}
\label{app:arity-compression-example}

To illustrate the identification step in~\cref{lem:bo-arity-bound}, let
$q=2$ and let $F\in W_{\cF}$ have arity nine.  Choose the two row indices
\[
\begin{aligned}
  \mathbf x&=(0,0,1,0,1,1,0,1),\\
  \mathbf y&=(0,1,0,0,1,0,1,1).
\end{aligned}
\]
\Cref{tab:pair-class-example} lists the four classes of positions
grouped by $(x_i,y_i)$.  Identify the variables within each class and
keep the column variable $z$ separate.

\begin{table}[!htbp]
  \centering
  \small
  \setlength{\tabcolsep}{3.6pt}
  \renewcommand{\arraystretch}{1.18}
  \begin{tabular}{c|cccc}
    class
      &$C_{00}\,(\bullet)$
      &$C_{01}\,(\circ)$
      &$C_{10}\,(\blacktriangle)$
      &$C_{11}\,(\diamond)$\\ \hline
    positions &$\{1,4\}$&$\{2,7\}$&$\{3,6\}$&$\{5,8\}$\\
    new variable &$u_{00}$&$u_{01}$&$u_{10}$&$u_{11}$\\
    $\mathbf x$-value &$0$&$0$&$1$&$1$\\
    $\mathbf y$-value &$0$&$1$&$0$&$1$
  \end{tabular}
  \caption{The position classes for the displayed Boolean row-index tuples.  Each
  black-and-white marker identifies one class of positions carrying the same
  ordered pair $(x_i,y_i)$.}
  \label{tab:pair-class-example}
\end{table}

In this example the function obtained by identifying variables is
\[
\begin{aligned}
  H(u_{00},u_{01},u_{10},u_{11},z)
   =F(&u_{00},u_{01},u_{10},u_{00},\\[-2pt]
      &u_{11},u_{10},u_{01},u_{11},z).
\end{aligned}
\]
Assigning $(u_{00},u_{01},u_{10},u_{11})$ the values $(0,0,1,1)$ or
$(0,1,0,1)$ recovers $\mathbf x$ or $\mathbf y$, respectively.
There are four identified row variables and one separate column
variable, so the resulting arity is $4+1=q^2+1=5$.

\subsection{Tensor realization}
\label{app:tensor-realization-example}

Take $q=2$, $k=2$, $p_+=2$, $p_-=1$, and $M=X_{(0,1)}X_{(1,1)}Y_{(1,0)}$. The three factors fix the input tuples $(0,1),(1,1),(1,0)$.
Transposing these tuples as in \cref{def:tensor-realization} gives $\beta_M(1)=(0,1,1), \quad \beta_M(2)=(1,1,0)$.

\Cref{fig:tensor-character-layers} visualizes this example and the
construction in \cref{lem:tensor-character}.

\begin{figure}[!t]
  \centering
  \resizebox{0.80\linewidth}{!}{%
  \begin{tikzpicture}[
      x=1cm,y=1.04cm,
      title/.style={font=\small\bfseries},
      panel/.style={draw,rounded corners=2pt,line width=.45pt},
      cell/.style={draw,minimum width=7mm,minimum height=5.5mm,
                   inner sep=0pt,font=\scriptsize},
      conjcell/.style={cell,double,double distance=.6pt},
      rowlabel/.style={font=\scriptsize,anchor=east},
      inst/.style={draw,rounded corners=3pt,minimum width=4.5cm,
                   minimum height=1.55cm,align=center,font=\scriptsize},
      pin/.style={draw,circle,minimum size=6.5mm,inner sep=1pt,
                  font=\scriptsize},
      layer/.style={draw,rounded corners=2pt,minimum width=4.35cm,
                    minimum height=6.5mm,align=left,font=\scriptsize},
      conjlayer/.style={layer,double,double distance=.6pt},
      arrow/.style={-{Latex[length=2mm]},line width=.55pt}
    ]
    % Panel (a): rows indexing the factors of the monomial.
    \draw[panel] (.35,7.85) rectangle (13.65,11.40);
    \node[title] at (7.00,11.05) {(a) factors of $M$ as rows};
    \node[font=\scriptsize] at (7.15,10.62) {$t=1$};
    \node[font=\scriptsize] at (7.95,10.62) {$t=2$};
    \node[rowlabel] at (6.68,10.19) {$X_{\mathbf a^{(1)}}$};
    \node[cell] at (7.15,10.19) {$0$};
    \node[cell] at (7.95,10.19) {$1$};
    \node[rowlabel] at (6.68,9.49) {$X_{\mathbf a^{(2)}}$};
    \node[cell] at (7.15,9.49) {$1$};
    \node[cell] at (7.95,9.49) {$1$};
    \node[rowlabel] at (6.68,8.79) {$Y_{\mathbf b^{(1)}}$};
    \node[conjcell] at (7.15,8.79) {$1$};
    \node[conjcell] at (7.95,8.79) {$0$};
    \node[font=\scriptsize,align=center] at (7.00,8.23)
      {single border: ordinary factor, double border: conjugated factor};

    \draw[arrow] (7.00,7.68)--(7.00,7.05)
      node[midway,right=3pt,font=\scriptsize]
      {transpose rows into $D^3$-valued columns};

    % Panel (b): bundle the columns into D^3-valued pins.
    \draw[panel] (.35,3.50) rectangle (13.65,6.85);
    \node[title] at (7.00,6.50) {(b) transpose and bundle columns};
    \node[inst] (IM) at (7.00,5.55)
      {$I^M$ has the same variables, labels, and scopes as $I$\\[2pt]
       every occurrence $f$ is replaced by $T_{2,1}(f)$};
    \node[pin] (pinone) at (5.55,4.15) {$1$};
    \node[pin] (pintwo) at (8.45,4.15) {$2$};
    \draw[line width=.45pt] (pinone.north |- IM.south)--(pinone.north);
    \draw[line width=.45pt] (pintwo.north |- IM.south)--(pintwo.north);
    \node[font=\scriptsize,align=center] at (3.55,4.15)
      {$\beta_M(1)=(0,1,1)$};
    \node[font=\scriptsize,align=center] at (10.45,4.15)
      {$\beta_M(2)=(1,1,0)$};

    \draw[arrow] (7.00,3.32)--(7.00,2.89)
      node[midway,right=3pt,font=\scriptsize]
      {split each $D^3$-assignment into its three components};

    % Panel (c): component layers of a D^3-valued assignment.
    \draw[panel] (.35,-.85) rectangle (13.65,2.68);
    \node[title] at (7.00,2.33) {(c) independent component layers};
    \node[layer,minimum width=8.0cm] (L1) at (7.00,1.75)
      {layer $1$: pin $(0,1)$, value $\zvec I(0,1)$};
    \node[layer,minimum width=8.0cm] (L2) at (7.00,.92)
      {layer $2$: pin $(1,1)$, value $\zvec I(1,1)$};
    \node[conjlayer,minimum width=8.0cm] (L3) at (7.00,.09)
      {layer $3$: pin $(1,0)$, value $\overline{\zvec I(1,0)}$};
    \node[font=\scriptsize,align=center] at (7.00,-.49)
      {multiply the three layer sums to obtain $\chi_M(I)$};
  \end{tikzpicture}%
  }
  \caption{Tensor realization in \cref{def:tensor-realization} for
  $M=X_{(0,1)}X_{(1,1)}Y_{(1,0)}$.  The rows indexing the three monomial factors
  are transposed into one $D^3$-valued pin at each label.  An assignment to
  $I^M$ is then exactly three independent $D$-valued assignments to the same
  instance $I$. Double borders mark the factor and layer that use
  conjugated constraint weights.
  Hence its pinned partition value factors as
  $\zvec I(0,1)\zvec I(1,1)\overline{\zvec I(1,0)}=\chi_M(I)$.}
  \label{fig:tensor-character-layers}
\end{figure}

\subsection{Support realization}
\label{app:support-realization-example}
In \cref{lem:equality-free-support-realization}, the relation $R$ consists of assignments to the free variables that extend
to an assignment of the existential variables making every selected
function nonzero.

For example, if $F_1,F_2\in W_{\cF}$ are binary, then
\[
  R(x_1,x_2)
  \quad\Longleftrightarrow\quad
  \exists y\;
  \SuppRel{F_1}(x_1,y)\wedge\SuppRel{F_2}(y,x_2).
\]
Both atoms use the same existential variable $y$.

\Cref{fig:equality-free-support-construction} shows how this binary example
is implemented by one $\CSP{\cF}$ instance.

\begin{figure}[!t]
  \centering
  \resizebox{0.94\linewidth}{!}{%
  \begin{tikzpicture}[
      x=1cm,y=1cm,
      title/.style={font=\small\bfseries},
      panel/.style={draw,rounded corners=2pt,line width=.45pt},
      var/.style={circle,draw,minimum size=6.5mm,inner sep=0pt,
                  font=\scriptsize},
      hidden/.style={var,dashed},
      factor/.style={draw,double,double distance=.45pt,rounded corners=2pt,
                     minimum width=9mm,minimum height=6.5mm,
                     font=\scriptsize},
      module/.style={draw,rounded corners=2pt,minimum width=2.2cm,
                     minimum height=6.5mm,line width=.65pt,
                     font=\scriptsize,align=center},
      edge/.style={line width=.42pt},
      arrow/.style={-{Latex[length=2mm]},line width=.55pt}
    ]
    % Panel (a): the support pattern to be implemented.
    \draw[panel] (-.20,5.95) rectangle (14.20,9.25);
    \node[title] at (7.00,8.90) {(a) equality-free support pattern};
    \node[var] (sx1) at (2.30,7.65) {$x_1$};
    \node[factor] (sG1) at (4.40,7.65) {$F_1$};
    \node[hidden] (sy) at (7.00,7.65) {$y$};
    \node[factor] (sG2) at (9.60,7.65) {$F_2$};
    \node[var] (sx2) at (11.70,7.65) {$x_2$};
    \draw[edge] (sx1)--(sG1)--(sy)--(sG2)--(sx2);
    \node[font=\scriptsize,align=center] at (7.00,6.88)
      {$R(x_1,x_2)\iff\exists y\;$
       $\SuppRel{F_1}(x_1,y)\wedge\SuppRel{F_2}(y,x_2)$};
    \node[font=\scriptsize,align=center] at (7.00,6.30)
      {the dashed node is the shared existential variable $y$};

    \draw[arrow] (7.00,5.78)--(7.00,5.15)
      node[midway,right=3pt,font=\scriptsize,align=left]
      {replace each support atom by $m$ fresh copies\\
       of a fixed $\cF$-instance};

    % Panel (b): m copies of instances for F_1 and F_2.
    \draw[panel] (-.20,-.55) rectangle (14.20,4.95);
    \node[title] at (7.00,4.60)
      {(b) one $\cF$-instance with $F_I^{[2]}=H_m$};

    \node[module] (I11) at (5.00,3.90) {$I_1^{(1)}\;\bullet$};
    \node[module] (I12) at (5.00,3.05) {$I_1^{(2)}\;\bullet$};
    \node[font=\scriptsize] at (5.00,2.51) {$\vdots$};
    \node[module] (I1m) at (5.00,1.93) {$I_1^{(m)}\;\bullet$};

    \node[module] (I21) at (9.00,3.90) {$I_2^{(1)}\;\bullet$};
    \node[module] (I22) at (9.00,3.05) {$I_2^{(2)}\;\bullet$};
    \node[font=\scriptsize] at (9.00,2.51) {$\vdots$};
    \node[module] (I2m) at (9.00,1.93) {$I_2^{(m)}\;\bullet$};

    \node[var] (x1) at (2.15,1.13) {$x_1$};
    \node[hidden] (y) at (7.00,1.13) {$y$};
    \node[var] (x2) at (11.85,1.13) {$x_2$};

    \foreach \u in {I11,I12,I1m}{
      \draw[edge] (x1)--(\u.west);
      \draw[edge] (y)--(\u.east);
    }
    \foreach \u in {I21,I22,I2m}{
      \draw[edge] (y)--(\u.west);
      \draw[edge] (x2)--(\u.east);
    }

    \node[font=\scriptsize,align=center] at (7.00,-.07)
      {solid boxes are fresh copies of fixed $\cF$-instances.
       Each $\bullet$ denotes fresh internal variables\\
       retain $x_1,x_2$ and sum over the shared dashed $y$ and all
       $\bullet$-variables};
  \end{tikzpicture}%
  }
  \caption{The instance construction in
  \cref{lem:equality-free-support-realization} for the binary example.
  For each atom, panel (b) schematically shows $m$ fresh copies of a
  $\cF$-instance representing the corresponding function, indexed by $[m]$.
  Identify their retained
  variables with the logical variables $x_1,y,x_2$. All internal
  variables, denoted by $\bullet$, remain fresh.  Summing out the internal
  variables first gives
  $F_1(x_1,y)^mF_2(y,x_2)^m$, and then summing over the one shared variable
  $y$ gives $H_m(x_1,x_2)$.}
  \label{fig:equality-free-support-construction}
\end{figure}

\subsection{Preservation constraints}
\label{app:preservation-constraints-example}

In the proof of \cref{lem:finite-support-maltsev}, the variable
$x_{(a,b,c)}$ represents the output of the candidate operation at
$(a,b,c)$.  A partial operation $p$ assigns these outputs.
For example, let $q=2$ and suppose a binary relation $R$ contains
$\mathbf u^1=(0,0)$, $\mathbf u^2=(0,1)$, and $\mathbf u^3=(1,1)$.
The preservation constraint for these three tuples is $R\bigl(x_{(0,0,1)},x_{(0,1,1)}\bigr)$, which requires
$(p(0,0,1),p(0,1,1))\in R$ under the assignment $p$.

\clearpage
\section{Why all labelled instances are needed}
\label{app:all-instance-isomorphism}

Young's main theorem~\cite[Theorem~8]{Young2025} characterizes
constraint-function isomorphism through partition values on all instances.
In the finite-family setting used here, let $\mathcal H$ and $\mathcal G$
be finite families of constraint functions on finite domains
$D_{\mathcal H}$ and $D_{\mathcal G}$, with values in a common field of
characteristic zero and an arity-preserving bijection between the families.
The theorem states that there is a single domain bijection preserving
every pair of corresponding functions if and only if
\[
  Z^{\mathcal H}(I)
  =Z^{\mathcal G}(I_{\mathcal H\to\mathcal G})
\]
for every finite unlabelled $\mathcal H$-instance $I$.
Here $Z^{\mathcal H}(I)$ is the unpinned partition value, and
$I_{\mathcal H\to\mathcal G}$ replaces each constraint function by its
corresponding function while retaining the variables and scopes.

Young's Corollary~35 asserts a stronger statement involving pins and a
restricted class of instances.  Following
\cite[Definition~4, p.~6]{Young2025}, a labelled $\mathcal H$-instance is
\emph{simple} if the variables in each constraint scope are pairwise
distinct and, for every fixed function $h\in\mathcal H$, no two
occurrences of $h$ have scope tuples that are permutations of one another.

Fix $k\ge0$ and pins $\beta:[k]\to D_{\mathcal H}$ and
$\gamma:[k]\to D_{\mathcal G}$.
Corollary~35~\cite[p.~27]{Young2025} asserts that if
\[
  Z_{\beta}^{\mathcal H}(I)
  =Z_{\gamma}^{\mathcal G}(I_{\mathcal H\to\mathcal G})
\]
for every finite simple $k$-labelled $\mathcal H$-instance $I$, then there
is a domain isomorphism $\sigma:D_{\mathcal H}\to D_{\mathcal G}$ such
that $\sigma(\beta(i))$ and $\gamma(i)$ are twins in $\mathcal G$ for
every $i\in[k]$.  Recall that twin domain elements can replace one
another in any input position without changing any constraint-function
value.

\Cref{thm:young-pinned} establishes this pinned conclusion under equality
on all finite labelled instances.  Its interpolation proof uses repeated
constraints to form the powers in its indicator polynomials, and its
single-constraint tests use repeated scope positions.
The counterexample below shows that restricting the premise to simple
instances is insufficient.

\begin{remark}[Simple-instance indistinguishability need not imply isomorphism]
\label{rem:simple-instance-counterexample}
The counterexample uses singleton sets of strictly positive
rational-valued unary functions with no distinct twin domain elements.
It also refutes the simple-instance formulation of
\cref{thm:young-pinned} in the first version of this paper.

Take $q=3$, $\mathcal H=\{f\}$, and $\mathcal G=\{g\}$, where
\[
  (f(0),f(1),f(2))=(1,2,6),\qquad
  (g(0),g(1),g(2))=(1,3,5).
\]
The families are similar under $f\leftrightarrow g$.  Take one label, with
pins $\beta,\gamma$ both assigning its variable $u$ the value $0$.  Every simple
$\mathcal H$-instance $I$ places at most one copy of $f$ on each variable.
If $r$ hidden variables carry a constraint and $s$ hidden variables are
isolated, independence of the hidden-variable sums gives
\[
  Z_{\beta}^{\mathcal H}(I)
  =f(0)^{\varepsilon}\Bigl(\sum_{a\in D}f(a)\Bigr)^r q^s
  =9^r3^s
  =Z_{\gamma}^{\mathcal G}(I_{\mathcal H\to\mathcal G}),
\]
where $\varepsilon\in\{0,1\}$ records whether $u$ carries a constraint.
Since $f(0)=g(0)=1$ and both unary sums equal $9$, the equality holds for
every simple instance, including empty constraint sets and either value
of $\varepsilon$.  No domain bijection $\sigma$ is an
isomorphism: it would require $2=f(1)=g(\sigma(1))$, whereas
$2\notin\{1,3,5\}$.  Both functions are injective, so allowing the prescribed
pins to match only up to twins does not change the conclusion.

The obstruction occurs explicitly in the proof of
\cite[Lemma~29, p.~20]{Young2025}, which is used in the unary case of
Theorem~34 and hence in Corollary~35.  That proof places $p_j$ copies of
the same unary constraint $(F_j,v)$ on one hidden variable $v$ and treats
the resulting instance as simple.  For $p_j\ge2$ it violates Definition~4,
so equality on simple instances does not supply the power-sum identities
needed for the interpolation.  Indeed, the unrestricted instance with two
copies of the unary constraint on a single hidden variable, and an unused
pinned label, has respective values
\[
  \sum_{a\in D}f(a)^2=41,\qquad
  \sum_{a\in D}g(a)^2=35.
\]
The two-copy instance therefore distinguishes these families in the
all-instance setting of \cite[Theorem~8]{Young2025} and
\cref{thm:young-pinned}.
\end{remark}

% END FULL PAPER
\end{document}